\documentclass[11pt,A4,twoside]{article}

\usepackage{enumitem}
\usepackage[utf8]{inputenc}
\usepackage[T1]{fontenc}
\usepackage[english]{babel}
\usepackage{acronym}
\usepackage{algorithm}
\usepackage{algpseudocode}
\usepackage{amsfonts}
\usepackage{amsmath}
\usepackage{amsthm}
\usepackage{amssymb}
\usepackage{stmaryrd}
\usepackage{array}
\usepackage{caption}
\usepackage{dirtytalk}
\usepackage{fancyhdr}
\usepackage{float}
\usepackage{placeins}
\usepackage[a4paper, margin=1in]{geometry} % Marges ajustées
\usepackage{graphicx}
\usepackage{indentfirst}
\usepackage{lmodern}
\usepackage{longtable}
\usepackage{makecell}
\usepackage{mathrsfs}
\usepackage{microtype}
\usepackage[nottoc]{tocbibind}
\usepackage{siunitx}
\usepackage{textpos}
\usepackage{tikz, forest}
\usetikzlibrary{shapes.geometric, arrows.meta, positioning}
\usepackage{titlesec}
\usepackage{booktabs}
\usepackage[authoryear,round]{natbib}
\usepackage{bbm}
\usepackage{mdframed}

\usepackage[normalem]{ulem}
\titleformat{\section}
  {\normalfont\Large\bfseries} % format
  {\thesection} % label
  {1em} % espacement entre le label et le titre
  {} % code avant le titre
  [] % code après le titre (supprime les traits horizontaux)

\titleformat{name=\section,numberless}
  {\normalfont\Large\bfseries} % format pour les sections non numérotées
  {} % pas de label
  {0pt} % espacement
  {} % avant
  [] % après

\numberwithin{equation}{section}

\newtheorem{theorem}{Theorem}[section]
\newtheorem{corollary}[theorem]{Corollary}

\newtheorem{lemma}[theorem]{Lemma}
\newtheorem{proposition}[theorem]{Proposition}
\newtheorem{definition}{Definition}[section]

\newtheorem{remark}[definition]{Remark}

\newcommand{\Tr}{\text{Tr}}
\newcommand{\id}{\text{id}}

\newcommand{\Input}{\State \textbf{Input:} }
\newcommand{\Output}{\State \textbf{Output:} }

\newcommand{\bo}[1]{\boldsymbol{#1}}
\newcommand{\mc}[1]{\mathcal{#1}} % Utiliser \mathscr au lieu de \mc
  
\newcommand{\tsig}{\boldsymbol{\tilde{\Sigma}}} 
\newcommand{\doublenorm}[1]{\left\lvert\!\left\lvert #1 \right\rvert\!\right\rvert}

\definecolor{softred}{RGB}{200,50,50} 
\definecolor{softgreen}{RGB}{50,168,82} 
\definecolor{softblue}{RGB}{40, 110, 200}
\definecolor{darkorange}{RGB}{200,100,0}

\newcommand{\red}[1]{\color{softred}}
\newcommand{\green}[1]{\color{softgreen}}
\newcommand{\blue}[1]{\color{softblue}}
\newcommand{\orange}[1]{\color{darkorange}}

\usepackage{todonotes}

\tikzstyle{startstop} = [rectangle, rounded corners, minimum width=3cm, minimum height=1cm,text centered, draw=black, fill=red!30]
\tikzstyle{process} = [rectangle, minimum width=3cm, minimum height=1cm, text centered, draw=black, fill=orange!30]
\tikzstyle{decision} = [diamond, minimum width=3cm, minimum height=1cm, aspect=2, text centered, draw=black, fill=green!30]
\tikzstyle{arrow} = [thick,->,>=stealth]

\usepackage{mathtools}
\mathtoolsset{showonlyrefs}
\usepackage{authblk}
\PassOptionsToPackage{hyphens}{url}
\usepackage[colorlinks,citecolor=blue,urlcolor=blue]{hyperref} % Pour les références croisées en couleur

\title{Complex discontinuities of the square root of Fredholm determinants in the Volterra Stein-Stein model}

\author[1]{Eduardo Abi Jaber\footnote{eduardo.abi-jaber@polytechnique.edu}}
\author[1,2,3]{Maxime Guellil\footnote{maxime.guellil@polytechnique.edu. I am grateful for the financial support provided by Crédit Agricole CIB. \\
This work is supported by the Chair ``Capital Markets Tomorrow: Modeling and Computational Issues'', a joint initiative of LPSM/Université Paris-Cité and Crédit Agricole CIB under the aegis of Institut Europlace de Finance.\\ Both authors would like to thank Stéphane Crépey for fruitful discussions.}}
\affil[1]{École Polytechnique, CMAP}
\affil[2]{Université Paris Cité, LPSM, CNRS UMR 8001}
\affil[3]{Crédit Agricole CIB}
\date{}

\date{}
\begin{document}
\maketitle

\begin{abstract}
    Fourier-based methods are central to option pricing and hedging when the Fourier–Laplace transform of the log-price and integrated variance is available semi-explicitly. This is the case for the Volterra Stein–Stein stochastic volatility model, where the characteristic function is known analytically. However, naive evaluation of this formula can produce discontinuities due to the complex square root of a Fredholm determinant, particularly when the determinant crosses the negative real axis, leading to severe numerical instabilities. We analyze this phenomenon by characterizing the determinant’s crossing behavior for the joint Fourier–Laplace transform of integrated variance and log-price. We then derive an expression for the transform to account for such crossings and develop efficient algorithms to detect and handle them. Applied to Fourier-based pricing in the rough Stein–Stein model, our approach significantly improves accuracy while drastically reducing computational cost relative to existing methods.
\end{abstract}



\section*{Introduction}

In the field of mathematical finance, considerable effort has been devoted by both academics and practitioners to deriving explicit formulas for option prices and hedging strategies in stochastic volatility models, including those proposed by \citet{stein1991stock}, \citet{heston1993closed}, and \citet{schobel1999stochastic}. Beyond their theoretical appeal, such expressions offer major practical advantages: they reduce computational time compared with Monte Carlo simulations, improve numerical precision in pricing and hedging, and provide deeper insights into the model parameters and their sensitivities.
Among the most widely used approaches, Fourier-based pricing techniques, such as those developed by \cite{carr1999option, Lewis2002, lipton2001mathematical, fang2009novel, eberlein2010analysis}, stand out as a powerful framework for deriving analytical and numerically tractable valuation formulas.
They exploit analytical expressions for the Fourier–Laplace transform of the log-price to efficiently compute option prices through Fourier inversion, making them particularly attractive for both industry and research applications.
These methods are especially valuable for model calibration, where thousands of derivative prices across multiple maturities and strikes must be evaluated simultaneously in real time, requiring algorithms that combine speed, stability, and precision, see \cite{andersen2000jump}.

However, in several stochastic volatility models, the analytical formulas for the Fourier–Laplace transform are known to exhibit discontinuities due to multi-valued operations such as complex square roots or logarithms, leading to severe numerical instabilities and inaccuracies.
These complex discontinuities often arise when transform formulas, originally derived for real-valued parameters, are improperly extended to the complex plane, typically in the context of fast Fourier-based pricing schemes.
A notable example is found in the original article on the \cite{heston1993closed} model, where an incorrect formula for the characteristic function of the log-price was provided and later identified by \cite{schobel1999stochastic}, and subsequently analyzed and corrected by \cite*{albrecher2007little, gatheral2011volatility, Kahl05, lord2006rotation}. To overcome this issue, \cite{Kahl05} proposed a method, known as the \textit{Rotation count algorithm} that consists in tracking discontinuities of a complex-valued function by counting the number of times it crosses the negative real axis. Although the benefits of this algorithm have been numerically demonstrated, for instance in the Heston model, Lord and Kahl only proved that in certain cases, applying the algorithm to calculate the characteristic function would result in a continuous function, see \citet[Part 4, Lemma 4 and Theorem 4]{Kahl08}. Moreover, they did not prove that the function obtained using the algorithm was indeed the desired characteristic function.

Similar problems have surfaced in other models with explicit expressions for the characteristic functions, including the stochastic volatility models of \cite{stein1991stock}  and \cite{schobel1999stochastic}, and even in higher-dimensional cases, such as finite-dimensional Wishart distributions.  \cite{mayerhofer2019reforming} recently highlighted this issue, particularly where square roots of matrix determinants appear in the formulas. The principal branch of the square root introduces discontinuities whenever the determinant crosses the negative real line $\mathbb{R}_-$. \cite{mayerhofer2019reforming} showed that such crossings might occur if the rank of the underlying matrix is greater than or equal to three. 

More recently, in the context of non-Markovian stochastic volatility models, a formula for the Fourier–Laplace transform of the Volterra Stein-Stein model has been derived by \citet[Theorem 3.3]{abi2022characteristic}. The model for the stock price $S$ with a stochastic volatility $|X|$ is defined as follows:  
\begin{equation}\label{eq:SSmodel}
    \begin{aligned}
    dS_t &= S_t X_t  dB_t, \quad S_0 > 0, \\
    X_t &= g_0(t) + \int_0^t K(t,s) \kappa X_s \, ds + \int_0^t K(t,s) \nu  dW_s,
    \end{aligned}
\end{equation}
with $ \rho \in [-1, 1]$, $\kappa, \nu \in \mathbb{R}$, $ g_0 $ a suitable deterministic input curve,  $K : [0, T]^2 \to \mathbb{R}$ a measurable, square integrable Volterra kernel, and $B = \rho W + \sqrt{1 - \rho^2} W^\perp $ where $(W, W^\perp)$ is a two-dimensional Brownian motion. Such formulas have been recently extended to include stochastic interest rates in \cite*{jaber2025volterrasteinsteinmodelstochastic}, where it was shown that the choice of kernels allows the model to reproduce key empirical features, such as the humped shaped ATM volatility term structure for cap options and the concave ATM skew term structure (in log-log scale) of S\&P 500 options.

In the real domain, \citet[Remark 3.4]{abi2022characteristic} shows that the formula for the Fourier–Laplace transform can be simplified to computing the square root of a \cite{Fredholm1903} determinant, which generalizes the notion of a matrix determinant to  operators in infinite-dimensional spaces. Similar expressions involving Fredholm determinants have appeared in bond pricing formulas derived in models where the short rate is given by a square of a Gaussian process; see, for example, \citet*{abi2022laplace, corcuera2013short, grasselli2005wiener,privault2015fredholm}. This simplification in terms of the Fredholm determinant is particularly appealing for numerical purposes, as it reduces computational complexity and improves precision. However, in the complex plane, this simplification may no longer hold, as the behavior of the Fredholm determinant, particularly whether it crosses the negative real axis, remains unclear. As a result, one must rely on a more intricate formula involving  time-integrated traces of  operators, which complicates numerical implementations, slows down calibration processes, and therefore diminishes the attractiveness of Fourier techniques in this context.  

The aim of this paper is to address the complex discontinuities that may arise from the complex square-root of the Fredholm determinant and to explore three key questions:  
\begin{enumerate}
    \item \textbf{Crossing behavior}: Does the Fredholm determinant cross the nonpositive real axis 
    $\mathbb{R}_-$ in the complex plane?  
    \item \textbf{Extension to the complex plane}: Can we extend to the complex plane the formula derived in \citet[Remark 3.4]{abi2022characteristic} for the Fourier–Laplace transform while accounting for such crossings?
    \item \textbf{Numerical algorithms}: Can we develop efficient numerical methods to compute the extended formula?
\end{enumerate}

In Section \ref{sec:model_&_char_fun}, we derive a formula for the Fourier–Laplace transform involving the complex square root of a Fredholm determinant, which is valid in the complex domain, see Theorem~\ref{thm:char_fun_sqrt_det_rot_count}. This formula introduces a prefactor $e^{i \pi k}$, where $k$ corresponds to the number of time the determinant has crossed the negative real axis. We thus recover the same structure than the one induced by the rotation count algorithm. However, the difference compared to classic stochastic volatility model, see for instance \citet[Algorithm 1]{Kahl08}, is that the number of revolutions $k$ of the determinant is generally not explicitly known, as it relies on an intricate formula based on the kernel $K$, which can be non-trivial (for example if $K \not\equiv 1$). 

In Section \ref{Sec:det_crossing}, we show that if the rank of a compact operator depending on $K$ is at least 3 (except in trivial cases), the Fredholm determinant  crosses the negative real axis infinitely many times, see Theorems~\ref{thm:crossing_neg_axis_integrated_variance}-\ref{thm:crossing_neg_axis_log_price} and Corollary \ref{cor:crossing_neg_axis_log_price}. Moreover, since for Fourier-based pricing techniques, discontinuities of the Fourier–Laplace transform are most significant before it has sufficiently decayed, we derive a sharp upper bound for the first crossing time, demonstrating that large eigenvalues for this compact operator lead to earlier crossings. Furthermore, we analyze the impact of the model parameters on the eigenvalue magnitudes. In particular, for the rough Stein-Stein model, where $K_H(t,s) = \mathbbm{1}_{s<t} \frac{(t-s)^{H-\frac{1}{2}}}{\Gamma(H + \frac{1}{2})}$ with $H \in (0,1)$, we demonstrate that the eigenvalues, and consequently the first crossing time, increase as the Hurst index $H$ approaches 0.

In Section \ref{Sec:hybrid_lip_algos}, we develop two algorithms for computing the number of crossing $n$, and thus being able to compute the prefactor $e^{i \pi k} = \pm 1$. The first one is based on a coarse approximation of the Fourier–Laplace transform based on the more complicated formula proposed in \cite{abi2022characteristic}, allowing one to recover the good sign for the Fourier–Laplace transform at a low computational cost. The second method relies on counting the number of times the determinant crosses the negative real axis by evaluating it on a sufficiently fine grid, whose resolution is determined by the Lipschitz constant of the determinant’s argument in order to ensure that no crossings are missed. The latter method was previously employed numerically by \citet[Remark 4.1]{abi2022characteristic}, without proper justification, to extend to the complex plane the Fourier–Laplace transform formula with Fredholm determinant originally proved in the real domain. However, it did not account for the Lipschitz constant, potentially leading to instability depending on the model parameters, as shown in Figure \ref{fig:arg_of_det_no_care_lip_cst}. 
Furthermore, we introduce a prefactor-free approximation formula for the Fourier–Laplace transform, which does not require the computation of the prefactor. Instead, it involves computing the determinant of the square root of a matrix, rather than the square root of its determinant.

The applicability of the new Fourier–Laplace transform formula, together with the two algorithms and the prefactor-free approximation formula, is presented in Section~\ref{sec:numerical_results}. There, we compare precision and computation time with the original formula, showing both significantly higher accuracy and faster computation. A detailed comparison between the two algorithms and the prefactor-free approximation formula is also provided. A Jupyter Notebook containing all the necessary code to reproduce the figures of this article is available on Google Colab at \url{https://colab.research.google.com/drive/1O87IJpqGo1E8oGXnZxFqamYnm3Oa-_8d?usp=sharing}. \\

\textbf{Notations}.
Let $\mathbb K=\mathbb R$ or $\mathbb C$. We denote by $\mc{B}(L^2_{\mathbb K})$ the set of bounded linear operators on the Hilbert space $L_{\mathbb K}^2 := \left(L^2([0,T], \mathbb{K}), \langle \cdot, \cdot \rangle_{L^2_{\mathbb K}} \right)$, where the inner product is defined by $\langle f, h \rangle_{L^2_{\mathbb K}} := \int_0^T f(s) \Bar{h}(s) \, ds$, $\overline{z}$ denoting the complex conjugate of $z$. The space $\mc{B}(L^2_{\mathbb K})$ is naturally endowed with the uniform operator norm, which induces the uniform topology. This norm is defined as $\|\bo{A}\| := \sup_{\lVert f \rVert_{L^2_{\mathbb K}}=1} \lVert\bo{A} f\rVert_{L^2_{\mathbb K}}$. We point out that the inner product $\langle \cdot, \cdot \rangle_{L^2_{\mathbb R}}$ defines a symmetric bilinear form on $L^2([0,T], \mathbb{C})$. When there is no ambiguity, we will also denote the space $L^2([0,T]^2, \mathbb{K})$ as $L_{\mathbb K}^2$. Finally, $\text{id}$ denotes the identity operator.

\subsection{Reminders on the principal branch of the logarithm and square-root}
For a complex number $z\in \mathbb C$, we will write its polar form as $z=|z| e^{i\theta}$, where $|z| \geq 0$ is the module of $z$ and $\theta \in \{ \arg{(z)} + 2 \pi n : n \in \mathbb Z \}$ an argument of $z$, with  $-\pi<\arg{(z)}\leq \pi$ being the principal argument of $z$.   We denote by $\sqrt[\bullet]{z}$ the two-valued square-root of $z$ defined by the solutions of the equation $w^2=z$, that is
\begin{equation}
    \sqrt[\bullet]{z} := \{w \in \mathbb C, \; w^2 = z\}
\end{equation} 
and by $\sqrt{z}$ the principal branch of the square-root of $z$ defined by 
\begin{equation}
    \sqrt{z} := e^{\frac{1}{2}\log(z)} = \sqrt{|z|} e^{i \frac{\arg{(z)}}2}, \quad z \in \mathbb C^*,
\end{equation}
and $\sqrt{0}=0$. In particular, when $z \neq 0$, $\sqrt{z}$ is the unique element of $\sqrt[\bullet]{z}$ with a strictly positive real part.  Similarly, we denote by $\log_n$ the function defined on $\mathbb C^*$ by
$$\log_n(z) := \ln(|z|) + i \left(\arg{(z)} + 2 \pi n\right)$$
and by $\log$ the function $\log_0$ which is called the principal branch of the logarithm.  We also recall that the principal branch of the logarithm and of the square-root are holomorphic  on $\mathbb C\backslash (-\infty, 0]$, i.e.~everywhere except on the set of non-positive real numbers, but fail to be even continuous on $(-\infty, 0)$.

\subsection{Reminders on integral operators} \label{subsubsec:recall_integral_op}

Let $\mathbb K=\mathbb R$ or $\mathbb C$. For any $K, L \in L^2([0, T]^2, \mathbb{K})$, we define the $\star$-product by
\begin{equation} \label{eq:star_product}
    (K \star L)(s, u) = \int_0^T K(s, z)L(z, u) \, dz, \quad (s, u) \in [0, T]^2,
\end{equation}
which is well-defined in $L^2([0, T]^2, \mathbb K)$ due to the Cauchy-Schwarz inequality. For any kernel $K \in L^2([0, T]^2, \mathbb K)$, we denote by $\bo{K}$ the integral operator induced by the kernel $K$, that is
\begin{equation}
    (\bo{K}f)(s) = \int_0^T K(s, u)f(u) \, du, \quad f \in L^2([0, T], \mathbb{K}).
\end{equation}
$\bo{K}$ is a linear bounded operator from $L^2([0, T], \mathbb K)$ into itself. If $\bo{K}$ and $\bo{L}$ are two integral operators induced by the kernels $K$ and $L$ in $L^2([0, T]^2, \mathbb K)$, then $\bo{K}\bo{L}$ is also an integral operator induced by the kernel $K \star L$. We use $\bo{K}$ to denote the integral operator with kernel $K$, $\bo{K}^*$ its adjoint, which is also an integral operator with kernel $K^*$ defined by $K^* = \overline{K}$.

An integral operator $\bo{K}$ is an Hilbert-Schmidt operator. Therefore, the product of two integral operators $\bo{K}$ and $\bo{L}$ is of trace class, and we have
\begin{equation}\label{eq:TR}
    \Tr(\bo{K} \bo{L}) = \int_0^T (K \star L)(s, s) \, ds,
\end{equation}
see \citet[Proposition 3]{brislawn1988kernels}. 

A kernel $K \in L^2([0,T]^2, \mathbb{K})$ is said to be separable if it can be written as a finite sum of tensor products of functions, that is:
\begin{equation}
    K(s,u) = \sum_{k=1}^N f_k(s)h_k(u),
\end{equation}
where $f_k, h_k \in L^2([0,T], \mathbb{K})$ for $k=1,\dots,N$, and $N \in \mathbb N^*$. An integral operator $\bo{K}$ with separable kernel is of trace class, since its rank is finite (in fact less than $N$), see Section \ref{sec:recall_trace_det} for the definition of the rank. For further details on  Hilbert-Schmidt and trace class operators, see Appendix \ref{sec:recall_trace_det}.

\section{Fourier–Laplace transform and Fredholm determinant} \label{sec:model_&_char_fun}
In this section,  we start by recalling  the expression for the joint Fourier–Laplace transform of the log-price $\log S_T$ and integrated variance $\int_0^T X_s^2ds$ 
 of the Volterra Stein-Stein model \eqref{eq:SSmodel} derived by \cite{abi2022characteristic} in terms of the trace of operators, see  Theorem~\ref{thm:trace_formula}.
 
Then, in Section~\ref{S:complexdisc}, we illustrate the discontinuities of the Fourier–Laplace transform that can arise when using a more efficient formula involving the principal branch of the square root of the Fredholm determinant, without accounting for the discontinuities of the complex square root.

Finally, we establish the formula involving the Fredholm determinant in Theorem~\ref{thm:char_fun_sqrt_det_rot_count}, which takes care of these complex discontinuities. This constitutes our first main theoretical result. 

Let us first introduce some preliminaries. We define $g$ as the adjusted conditional mean given by
\begin{equation} \label{eq:g_t}
    g_t(s) = \mathbbm{1}_{t \leq s \leq T} \mathbb{E} \left[ X_s - \int_t^s K(s, r) \kappa X_r \, dr \bigg| \mc{F}_t \right], \quad 0 \leq s, t \leq T.
\end{equation}
We define the set
\begin{equation}\label{def:set_U}
    \mc{U} := \left\{ (u,w) \in \mathbb{C}^2: 0 \leq \Re(u) \leq 1, \; 0 \leq \Re(w) \right\}.
\end{equation}
For $0 \leq t \leq T$ and $u, w \in \mc U$, we define the following variables:
\begin{equation}\label{eq:a,b}
    a(u,w) := w + \frac{1}{2}(u^2 - u), \quad b(u) := \kappa + \rho \nu u.
\end{equation}
We define $\tsig_t(u)$ as the adjusted covariance integral operator  given by
\begin{equation}
    \tsig_t(u) := (\text{id} - b(u)\bo{K})^{-1} \bo{\Sigma}_t (\text{id} - b(u)\bo{K}^*)^{-1}, \label{eq:tilde_Sigma}
\end{equation}
where $\bo{\Sigma}_t$ is the integral operator with kernel
\begin{equation}\label{eq:Sigma}
    \Sigma_t(s, z) := \nu^2 \int_t^T K(s, r) K(z, r) \, dr, \quad t \leq s, z \leq T.
\end{equation} 
Finally, we define $\bo{\Psi}_t(u,w)$ as the linear operator
\begin{equation}\label{eq:Psi_t}
    \bo{\Psi}_t(u,w) := a(u,w) (\text{id} - b(u)\bo{K}^*)^{-1} \left( \text{id} - 2a(u,w)\tsig_t(u) \right)^{-1} (\text{id} - b(u)\bo{K})^{-1},
\end{equation}
and $\bo{\Phi}_t(u,w)$ as the operator
 \begin{equation}\label{eq:Phi_t}
    \bo{\Phi}_t(u,w) := \text{id} - 2a(u,w)\tsig_t(u).
\end{equation}
We consider the following class of square-integrable Volterra kernels $K$. 
\begin{definition} \label{def:cont_bound_kernel} A kernel $ K : [0, T]^2 \to \mathbb{R} $ is a Volterra kernel of continuous and bounded type in $L^2_{\mathbb R}$ if $ K(t, s) = 0 $ whenever $ s \geq t $ and
\begin{align}
    &\sup_{t \in [0, T]} \int_0^T |K(t, s)|^2 ds < \infty, \quad \sup_{t \in [0, T]} \int_0^T |K(s, t)|^2 \, ds < \infty, \\
    &\lim_{h \to 0} \int_0^T |K(t + h, s) - K(t, s)|^2 ds = 0, \quad t \leq T. 
\end{align}
\end{definition}
This class of kernels encompasses any Volterra kernel of convolution form $K(t,s)= 1_{s<t} k(t-s)$, with $\int_0^T k(s)^2 ds <\infty$. A notable example is the fractional kernels $k(t)=\frac{t^{H-\frac{1}{2}}}{\Gamma(H + \frac{1}{2})}$, with a Hurst parameter $H \in (0,1)$, which has a singularity at zero when $H<1/2$. In addition, any continuous Volterra kernel $K$ on $[0,T]^2$ satisfies Definition~\ref{def:cont_bound_kernel}, and if $K$ satisfies Definition \ref{def:cont_bound_kernel}, then so does its adjoint $K^*$.

Finally, the well-definedness of $\tsig$ and $\bo{\Psi}$ is ensured by the following result, whose proof combines  \citet[Lemma A.5]{abi2022characteristic} and Lemma \ref{lemma:invertibility}.

\begin{lemma} \label{lemma:tsig_psi_well_def}
    Let $0 \leq t \leq T$, $0 \leq s \leq 1$, $(u,w) \in \mc U$ and $K$ be a Volterra kernel of continuous and bounded type in $L^2_{\mathbb R}$. Then, $(\id - b(u)\bo{K})$ is invertible in $\mc{B}(L^2_{\mathbb C})$, and $\tsig_t(u)$ is a trace class integral operator with continuous kernel. Furthermore, $\id-2a(u,w)\tsig_t(u)$ is invertible in $\mc{B}(L^2_{\mathbb C})$, and $\bo{\Psi}_{t,s}(u,w)$ is well defined.
\end{lemma}

\begin{remark} \label{rmk:det_Phi_well_def}
    For $(t,u,w) \in [0,T] \times \mc U$, since from Lemma \ref{lemma:tsig_psi_well_def}, $\tsig_t(u)$ is a trace class operator, we can consider its Fredholm determinant, see Appendix \ref{sec:recall_trace_det}. In particular, $\det(\bo{\Phi}_t(u,w))$ is well defined.
\end{remark}

\begin{remark} \label{rmk:kappa=0}
Following \citet[Theorem A.3]{abi2022characteristic}, let $g_0 \in L^2([0,T],\mathbb R)$ and let $K$ be a Volterra kernel of continuous and bounded type in $L^2_{\mathbb R}$. Then, for all $0 \leq t \leq T$, the process $(X_t)_t$ satisfies
\begin{equation} \label{eq:integral_eq_X_t_kappa=0}
    X_t = (\id + \bo{R}_{\kappa})g_0(t) + \int_0^t \frac{1}{\kappa} R_{\kappa}(t,s) \nu\, dW_s,
\end{equation}
where $\bo{R}_{\kappa}$ is the integral operator with kernel $R_{\kappa}$, the resolvent kernel of $\kappa K$. By convention, when $\kappa = 0$, we set $\frac{1}{\kappa} R_{\kappa} := K$. Moreover, from \citet[Lemma A.2]{abi2022characteristic}, $R_{\kappa}$ is itself a Volterra kernel of continuous and bounded type in $L^2_{\mathbb R}$.

We observe that \eqref{eq:integral_eq_X_t_kappa=0} corresponds to the original integral equation \eqref{eq:SSmodel} for $(X_t)_t$ in the special case $\kappa = 0$, where $g_0$ is replaced by $\tilde{g}_0 := (\id + \bo{R}_{\kappa})g_0$ and $K$ is replaced by $\tilde{K} := \frac{1}{\kappa} R_{\kappa}$. Therefore, without loss of generality, we reduce our analysis to the case $\kappa = 0$.
\end{remark}

The joint Fourier–Laplace transform of the log-price and the integrated variance are derived in \citet[Theorem 3.3]{abi2022characteristic}. We recall the result in the next theorem.

\begin{theorem}[Trace formula for the Fourier–Laplace transform]\label{thm:trace_formula}

Let $g_0 \in L^2([0,T], \mathbb{R})$ and $K$ be a Volterra kernel of continuous and bounded type in $L^2_{\mathbb R}$. Let $t \leq T$ and $u, w \in \mathcal U$. Then,
\begin{equation}\label{eq:char_func_trace_formula}
    \mathbb{E} \left[ \exp \left( u \log \frac{S_T}{S_t} + w \int_t^T X_s^2 ds \right) \bigg| \mc{F}_t \right] = \exp \left( \phi_t(u,w) + \langle g_t, \bo{\Psi}_t(u,w) g_t \rangle_{L^2_{\mathbb R}} \right), 
\end{equation}
with
\begin{equation} \label{eq:phi_t}
    \phi_t(u,w) = -\int_t^T \mathrm{Tr} (\bo{\Psi}_s(u,w) \dot{\bo{\Sigma}}_s) ds, 
\end{equation}
where $\dot{\bo{\Sigma}}_t$ is the strong derivative of $t \mapsto \Sigma_t$ induced by the kernel
\begin{equation} \label{eq:dot_sig_t}
    \dot{\Sigma}_t(s, z) = -\nu^2 K(s,t)K(z,t), \quad \text{a.e.},
\end{equation}
and $\mathrm{Tr}$ is the trace operator.
\end{theorem}

We will refer to the formula \eqref{eq:char_func_trace_formula} for the Fourier–Laplace transform, which involves an integral representation \eqref{eq:phi_t} for $\phi_t$, as the \textbf{trace formula for the Fourier–Laplace transform}, since the representation for $\phi_t$ includes the trace of an operator. 

\begin{remark} \label{rmk:why_using_det}
At this point, fixing $t,u$ and $w$, it is crucial to emphasize that in order to numerically compute the Fourier–Laplace transform from \eqref{eq:char_func_trace_formula}, one must discretize both $\bo{\Psi}_t(u,w)$ and $g_t$ to compute the inner product, as well as $\bo{\Sigma}_t$ to compute $e^{\phi_t(u,w)}$. Additionally, the integral in \eqref{eq:phi_t} must also be discretized, since in general, no analytical formula is available for $\phi_t(u,w)$. Therefore, effective computation requires discretizing both the operators and the integral. 
\end{remark}

\subsection{The complex discontinuity problem when using the Fredholm determinant}\label{S:complexdisc}

\citet[Remark 3.4]{abi2022characteristic} demonstrated that if $u$ and $w$ are real numbers, the numerical computation of the trace formula in  \eqref{eq:phi_t}  can be significantly simplified. For simplicity, let's omit in this paragraph any dependence on $u$ and $w$. Specifically, one can reduce the computational effort to the discretization of the operators by making a link between $e^{\phi_t}$ and the \cite{Fredholm1903} determinant of $\bo{\Phi}_t$. To achieve this, he shows that $\varphi: t \mapsto \varphi_t := \log(\det(\bo{\Phi}_t))$ is well defined, differentiable and verifies 
\begin{align}
    \begin{cases}
        \partial_t \varphi_t = -2 \partial_t \phi_t, &  0 \leq t \leq T \\
        \varphi_T = \phi_T = 0.
    \end{cases} \label{eq:diff_log_det}
\end{align}
Thus, it follows that 
\begin{equation} 
    \phi_t = -\frac{1}{2}\log(\det(\bo{\Phi}_t)),
\end{equation}
and
\begin{equation}\label{eq:phi=sqrt_det}
    e^{\phi_t} = \frac{1}{\sqrt{\det(\bo{\Phi}_t)}}.
\end{equation}
This alternative formula for $e^{\phi_t}$ is numerically really interesting. Indeed, as  hinted in Remark~\ref{rmk:why_using_det}, in order to compute $e^{\phi_t}$ using \eqref{eq:phi_t}, one needs to discretize both the integral and the operators $\bo{\Psi}_t$ and $\bo{\Sigma}_t$, whereas when using \eqref{eq:phi=sqrt_det}, one only needs to discretize the operator $\bo{\Phi}_t$. 

However, in the general case where \( u, w \in \mathbb{C} \), such that \( a, b \in \mathbb{C} \), the function \( t \mapsto \det(\bo{\Phi}_t) \) is complex-valued. If this function crosses the negative real axis over $[0,T]$, then \( \varphi_t \) would not be differentiable as the principal branch of the logarithm is not even continuous on the negative real axis. Consequently, the previous reasoning would no longer apply, and as we will show in the next section, \eqref{eq:phi=sqrt_det} does not hold in general. However, if these discontinuities are neglected and one continues to use \eqref{eq:phi=sqrt_det}, the resulting Fourier–Laplace transform becomes discontinuous.

Figure \ref{fig:chardiscontinuity} illustrates this in the case of the rough Stein-Stein model, corresponding to \eqref{eq:SSmodel} with the following Riemann-Liouville fractional kernel, and input curve:
\begin{align}
    K(t,s) := \mathbbm{1}_{s<t} \frac{(t-s)^{H-\frac{1}{2}}}{\Gamma(H + \frac{1}{2})}, \quad 
    g_0(t) := X_0 +  \int_0^t K_H(t,s) \theta \, ds 
    = X_0 + \theta \frac{t^{H+\frac{1}{2}}}{\Gamma(H + \frac{1}{2})(H + \frac{1}{2})},
\end{align}
for $s,t\leq T$, with $H\in (0,1)$ and $X_0,\theta \in \mathbb R$. 

We compute the Fourier–Laplace transform using the trace formula \eqref{eq:char_func_trace_formula} in combination with \eqref{eq:phi=sqrt_det}. The  method follows the approach described in Section \ref{subsec:discretized_operators}, where a partition $(t_i := \frac{iT}{n})_{1\leq i\leq n}$ of $[0,T]$ is used for discretizing the operators.

We observe from the figure that the two graphs differ by a factor -1 between each discontinuity point. We will see that it is indeed the case, and that each discontinuity point corresponds to a moment where the determinant crosses the negative real axis. In particular, we will see that \eqref{eq:phi=sqrt_det} is always true up to a factor -1.

\begin{figure}[h!]
    \centering
    \includegraphics[width=\textwidth]{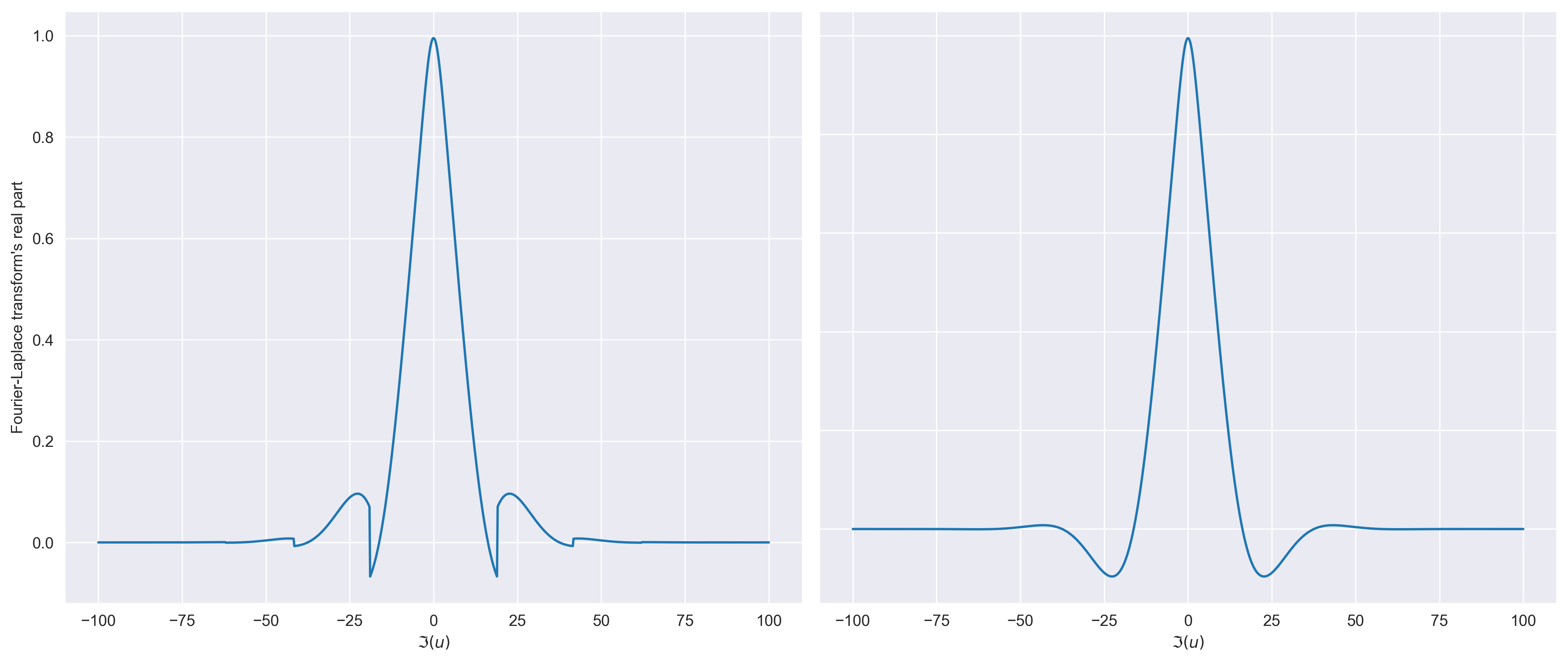} 
    \caption{Real part of the Fourier–Laplace transform as a function of $\Im(u)$ in the fractional kernel setting, computed via the determinant representation \eqref{eq:phi=sqrt_det} (Left) and the trace representation \eqref{eq:phi_t} (Right) for $\phi_t$. The parameters are $\kappa = 0$, $\nu = 0.2$, $\theta = -0.1$, $\rho = -0.9$, $X_0 = -0.05$, $H = 0.3$, $T = 1$, $t = 0$, $n=200$, $\Re(u) = \frac{1}{2}$, $w = 0$.}
    \label{fig:chardiscontinuity}
\end{figure}
\FloatBarrier

\subsection{First main result: Extending the Fourier–Laplace transform to the complex plane} \label{subsec:det_formula}

Our first main result, Theorem~\ref{thm:char_fun_sqrt_det_rot_count}, provides the extension of the Fourier–Laplace transform formula to the complex plane. In particular, Lemma~\ref{lemma:phi=sqrt_det} shows that $e^{\phi_t(u,w)}$ is the reciprocal of the principal square root of $\det(\bo{\Phi}_t(u,w))$, up to a sign determined by the number of crossings of the nonpositive real line $\mathbb{R}_-$ by the Fredholm determinant. This crossing behavior is encoded in the integer $k_t(u,w) \in \mathbb{Z}$, defined by
\begin{equation}\label{eq:k_t(u,w)}
\pi k_t(u,w) := \frac{\arg(\det(\bo{\Phi}_t(u,w)))}{2} + \Im(\phi_t(u,w)).
\end{equation}

\begin{theorem}[Determinant formula for the Fourier–Laplace transform]\label{thm:char_fun_sqrt_det_rot_count}
    Let $ g_0 \in L^2([0,T], \mathbb{R}) $ and $ K $ be a Volterra kernel of continuous and bounded type in $L^2_{\mathbb R}$. Let $t\leq T$, $(u, w) \in \mc U$ and let $k_t$ be given by \eqref{eq:k_t(u,w)}. Then,
    \begin{equation} \label{eq:char_func_det_formula}
        \mathbb{E} \left[ \exp \left( u \log \frac{S_T}{S_t} + w \int_t^T X_s^2 ds \right) \bigg| \mc{F}_t \right] = e^{i \pi k_t(u,w)} \frac{\exp \left(\langle g_t, \bo{\Psi}_t(u,w) g_t \rangle_{L^2_{\mathbb R}} \right)}{\sqrt{\det(\bo{\Phi}_t(u,w))}},
    \end{equation}
 In particular, for any connected subset $\mc V \subset \mc U$ containing two real numbers $(x,y)$ where $0 \leq x \leq 1$ and $0 \leq y$, if $\det(\bo{\Phi}_t): \mc V \to \mathbb C$ never crosses the negative real axis, then, for all $(u,w)\in \mc V$, it holds  that $k_t(u,w)=0$ and  
 \begin{equation} \label{eq:char_fun_sqrt_det}
        \mathbb{E} \left[ \exp \left( u \log \frac{S_T}{S_t} + w \int_t^T X_s^2 ds \right) \bigg| \mc{F}_t \right] =  \frac{\exp \left(\langle g_t, \bo{\Psi}_t(u,w) g_t \rangle_{L^2_{\mathbb R}} \right)}{\sqrt{\det(\bo{\Phi}_t(u,w))}}.
    \end{equation}
\end{theorem}

\begin{proof}
    The proof is given in Section~\ref{S:proofofthmcharfun}.
\end{proof}

We will refer to the formula \eqref{eq:char_fun_sqrt_det} for the Fourier–Laplace transform, which involves the determinant representation \eqref{eq:det_exp_phi} for $\phi$, as the \textbf{determinant formula for the Fourier–Laplace transform}.

\begin{remark}
    By considering $\mc V := \Re(\mc U) \subset \mc U$, the set of real numbers $(u,w)$ verifying $0 \leq u \leq 1$ and $w \leq 0$, and because the determinant is positive since $\tsig_t(u)$ is a positive operator on $\mc V$, we re-obtain formula \eqref{eq:phi=sqrt_det} proved in \citet[Remark 3.4]{abi2022characteristic}: $e^{\phi_t(u,w) } = \frac{1}{\sqrt{\det(\bo{\Phi}_t(u,w))}}$ for $u,w$ real numbers.
\end{remark}

From a numerical perspective, evaluating the determinant formula in Theorem~\ref{thm:char_fun_sqrt_det_rot_count} is less computationally expensive than computing the trace formula of Theorem~\ref{thm:trace_formula}, provided that one knows the sign of the prefactor $e^{i \pi k_t(u,w)}$ which takes value in $\pm 1$ (since $k_t(u,w)\in \mathbb Z$) and is a priori unknown, since we only know the value of $k_t(u,w)$ from \eqref{eq:k_t(u,w)}, which depends on $\phi_t(u,w)$ that we do not want to compute anymore via its trace representation \eqref{eq:phi_t}.

Furthermore, in the context of pricing derivatives on the stock price  or the integrated variance by using Fourier inversion techniques, one would, in the first case, which we will refer to as the \textbf{log-price case}, set the value of $w$ to 0, fix the real part of $u$, and allow its imaginary part to vary between 0 and $\infty$. In the second case, referred to as the \textbf{integrated variance case}, this is the contrary. For instance, the Lewis formula \citet[Formula 3.11]{Lewis2002} for calculating the price of a put option on $S$ with strike $K$ is given by:
\begin{align} \label{eq:Lewis_formula}
P &= K - \frac {\sqrt{S_0K}}{\pi } \int_0^{+\infty} \Re \left(e^{i\Im(u)k} \frac{\xi_T\left(\frac{1}{2}+i\Im(u)\right)}{\Im(u)^2 + \frac{1}{4}}\right) d\Im(u), \\
\xi_T(u) :&= \mathbb{E} \left[ \exp \left( u \log \frac{S_T}{S_0}\right)\right],
\end{align}
where $k = \log(S_0/K)$.

Therefore, it is important to highlight that as the integrand in \eqref{eq:Lewis_formula} decays rapidly to zero, and that the difference between Formulas \eqref{eq:char_func_det_formula} and \eqref{eq:char_fun_sqrt_det} is by a factor $\pm 1$, one could argue that if the crossings of the determinant only happen where the integrand is close to zero, using one or the other formula might not make a big difference in the put price, and that we would additionally save the computation of the prefactor $e^{i \pi k_t(u)}$ when using \eqref{eq:char_fun_sqrt_det}, which would improve the computation time. \\ 
From all these observations, three key questions arise:

\begin{enumerate}
    \item \textbf{Crossing behavior}: Does the Fredholm determinant cross the nonpositive real axis     $\mathbb{R}_-$ when $\Im(u)$ or $\Im(w)$ is varying? 
    \item \textbf{First crossing's position}: If crossings occur, where does the first crossing appear?
    \item \textbf{Numerical algorithms}: Can efficient numerical methods be developed to determine the sign of the prefactor $e^{i \pi k_t(u,w)}$?  
\end{enumerate}

The first two questions are addressed in Section~\ref{Sec:det_crossing} and the third one in Section~\ref{Sec:hybrid_lip_algos}. Furthermore, we propose The rest of this section is dedicated to the proof of Theorem~\ref{thm:char_fun_sqrt_det_rot_count}.

The first two questions are addressed in Section~\ref{Sec:det_crossing}, and the third one in Section~\ref{Sec:hybrid_lip_algos}, where we also propose a prefactor-free approximation of the Fourier–Laplace transform that avoids the explicit computation of the prefactor by evaluating the determinant of a matrix square root instead of the square root of its determinant. The rest of this section is devoted to the proof of Theorem~\ref{thm:char_fun_sqrt_det_rot_count}.

\subsection{Proof of Theorem~\ref{thm:char_fun_sqrt_det_rot_count}} \label{S:proofofthmcharfun}

For $K$ being a Volterra kernel of continuous and bounded type in $L^2_{\mathbb R}$, we introduce the following notation for $\det(\bo{\Phi})$:

\begin{equation}
    d_t(u,w) := \det(\bo{\Phi}_t(u,w)), \quad t \leq T, \, (u,w) \in \mc{U}.
\end{equation}

The key to recover a relationship between $e^{\phi}$ and $\sqrt{d_t}$ is the following: for fixed $(u,w) \in \mc U$, instead of differentiating $t \mapsto \log(d_t(u,w))$ as in \eqref{eq:diff_log_det}, which may not be differentiable as already discussed, we will directly differentiate $t \mapsto d_t(u,w)$. This slightly different approach is detailed in the following lemma.

\begin{lemma} \label{lemma:det_exp_phi}
    Let $(u,w) \in \mc U$ and $K$ be a Volterra kernel of continuous and bounded type in $L^2_{\mathbb R}$. Then $t \mapsto d_t(u,w)$ is differentiable and satisfies the differential equation
    \begin{align}
        \begin{cases}
            \partial_t d_t(u,w) = -2 \phi_t(u,w) d_t(u,w), &  0 \leq t \leq T \\
            d_T(u,w) = \phi_T(u,w) = 0.
        \end{cases} \label{eq:diff_log_det}
    \end{align}
    Therefore, 
    \begin{equation} \label{eq:det_exp_phi}
        d_t(u,w) = e^{-2 \phi_t(u,w)}, \quad 0 \leq t \leq T.
    \end{equation}
\end{lemma}

\begin{proof}
    Let $(u,w) \in \mc U$. For the sake of readability, we omit all dependence on $u, w$. Differentiation using the derivative of the Fredholm determinant, see \citet[Chapter 4, p.158, 1.3]{gohberg69}, as well as \eqref{eq:tilde_Sigma} yields, for $0 \leq t \leq T$,
    \begin{equation}
        \partial_t d_t = d_t \Tr \left( -2a \left( \id - 2 \tsig_t a \right)^{-1} {{\boldsymbol{\dot{\tilde{\Sigma}}}}_t} \right) =-2 d_t \Tr \left( a \left( \id - 2 \tsig_t a \right)^{-1} \left(\id - b\bo{K}\right)^{-1} \dot{\bo{\Sigma}}_t \left(\id - b\bo{K}^*\right)^{-1} \right).
    \end{equation}
    Set $\bo{F} := \left( \id - 2 \tsig_t a \right)^{-1} \left(\id - b\bo{K}\right)^{-1} \dot{\bo{\Sigma}}_t$ and $\bo{G} := \left(\id - b\bo{K}^*\right)^{-1}$. Recalling \eqref{eq:dot_sig_t}, $\dot{\bo{\Sigma}}_t$ is an integral operator with a separable kernel, and is therefore of trace class, see Section \ref{subsubsec:recall_integral_op}. Moreover, since from Lemma \ref{lemma:tsig_psi_well_def}, $\left( \text{id} - 2 \tsig_t a \right)^{-1} (\text{id} - b\bo{K})^{-1}$ is a bounded linear operator, it follows from \citet[Chapter 9, Section 2, Exercise 20]{conway2019course} that $\bo{F}$ is of trace class. Furthermore, Lemma \ref{lemma:tsig_psi_well_def} also ensures that $\bo{G}$ is a bounded linear operator. Thus, applying \citet[Chapter 9, Section 2, Exercise 20]{conway2019course} once more, we obtain the identity  
    \begin{equation}
    \Tr(\bo{F}\bo{G}) = \Tr(\bo{G}\bo{F}),
    \end{equation}
    from which we obtain, combined with \eqref{eq:phi_t}, that
    \begin{equation} \label{eq:differential_det}
        \partial_t d_t = -2 \phi_t d_t, \quad 0 \leq t \leq T.
    \end{equation}
    
    Solving this first order linear differential equation, and using the terminal condition lead to Formula \eqref{eq:det_exp_phi}.
\end{proof}

Since the Fourier–Laplace transform \eqref{eq:char_func_trace_formula} depends on the quantity $e^{\phi_t(u,w)}$, one might be tempted to take the principal branch of the square-root in  \eqref{eq:det_exp_phi} in order to obtain the identity:
\begin{align}\label{eq:wrongforumla}
    `` e^{\phi_t(u,w)} = \frac{1}{\sqrt{d_t(u,w)}} ".
\end{align} 
As explained previously, this identity does not generally hold. It is nevertheless important to observe that the inconsistency here differs from the one arising when differentiating $t \mapsto \log(d_t(u,w))$ as in \eqref{eq:diff_log_det}. 
Indeed, the first case is static ($t$ is fixed) whereas the second is dynamic. More precisely, in the first case, the error was to believe that $\sqrt{e^{2z}}=e^z$ which is only true when $\Im(z) \in (-\frac{\pi}{2}, \frac{\pi}{2}]$, whereas in the second case, the error is to believe that as soon as a complex valued function $t \mapsto f(t)$ is differentiable, then so is $t \mapsto \log(f(t))$, which is true only if $f$ never crosses the negative real axis.

The following lemma establishes the relation between $e^{\phi_t(u,w)}$ and $\sqrt{d_t(u,w)}$.

\begin{lemma}\label{lemma:phi=sqrt_det} 
    Let $t\leq T$, $(u,w) \in \mc U$, and $K$ be a Volterra kernel of continuous and bounded type in $L^2_{\mathbb R}$. Then,    \begin{equation}\label{eq:phi=sqrt_det_rot_count}
        e^{\phi_t(u,w) } = \frac{e^{i\pi k_t(u,w)}}{\sqrt{d_t(u,w)}},
    \end{equation}
    with $k_t(u,w)\in \mathbb Z$ given in \eqref{eq:k_t(u,w)}. 
\end{lemma}

\begin{proof} Let $t\leq T$ and $(u,w)\in \mc{U}$. To ease notations we drop the dependence in $(u,w)$.
    We first observe that the identity  \eqref{eq:det_exp_phi} yields that $|d_t|=e^{- 2 \Re(\phi_t)}$, so that the polar form of $d_t$ reads  
    $$ d_t = e^{- 2 \Re(\phi_t)} e^{i \arg{(d_t)}}.$$
    By definition of the principal branch of the square-root, it follows that 
        $$ \sqrt{d_t} = e^{- \Re(\phi_t)} e^{i \frac{\arg{(d_t)}}{2}}  =  e^{- \phi_t} e^{i \left(\Im(\phi_t) + \frac{\arg{(d_t)}}{2}  \right)} =  e^{- \phi_t} e^{i\pi k_t(u,w)}, $$
    where the last equality comes from the definition of $k_t(u,w)$ in \eqref{eq:k_t(u,w)}. Finally, since from Formula \eqref{eq:det_exp_phi}, $-2\Im(\phi_t)$ is an argument of $d_t$, we have that $-2\Im(\phi_t) - \arg(d_t) \in 2\pi \mathbb Z$, or equivalently that $k_t(u,w) \in \mathbb Z$.
\end{proof}

The following lemma establishes that $(t,u,w) \mapsto \phi_t(u,w)$ is continuous. This will help prove, in Lemma \ref{lemma:no_rot_count_needed}, that for a fixed $t \leq T$ and any connected subset $\mc V \subset \mc U$ containing two real numbers $(x,y)$, if the trajectory of $(u,w) \in \mc V \mapsto d_t(u,w)$ does not cross the negative real axis, then $e^{i\pi k_t(u,w)}=1$ for all $(u,w) \in \mc V$, or equivalently, formula \eqref{eq:wrongforumla} holds on $\mc V$. 

\begin{lemma} \label{lemma:phi_det_C^0}
    Let $K$ be a Volterra kernel of continuous and bounded type in $L^2_{\mathbb R}$. The functions $(u,w,t) \in \mc U \times [0,T] \mapsto \phi_t(u,w)$  and $(u,w,t) \in \mc U \times [0,T] \mapsto d_t(u,w)$ are continuous.
\end{lemma}

\begin{proof}
We first establish the continuity of $\phi$. The continuity in $t$ follows directly from its definition. 
By continuity under the integral, it suffices to show that 
\begin{equation}
   (u,w,t) \in \mc U \times [0,T] \mapsto 
   \Tr(\bo{\Psi}_t(u,w)\dot{\bo{\Sigma}}_t)
   = \Tr(\dot{\bo{\Sigma}}_t\bo{\Psi}_t(u,w))
\end{equation}
is continuous.

For any $t \in [0,T]$ and $(u,w) \in \mc U$, since 
\begin{equation}
   \dot{\Sigma}_t(s,z) = -\nu^2 K(s,t)K(z,t),
\end{equation}
the kernel $\dot{\Sigma}_t$ is separable and has rank at most one (see Section~\ref{subsubsec:recall_integral_op}). 
The image of the induced operator $\dot{\bo{\Sigma}}_t$ is thus spanned by $K(\cdot,t)$, and 
$\dot{\bo{\Sigma}}_t\bo{\Psi}_t(u,w)$ has at most one non-zero eigenvalue, corresponding to this vector.  
By Lidskii’s theorem \citet[Theorem~3.7]{simon2005trace}, the continuity of the trace 
is therefore equivalent to the continuity of this single eigenvalue. By \citet[Chapter~1, Theorem~4.2]{gohberg69}, it is enough to prove the continuity of 
\begin{equation}
   (u,w,t) \mapsto \bo{\Psi}_t(u,w)\dot{\bo{\Sigma}}_t
\end{equation}
with respect to the uniform norm on $\mc B(L^2_{\mathbb C})$.  
We treat the two factors separately.

Recalling the definition \eqref{eq:Psi_t} of $\bo{\Psi}_t(u,w)$, its continuity in $(u,w)$ follows from that of 
\begin{equation}
   (u,w) \mapsto (a(u,w), b(u)),
\end{equation}
and from the continuity of the inverse in the uniform norm. 
The continuity in $t$ follows from that of $t \mapsto \Sigma_t$ in $L^2_{\mathbb R}$, 
which is immediate from \eqref{eq:Sigma}.  
For the operator $\dot{\bo{\Sigma}}_t$, it suffices also to establish the continuity of $t \mapsto \dot{\Sigma}_t$ in $L^2_{\mathbb R}$: For any $h>0$,
\begin{equation}
   \doublenorm{\dot{\Sigma}_{t+h}-\dot{\Sigma}_t}_{L^2_{\mathbb R}}^2
   = \nu^4 \int_0^T (K(s,t+h)^2 - K(s,t)^2)^2 ds 
   \xrightarrow{h \to 0} 0,
\end{equation}
since $K$ is of continuous type in $L^2_{\mathbb R}$. 
This concludes for the continuity of $\phi$.

We now turn to the continuity of $d_t(u,w)$.  
From \citet[Section~4.1]{bornemann2010numerical}, it suffices to show that 
\begin{equation}
   (u,w,t) \mapsto \tsig_t(u,w)
\end{equation}
is continuous in the trace norm (see Section~\ref{sec:recall_trace_det}).  
We have the estimate
\begin{equation}
   \doublenorm{\tsig_t(u,w)}_1 
   \leq \doublenorm{(\id - b(u)K)^{-1}}_{L^2_{\mathbb C}}
        \doublenorm{(\id - b(u)K^*)^{-1}}_{L^2_{\mathbb C}}
        \doublenorm{\bo{\Sigma}_t}_1,
\end{equation}
and the same argument as above yields continuity in $(u,w)$.  
For the trace-norm continuity in $t$, let $h > 0$, and, since from Lemma \ref{lemma:tsig_psi_well_def} (applied with $u=0$), $\bo{\Sigma}_{t+h} - \bo{\Sigma}_t$ has a continuous kernel, we have from Mercer's theorem \citet[Theorem 8.1]{gohberg2012traces}, that 
\begin{equation}
   \Tr(\bo{\Sigma}_{t+h}-\bo{\Sigma}_t)
   = \int_0^T (\Sigma_{t+h}(s,s) - \Sigma_t(s,s)) ds.
\end{equation}
Moreover, observe that
\begin{equation}
   \Sigma_{t+h}(s,z) - \Sigma_t(s,z)
   = \nu^2 \int_t^{t+h} K(s,r)K(z,r)\,dr, \quad 0 \le s,z \le T,
\end{equation}
so that $\bo{\Sigma}_{t+h}-\bo{\Sigma}_t$ is a positive symmetric operator. Consequently,
\begin{align}
   \doublenorm{\bo{\Sigma}_{t+h}-\bo{\Sigma}_t}_1 
   &= \Tr(\bo{\Sigma}_{t+h}-\bo{\Sigma}_t) \\
   &= \nu^2 \int_t^{t+h} \int_0^T |K(s,r)|^2 ds\,dr 
   \xrightarrow{h \to 0} 0,
\end{align}
by the monotone convergence theorem, which concludes the proof.
\end{proof}

\begin{lemma} \label{lemma:no_rot_count_needed}
    Let $t\leq T$ and a connected subset $\mc V \subset \mc U$ containing two real numbers $(x,y)$ where $0 \leq x \leq 1$ and $0 \leq y$. Let $K$ be a Volterra kernel of continuous and bounded type in $L^2_{\mathbb R}$. If $d_t: \mc V \to \mathbb C$ never crosses the negative real axis, i.e.~$d_t(\mc V) \subset \mathbb C \backslash (-\infty, 0)$,  then, 
    
    \begin{equation}
        e^{\phi_t(u,w)} = \frac{1}{\sqrt{d_t(u,w)}}, \quad (u, w)\in \mc V. 
    \end{equation}
\end{lemma}

\begin{proof}
Let $t\leq T$ and $(u,w)\in \mc U$. Then, an application of Lemma~\ref{lemma:phi=sqrt_det} yields the existence of $k_t(u,w) \in \mathbb Z$ such that 
    \begin{align}\label{eq:equalityn}
    2 \pi k_t(u,w)  = 2 \Im(\phi_t(u,w)) + \arg(d_t(u,w))
    \end{align}
    and
    \begin{align}
        e^{-\phi_t(u,w) } = \sqrt{d_t(u,w)} e^{-i\pi k_t(u,w)}.
    \end{align}
    It remains to argue that $k_t(u,w)=0$. First, from Lemma \ref{lemma:phi_det_C^0}, $\phi_t$ and $d_t$ are continuous in $(u,w)$. Moreover, observe, that by assumption, since $d_t:\mc V\to \mathbb C$, never crosses the negative real axis,   we have the strict inequality 
    $$ -\pi <\arg (d_t(u,w)) <\pi, \quad (u,w) \in \mc V, $$
    showing that $\arg (d_t)$ is necessarily continuous in $(u,w) \in \mc V$.  Hence, from \eqref{eq:equalityn}, it follows that the function $k_t: \mc V \to \mathbb Z$ is continuous, and since $\mc V$ is connected, $k_t$ must be constant on $\mc V$. In other words, $k_t(u, w) = k_t(x, y)$ for all $(u, w) \in \mc V$. It remains to observe that $d_t(x,y)>0$ since $\tsig_t(x,y)$ is a positive operator, and that $\phi_t(x,y) \in \mathbb R$, to deduce again from \eqref{eq:equalityn} that $k_t(0,0)=0$, and hence $k_t(u,w) \equiv 0$ on $\mc V$. 
\end{proof}

Combining the above we can now conclude the proof of Theorem~\ref{thm:char_fun_sqrt_det_rot_count}.

\begin{proof}[Proof of Theorem~\ref{thm:char_fun_sqrt_det_rot_count}]
    Equation \eqref{eq:char_func_det_formula} follows directly by inserting \eqref{eq:phi=sqrt_det_rot_count} into \eqref{eq:char_func_trace_formula}. Subsequently, \eqref{eq:char_fun_sqrt_det} is obtained by applying Lemma \ref{lemma:no_rot_count_needed} to \eqref{eq:char_func_det_formula}.
\end{proof}

\section{Crossing behavior of the Fredholm determinant}\label{Sec:det_crossing}

The aim of this section is to provide an answer to the first two questions following  Theorem~\ref{thm:char_fun_sqrt_det_rot_count} regarding the crossing behavior.
It turns out that in the context of  computing the Fourier–Laplace transform of finite dimensional Wishart distributions,  the question has been partially studied for the determinant of matrices by \cite{mayerhofer2019reforming}. In fact, a formula similar to \eqref{eq:char_fun_sqrt_det} links the Fourier–Laplace transform of the Wishart distribution to the square root of the determinant of a certain matrix, see \citet[Equation (1.2)]{mayerhofer2019reforming}. The problem differs from ours, as it is static, in the sense that it concerns only one marginal at a given time $T$, and concerns finite dimensions, whereas we work with operators of possibly infinite rank and integrations in the time variable. However, the finite dimensional case is highly instructive and serves as an introduction to our answer to the crossing behavior question for the Fredholm determinants. 

In his article, \cite{mayerhofer2019reforming} showed by mean of examples, that when the rank of the matrix is less than or equal to 2, the determinant does not cross the negative real axis. However, when the rank is 3 or greater, the determinant may cross the negative real axis, leading to a discontinuous Fourier–Laplace transform, as the discontinuities of the complex square root are not handled by this formula. Therefore, even in this simpler case, the answer to the first question seems to be "yes". 

Finally, the paper of \cite{mayerhofer2019reforming} does not provide conditions under which the formula for the Fourier–Laplace transform is valid or not, and didn't provide a corrected formula as done by \eqref{eq:char_func_det_formula} in the context of operators. Instead, it presents an alternative formula, which corresponds exactly to the trace formula, see \citet[Theorem 1.1]{mayerhofer2019reforming}, which is numerically less suitable than the determinant formula as already explained.
 
Interestingly, the same rank condition appears in our operator case. To see this, it suffices to consider the spectral decomposition of the covariance kernel $\Sigma_t$ as illustrated in the next section. 

For additional details on compact operators and their rank, see Appendix \ref{sec:recall_trace_det}.

\subsection{The rank condition: an illustration} \label{subsec:rank_cond_illustration}

 To illustrate the rank condition, we set $g_0 \equiv 0$ and $\kappa = 0$, and examine the Fourier–Laplace transform of the integrated variance for $w \in \mathbb{C}$ with $\Re(w) \leq 0$, that is, we analyze $w \mapsto \mathbb{E} \left[ \exp \left(w \int_0^T X_s^2 \, ds \right) \right]$, where $X_t = \nu \int_0^T K(t,s) \, dW_s$. Therefore, $u$ is set to 0 and $b(u)=b=0 \in \mathbb R$ in \eqref{eq:tilde_Sigma}, so that $\tsig_0=\bo{\Sigma}_0$ is a self-adjoint operator. Moreover, recalling Lemma \ref{lemma:tsig_psi_well_def}, $\bo{\Sigma}_0$ is also an integral operator, so is compact. Thus, by the spectral theorem, $\bo{\Sigma}_0$ has a spectral decomposition:
\begin{equation}
    \Sigma_0(s,z) = \sum_{n=1}^M \lambda_n e_n(s) e_n(z), \quad 0 \leq s,z \leq T,
\end{equation}
where $M \in \{1, \dots, +\infty\}$ is the rank of $\bo{\Sigma}_0$, $(e_n)_{0 \leq n < \infty}$ is an orthonormal basis of $L^2([0,T], \mathbb{R})$ consisting of eigenvectors of $\bo{\Sigma}_0$, and $(\lambda_n)_{1 \leq n \leq M}$ are the corresponding eigenvalues, which satisfy, as $\bo{\Sigma}_0$ is a positive Hilbert-Schmidt operator,
\begin{align}
    \lambda_1 \geq \lambda_2 \geq \dots \geq 0, \quad     \sum_{n=1}^M \lambda_n^2 < \infty.
\end{align}
Since $\Sigma_0$ is the covariance kernel of $(X_t)_t$, the \citeauthor{karhunen1947under}-\citeauthor{loeve1978probability} expansion theorem yields
\begin{equation} \label{eq:X_t_Loeve_representation}
    X_t = \sum_{n=1}^M Y_n e_n(t), \quad 0 \leq t \leq T,
\end{equation}
where $Y_n := \int_0^T X_t e_n(t) \, dt$ defines a sequence of independent, zero-mean Gaussian random variables with respective variance $\lambda_n$. Additionally, from \eqref{eq:X_t_Loeve_representation} and Parseval's identity, we  obtain 
\begin{equation}
    \int_0^T X_s^2 \, ds = \sum_{n=1}^M Y_n^2.
\end{equation}
Finally, since the $(Y_n)_{1 \leq n \leq M}$ are independent, normally distributed with zero mean and variance $\lambda_n$, and that $\Re(w)<0$, we have by the monotone convergence theorem ($M$ is possibly infinite), that
\begin{equation} \label{eq:char_func_int_var_example_1}
    \mathbb{E} \left[ \exp \left(w \int_0^T X_s^2 \, ds \right) \right] = \frac{1}{\prod_{n=1}^M \sqrt{1 - 2w\lambda_n}}, \quad w \in \mathbb{C}, \quad  \Re(w) \leq 0,
\end{equation}
where $\sqrt{\cdot}$ is the principal branch of the square root, and where we used that $\mathbb{E} \left[ \exp \left(w Y_n^2 \right) \right]=\frac{1}{\sqrt{1 - 2w\lambda_n}}$. On the other hand, notice that for $\kappa = u = 0$, we have that  
\begin{align}\label{eq:detproduct}
 \sqrt{ \det(\bo{\Phi}_t(w))} = \sqrt{\det(\id - 2w\bo{\Sigma}_0)} = \sqrt{\prod_{n=1}^M (1 - 2w\lambda_n)},  
\end{align}
so that the denominator in \eqref{eq:char_func_int_var_example_1} would collapse into $\sqrt{ \det(\bo{\Phi}_t(w))}$ if and only if the product of the square roots of $(1-2w\lambda_n)_n$ is equal to the square root of the product.  In the complex plane this is not always the case: for complex numbers $z_1, \dots, z_M$, we have
\begin{equation}
    \sqrt{\prod_{n=1}^M z_n} = e^{ik\pi} \prod_{n=1}^M \sqrt{z_n},
\end{equation}
where $k = \frac{1}{2\pi} \left( \arg \left( \prod_{n=1}^M z_n \right) - \sum_{n=1}^M \arg(z_n) \right)$ is an integer. In particular, 
\begin{align}
    -\pi < \sum_{n=1}^M \arg(z_n) \leq \pi \Rightarrow  \sqrt{\prod_{n=1}^M z_n} = \prod_{n=1}^M \sqrt{z_n}.
\end{align}
Combining the above, we arrive at the following rank condition:

\begin{itemize}
    \item \textbf{Rank $M \leq 2$:}
    since $1 - 2\Re(w)\lambda_n > 0$, we have $\arg \left( \sqrt{1 - 2w\lambda_n} \right) \in \left(-\frac{\pi}{2}, \frac{\pi}{2}\right)$. This implies that when the rank $M$ of $\bo{\Sigma}_0$ is less than 2, $\sum_{n=1}^M \arg \left( \sqrt{1 - 2w\lambda_n} \right) \in (-\pi, \pi)$, so that
    \begin{equation}
        \prod_{n=1}^M \sqrt{1 - 2w\lambda_n} = \sqrt{\prod_{n=1}^M (1 - 2w\lambda_n)}
    \end{equation}
    and it follows from \eqref{eq:char_func_int_var_example_1} and  \eqref{eq:detproduct} that  for $M\leq 2$ and for all $w\in \mathbb C$ such that $\Re(w)\leq 0$:
    \begin{equation} 
     \mathbb{E} \left[ \exp \left(w \int_0^T X_s^2 \, ds \right) \right] =  \frac{1}{\sqrt{\det(\bo{\Phi}_t(w))}}.
    \end{equation}
    \item \textbf{Rank  $M \geq 3$:} we have, for $w$ sufficiently big, that $3\pi > \left|\sum_{n=1}^M \sqrt{1 - 2w\lambda_n}\right| > \pi$, resulting in $e^{ik(w)\pi}=-1$, that is
    $$
    \prod_{n=1}^M \sqrt{1 - 2w\lambda_n} = e^{-ik(w)\pi} \sqrt{\prod_{n=1}^M (1 - 2w\lambda_n)} = -\sqrt{\prod_{n=1}^M (1 - 2w\lambda_n)},
    $$
    which combined with equation \eqref{eq:char_func_int_var_example_1} yields that 
    \begin{equation} \label{eq:char_func_int_var_example_2}
        \mathbb{E} \left[ \exp \left(w \int_0^T X_s^2 \, ds \right) \right] = - \frac{1}{\sqrt{\det(\bo{\Phi}_t(w))}} \neq \frac{1}{\sqrt{\det(\bo{\Phi}_t(w))}},
    \end{equation}
    which is precisely formula \eqref{eq:char_func_det_formula} of Theorem \ref{thm:char_fun_sqrt_det_rot_count}, provided we can show that $k(w) = n_0(w)$. This shows that when the rank of $\bo{\Sigma}_0$ is greater than 2, then formula \eqref{eq:char_fun_sqrt_det} isn't valid, meaning that the determinant has crossed the negative real axis.
\end{itemize}

Note that Formula~\eqref{eq:char_func_int_var_example_1} can be rewritten as
\begin{equation}\label{eq:intvariancedetsqrt}
    \mathbb{E} \left[ \exp \left(w \int_0^T X_s^2 \, ds \right) \right] = \frac{1}{\det\left(\sqrt{\id - 2w\bo{\Sigma}_0}\right)},
\end{equation}
This formulation removes the sign ambiguity associated with the $\pm1$ prefactor, by taking the determinant of the matrix square root instead of the square root of the determinant.
While this manipulation is straightforward in the integrated variance case, the general proof is considerably more delicate, since one cannot, in general, rely on a spectral decomposition as for $\bo{\Sigma}_0$.
In finite dimension, after discretizing the operators by matrices, we are nevertheless able to obtain such a prefactor-free formula. This is presented in Section~\ref{Sec:hybrid_lip_algos}, as already mentioned in Introduction and Section~\ref{subsec:det_formula}.

\subsection{
Main results on the crossing behavior} \label{subsec:condition_on_crossings}

Now that we have developed our intuition for the rank condition, we can answer the crossing behavior questions. In Theorem \ref{thm:crossing_neg_axis_integrated_variance} and \ref{thm:crossing_neg_axis_log_price}, we characterize the crossings of the negative real axis $\mathbb{R}_-$ by the determinant, or equivalently, when formula \eqref{eq:char_fun_sqrt_det} is not valid and must be replaced by \eqref{eq:char_func_det_formula}. We also give a tight upper bound for the first crossing moment. As already mentioned, referring to formula \eqref{eq:k_t(u,w)}, the crossings of the determinant are totally related to the value of $e^{i\pi k_t}=\pm 1$. Therefore, all properties will be expressed in terms of the value of $e^{i\pi k_t}$. 

For $0 \leq t \leq T$ and  $u \in \mathbb C$ with $\rho \Im(u) = 0$, the  operator $\tsig_t(u)$ is compact (since it is of trace class from Lemma \ref{lemma:tsig_psi_well_def}, see Appendix \ref{sec:recall_trace_det}), we can denote its eigenvalues by $(\lambda_{n,t}(u))_{n \geq 1}$, counting multiplicity, see Appendix \ref{sec:recall_trace_det} for details. Moreover, since it is symmetric positive, we can order them as 
\begin{equation}
    \lambda_{1,t}(u) \geq \lambda_{2,t}(u) \geq \dots \geq 0,
\end{equation}
where we set $\lambda_{n,t}(u) = 0$ for $n > N(\tsig_t(u))$. Thus, $\lambda_{n,t}(u) > 0$ for $n \leq N(\tsig_t(u))$, where $N(\tsig_t(u))$ denotes the rank of $\tsig_t(u)$, see Appendix \ref{sec:recall_trace_det}. 

The first theorem concerns the Fourier–Laplace transform in the \textit{integrated variance case} $\mathbb{E} \left[ \exp \left(w \int_t^T X_s^2 \, ds \mid \mathcal F_t \right) \right]$, i.e. we set $u=0$.

\begin{theorem}[\textbf{Integrated variance case}] \label{thm:crossing_neg_axis_integrated_variance}
    Let $0 \leq t \leq T$, and let $K$ be a Volterra kernel of continuous and bounded type in $L^2_{\mathbb R}$. Fix $u=0$ and $\Re(w) \leq 0$, so that $\tsig_t$ no longer depends on $u$. Then \begin{itemize}
        \item If $N(\tsig_t) \in \{1, 2\}$, then $e^{i\pi k_t(w)} = 1$ for all $\Im(w) \in \mathbb{R}$.
        \item If $N(\tsig_t) \geq 3$, then $e^{i\pi k_t(w)} = -1$ for $\Im(w)$ in a union of disjoint intervals with strictly positive length. If $N(\tsig_t) = +\infty$, this union is infinite and unbounded. Moreover, define \hbox{$\Im^*(w) := \inf \{\Im(w)>0, \, e^{i\pi k_t(w)}=-1\}$}, and for $r>0$, define $N_r$ as the number of eigenvalues of $\tsig_t$ greater than $r$. Then, if $N_r \geq 3$, we have
        \begin{equation} \label{ineq:first_crossing_int_variance}
            \Im^*(w) \leq \tan\left(\frac{\pi}{N_r}\right)\left(\frac{1}{2r}-\Re(w)\right).
        \end{equation}
    \end{itemize}
\end{theorem}

\begin{proof}
    The proof is given in Section \ref{sec:proof_thm_crossing}.
\end{proof}

The next theorem addresses the \textit{log-price} case, $ \mathbb{E} \left[ \exp \left(u \log \frac{S_T}{S_t} \right) \mid \mathcal F_t \right]$, i.e. we set $w=0$, which is more intricate to analyze because, when $\rho \neq 0$, $\tsig_t$ depends on $u$. Therefore, we first establish the properties in the case $\rho = 0$ and then extend them, in Corollary \ref{cor:crossing_neg_axis_log_price}, to a continuous range $\rho \in [-\epsilon, \epsilon]$ using a continuity argument. 

\begin{remark}
    Writing, from the dynamics of $(S_t)_{t \geq 0}$ and the fact that we fixed $\rho = 0$, that
    \begin{equation}
        \log\left(\frac{S_T}{S_t}\right) = -\frac{1}{2}\int_t^T X_s^2 ds + \int_t^T X_s dW_s^{\perp},
    \end{equation}
    and using the tower property by projecting onto the trajectory of the volatility process $X$, one can show that
    \begin{equation} \label{eq:fourier_laplace_log_price_alternative_form}
        \mathbb{E} \left[ \exp \left( u \log \frac{S_T}{S_t}\right) \bigg| \mc{F}_t \right] = \mathbb{E} \left[ \exp \left( \frac{u^2-u}{2} \int_t^T X_s^2 ds \right) \bigg| \mc{F}_t \right].
    \end{equation}
    Noticing that \eqref{eq:fourier_laplace_log_price_alternative_form} corresponds to the Fourier–Laplace transform in the integrated variance case when setting $w=\frac{u^2-u}{2}$, one could argue that since $\Im(w)=\Im\left(\frac{u^2-u}{2}\right) = \Im(u)(\Re(u)-\frac{1}{2})$ is linear in $\Im(u)$, the behavior of the prefactor $\Im(u) \mapsto e^{i\pi k_t(u)}$ can be deduced from Theorem \ref{thm:crossing_neg_axis_integrated_variance}, and will be the same as in the integrated variance case. However, notice that in this case, $\Re(w)=\frac{1}{2}(\Re(u)(\Re(u)-1) -\Im(u)^2)$ depends also on $\Im(u)$. This introduces a significant difference compared to the framework of Theorem \ref{thm:crossing_neg_axis_integrated_variance}, where $\Re(w)$ is fixed. Therefore, the analysis of the behavior of the prefactor must be conducted specifically for the log-price case. In particular, a condition on the volatility of volatility $\nu$ will appear.
\end{remark}

\begin{theorem}[\textbf{Uncorrelated log-price case}] \label{thm:crossing_neg_axis_log_price}
    Let $0 \leq t \leq T$ and $K$ be a Volterra kernel of continuous and bounded type in $L^2_{\mathbb R}$. Fix $w=0$ and $0 \leq \Re(u) \leq 1$. Suppose $\rho = 0$, so that $\tsig_t$ is independent of $u$. Then there exists $\nu^* >0$ such that: 
    \begin{itemize}
        \item If $N(\tsig_t) \in \{1, 2\}$ or $\Re(u) = \frac{1}{2}$ or $\nu \leq \nu^*$, then $e^{i\pi k_t(u)} = 1$ for all $\Im(u) \in \mathbb{R}$.
        \item If $N(\tsig_t) \geq 3$ and $\Re(u) \neq \frac{1}{2}$ and $\nu > \nu^*$, then $e^{i\pi k_t(u)} = -1$ for $\Im(u)$ in a finite union of disjoint intervals with strictly positive length. 
        
        Moreover, define \hbox{$\Im^*(u) := \inf \{\Im(u)>0, e^{i\pi k_t(u)}=-1\}$}, and for $r>0$, define $N_r$ as the number of eigenvalues of $\tsig_t$ greater than $r$. Then, if $N_r \geq 3$, we have
        \begin{equation}
            \Im^*(u) \leq \frac{|1-2\Re(u)|}{2\tan\left(\frac{\pi}{N_r}\right)}-\sqrt{\left(\frac{1-2\Re(u)}{2\tan\left(\frac{\pi}{N_r}\right)}\right)^2 -4\left(\frac{1}{r} + \Re(u)(1-\Re(u))\right)}.
        \end{equation}
    \end{itemize}
\end{theorem}

\begin{proof}
    The proof is given in Section \ref{sec:proof_thm_crossing}.
\end{proof}

The following corollary extends Theorem \ref{thm:crossing_neg_axis_log_price} to a continuum of correlation. In this case, and only this case, we will state the dependence of $\tsig_t$ on $\rho$.
\begin{corollary}[\textbf{Correlated log-price case}] \label{cor:crossing_neg_axis_log_price}
    Let $0 \leq t \leq T$ and $K$ be a Volterra kernel of continuous and bounded type in $L^2_{\mathbb R}$. Fix $w=0$ and $0 \leq \Re(u) \leq 1$. Then there exists $\nu^* >0$ such that, for any $\delta>0, \; M>0$, there exists $0 < \epsilon \leq 1$ such that for $|\rho| \leq \varepsilon$: \begin{itemize}
        \item If $N(\tsig_t(\rho = 0)) \in \{1, 2\}$ or $\Re(u) = \frac{1}{2}$ or $\nu \leq \nu^*$, then $e^{i\pi k_t(u)} = 1$ for any $|\Im(u)| \leq M$.
        \item If $N(\tsig_t(\rho = 0)) \geq 3$ and $\Re(u) \neq \frac{1}{2}$ and $\nu > \nu^*$, then $e^{i\pi k_t(u)} = -1$ for any $\Im(u)$ in a finite union of disjoint intervals with strictly positive length. 
        
        Moreover, define \hbox{$\Im^*(u) := \inf \{\Im(u)>0, e^{i\pi k_t(u)}=-1\}$}, and for $r>0$, define $N_r$ as the number of eigenvalues of $\tsig_t(\rho=0)$ greater than $r$. Then
        \begin{equation} \label{ineq:first_crossing_log_price}
            \Im^*(u) \leq \frac{|1-2\Re(u)|}{2\tan\left(\frac{\pi}{N_r}\right)}-\sqrt{\left(\frac{1-2\Re(u)}{2\tan\left(\frac{\pi}{N_r}\right)}\right)^2 -4\left(\frac{1}{r} + \Re(u)(1-\Re(u))\right)} + \delta.
        \end{equation}
    \end{itemize}
\end{corollary}

\begin{proof}
    The proof is given in Section \ref{sec:proof_thm_crossing}.
\end{proof}

\begin{remark}
    Note that if one wants to price a call option on an asset uncorrelated with its volatility, i.e. $\rho =0$, using the Lewis formula \eqref{eq:Lewis_formula}, then, from case 1 of Theorem \ref{thm:crossing_neg_axis_log_price}, since $\Re(u) = \frac{1}{2}$, no correction needs to be applied to the determinant.
\end{remark}

In the \textit{Integrated variance} case, we observe that when the rank of $\tsig_t$ is less than or equal to 2, the formula \eqref{eq:char_fun_sqrt_det} is valid; however, it fails to hold when the rank is greater than or equal to 3, for $\Im(w)$ in a disjoint union of intervals, which is also infinite and unbounded when the rank of $\tsig_t$ is infinite. 

In the \textit{log-price} case, the same result applies, at least for a continuum of correlation, when the rank of $\tsig_t(\rho=0)$ is less than or equal to 2, or when the vol-of-vol does not exceed a certain threshold (which may be infinite) that depends solely on the kernel $K$ and the model parameters, except for the vol-of-vol itself. However, in the opposite case, formula \eqref{eq:char_fun_sqrt_det} does not hold for $\Im(u)$ in a finite union of disjoint intervals. 

Moreover, the upper bounds for the first crossing of the determinant, provided in the integrated variance case by \eqref{ineq:first_crossing_int_variance} and in the log-price case by \eqref{ineq:first_crossing_log_price}, are strictly decreasing in $N_r$. This implies that the larger the eigenvalues of $\tsig_t$, the faster the crossing of the determinant occurs. \\

It is important to understand how the model parameters influence the eigenvalues and, consequently, the first crossing instant. 

Since $\tilde{\bo{\Sigma}}_t$ (defined in \eqref{eq:tilde_Sigma}) is proportional to $\nu^2$, its eigenvalues increase with $\nu$, causing the first crossing to occur earlier. This behavior is illustrated in Figure \ref{fig:first_crossing_nu}, which also shows that regardless of the value of $\nu$, the first crossing always happens before the Fourier–Laplace transform has sufficiently decayed, making the prefactor $e^{i\pi k_t(w)}$ of utmost importance, as explained in the last paragraph of Section \ref{subsec:det_formula}. The dotted curve in the figure represents the Fourier–Laplace transform obtained without applying the prefactor $e^{i\pi k_t(w)}$.

\begin{figure}[h!]
    \centering
    \includegraphics[width=1\textwidth]{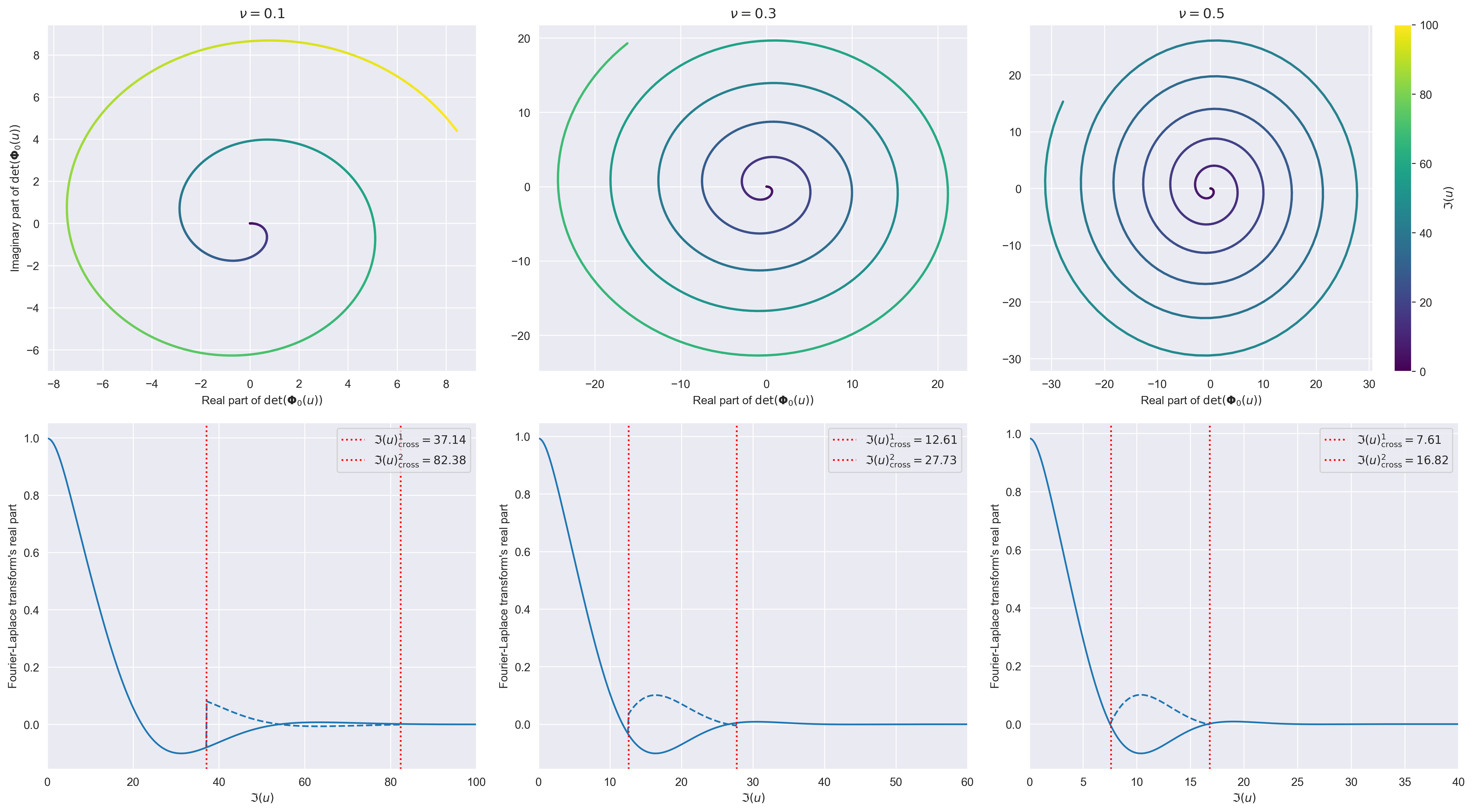} 
    \caption{\small 
        Fredholm determinant of $\bo{\Phi}_0$ as a function of $\Im(u)$ (Top), and comparison of the first and second crossing moments for different vol-of-vol values (Bottom), in the fractional kernel setting and log-price case.
        The parameters are $\kappa = 0$, $\theta = -0.1$, $\rho = -0.9$, $X_0 = -0.05$, $H = 0.3$, $T = 1$, $n=200$, $\Re(u) = 0.7$, $w = 0$.
    }
    \label{fig:first_crossing_nu}
\end{figure}

In the same spirit as in the proof of Lemma \ref{lemma:majo_lambda_n}, we can prove, at least when $\rho \Im(u) = 0$, which implies that $\tsig_t$ is a symmetric positive operator, that the eigenvalues increase with the maturity $T$. Moreover, for the Riemann-Liouville fractional kernel, the same idea shows that when $T \leq 1$, the eigenvalues increase as $H \in (0, \frac{1}{2}]$ decreases. In practice, we have observed that this pattern also holds when $\rho \Im(u) \neq 0$. These two properties are illustrated in Figure \ref{fig:eigenvalues_T_H}.
\begin{figure}[h!]
    \centering
    \includegraphics[width=1\textwidth]{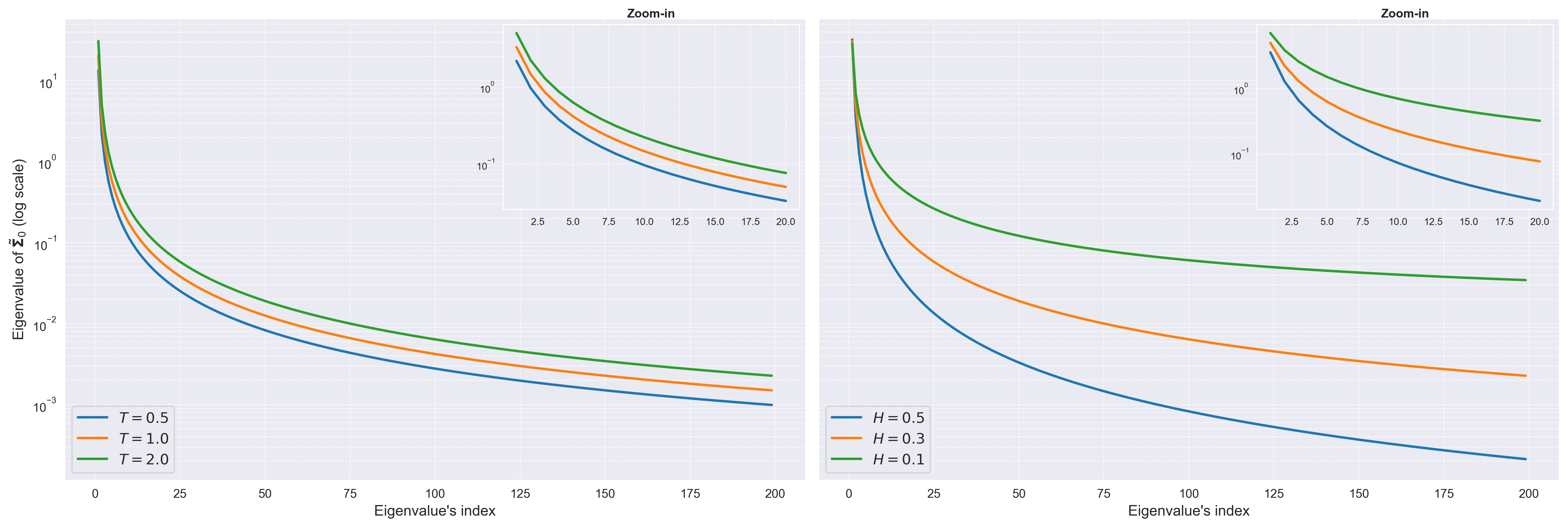} 
    \caption{\small 
        Comparison of eigenvalues of $\tsig_0$ for different maturity and Hurst index, in the fractional kernel setting.
        The parameters are $\kappa = 0$, $\nu = 0.2$, $\theta = 0.1$, $\rho = 0$, $X_0 = 0.1$, $H = 0.3$ (Left), $T = 1$ (Right), $n=1000$.
    }
    \label{fig:eigenvalues_T_H}
\end{figure}
In Figure \ref{fig:arg_of_det_H_T}, we illustrate the impact of $T$ and $H$ on the first crossing moment, in the log-price case. Recalling that a crossing corresponds to the argument of the determinant being equal to $\pm \pi \approx \pm 3.14$, the figure shows that as $T$ increases or $H$ decreases, crossings occur more frequently.
\begin{figure}[h!]
    \centering
    \includegraphics[width=0.8\textwidth]{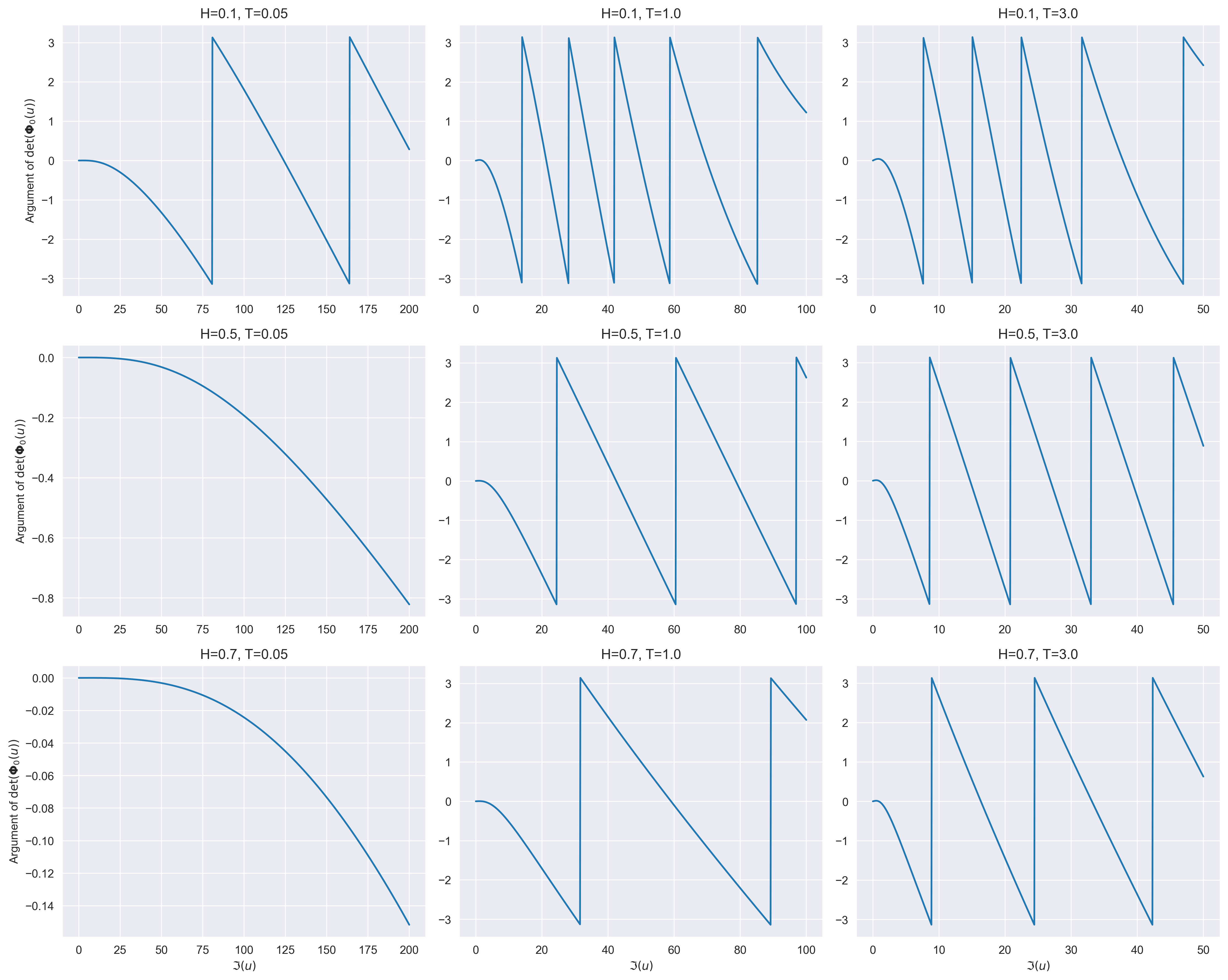} 
    \caption{\small Argument of the Fredholm determinant of $\bo{\Phi}_0$ as a function of $\Im(u)$, for short to long maturities and Hurst indices, in the fractional kernel setting. The model parameters are $\kappa = 0$, $\nu = 0.2$, $\theta = -0.1$, $\rho =-0.9$, $X_0 = -0.05$, $n = 200$, $\Re(u) = 0.33$.}
    \label{fig:arg_of_det_H_T}
\end{figure}

In Figure \ref{fig:first_crossing_eigen_values}, in the integrated variance case, we show that as the eigenvalues increase (or as $H$ decreases), the first crossing occurs earlier. Additionally, regardless of the value of $H$, the first crossing always happens before the Fourier–Laplace transform has sufficiently decayed. These observations reinforce the results of Theorems \ref{thm:crossing_neg_axis_integrated_variance} and \ref{thm:crossing_neg_axis_log_price}.
\begin{figure}[h!]
    \centering    \includegraphics[width=1\textwidth]{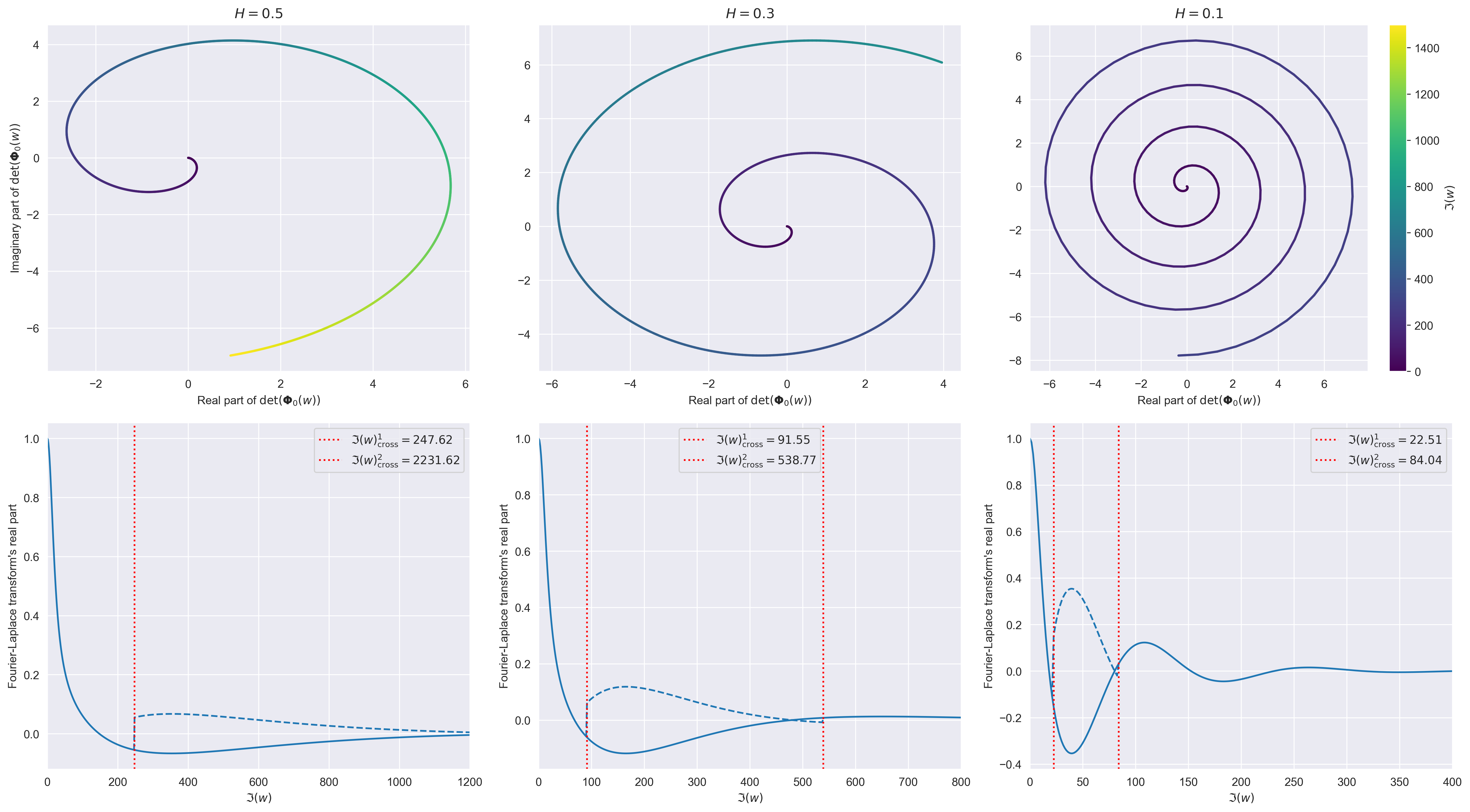} 
    \caption{\small Fredholm determinant of $\bo{\Phi}_0$ as a function of $\Im(w)$ (Top), and comparison of the first and second crossing moment (Bottom), for different Hurst indices, in the fractional kernel setting and integrated variance case. 
    The parameters are $\kappa = 0$, $\theta = 0.1$, $\nu = 0.2$, $\rho = 0$, $X_0 = 0.1$, $H = 0.3$, $T = 1$, $n=200$, $u = 0$, $\Re(w) = 0$.}
    \label{fig:first_crossing_eigen_values}
\end{figure}

Finally, while we were unable to theoretically determine the impact of the correlation $\rho$ on the eigenvalues, we observed in practice that as $|\rho|$ increases, the first crossing occurs earlier, and crossings happen more frequently. This behavior is illustrated in Figure \ref{fig:arg_of_det_rho}. 
\begin{figure}[h!]
    \centering
    \includegraphics[width=0.8\textwidth]{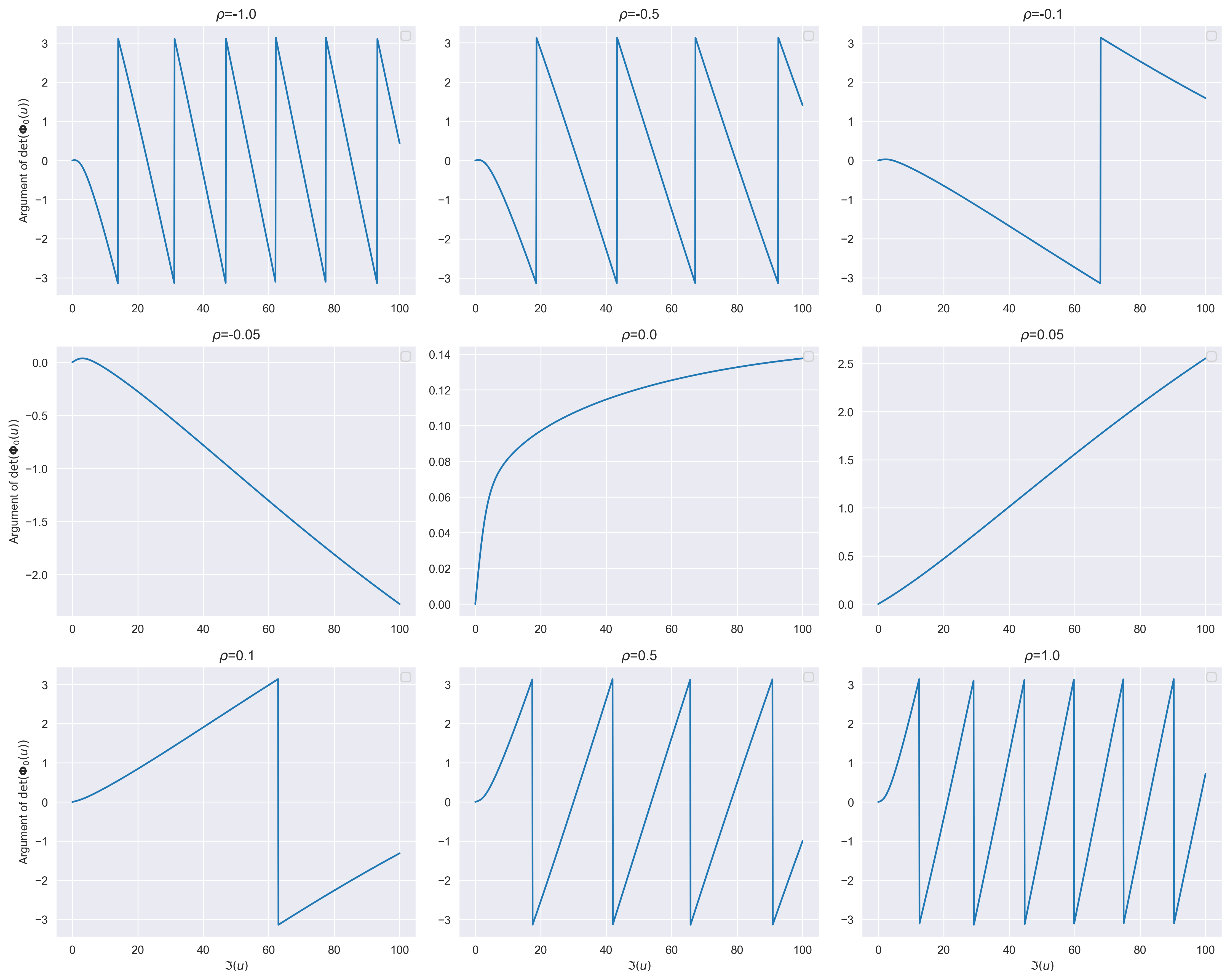} 
    \caption{\small Argument of the Fredholm determinant of $\bo{\Phi}_0$ as a function of $\Im(u)$, for correlations going from -1 to 1, in the fractional kernel setting. The parameters are $\kappa = 0$, $\nu = 0.2$, $\theta = -0.1$, $\rho = 0$, $X_0 = -0.05$, $H = 0.3$, $T = 1.5$, $n=200$, $\Re(u) = 0.33$, $w=0$.}
    \label{fig:arg_of_det_rho}
\end{figure}
To sum up, we collect in  Table \ref{tab:parameter_influence}  the parameters influence on the eigenvalues and first crossing instant. \\
\begin{table}[h!]
\centering
\caption{\small Influence of parameters on eigenvalues and the first crossing instant}
\label{tab:parameter_influence}
\renewcommand{\arraystretch}{1.3} % Adjust row height for better readability
\footnotesize % Slightly reduce font size
\begin{tabular}{|c|c|c|c|c|}
\hline
& \(\nu \nearrow\) & \(T \nearrow\) & \(|\rho| \nearrow\) & \(H \searrow\) \\ \hline
\textbf{Eigenvalues} & 
Increase & 
Increase & 
No theoretical result & 
Increase \\ \hline
\textbf{First crossing instant} & 
Decrease & 
Decrease & 
No theoretical result, observations: decrease & 
Decrease \\ \hline
\end{tabular}
\end{table}

Now that the first questions regarding the crossing behavior have been answered, we must address the last question, namely, whether we can practically and efficiently compute  the value of $k_t$ or $e^{i\pi k_t}$. We provide two algorithms in the next section, as well as a prefactor-free approximation formula for the Fourier–Laplace transform.

\section{Numerical computation of the Fourier–Laplace transform }\label{Sec:hybrid_lip_algos}

We recall the determinant formula \eqref{eq:char_func_det_formula} for the Fourier–Laplace transform:
\begin{equation}
    \xi_t(u,w) := \mathbb{E} \left[ \exp \left( u \log \frac{S_T}{S_t} + w \int_t^T X_s^2 \, ds \right) \bigg| \mc{F}_t \right] 
    = e^{i \pi k_t(u,w)} \frac{\exp \left(\langle g_t, \bo{\Psi}_t(u,w) g_t \rangle_{L^2_{\mathbb R}} \right)}{\sqrt{\det(\bo{\Phi}_t(u,w))}},
\end{equation}
where $e^{i\pi k_t(u,w)} = \pm 1$. 

When performing Fourier pricing, $t$, $\Re(u)$ and $\Re(w)$ are fixed, and one evaluates $\xi_t$ on a grid of $\Im(u)$ or $\Im(w)$ defined by a quadrature, depending on the derivative's nature (log-price or integrated variance case). Therefore, for the sake of readability and without loss of generality, we will only consider the log-price case (i.e. we set $w=0$), and we will identify $u$ to $\Im(u)$.

In order to compute the Fourier–Laplace transform using the determinant formula \eqref{eq:char_func_det_formula}, one first evaluates
\begin{equation} \label{eq:char_func_first_part}
    \xi_{\text{sign}}(u) := e^{-i\pi k_t(u)}\xi_t(u),
\end{equation}
over a grid of values for $u$, using, for instance, the numerical method provided in Section \ref{subsec:discretized_operators}. Finally, it remains to compute the prefactor $e^{i\pi k_t(u)}$ over the grid.

Observe again that once \eqref{eq:char_func_first_part} is computed, the Fourier–Laplace transform over the grid is known up to a factor of $\pm 1$. Moreover, recalling formula \eqref{eq:k_t(u,w)}, it follows that $k_t(u)$ represents the net number of times $\det(\bo{\Phi}_t)$ crosses the negative real axis anticlockwise between $0$ and $u$.

Two approaches arise from these observations: One can either determine the value of $e^{i\pi k_t(u)}$, i.e., find the correct sign of the Fourier–Laplace transform, by identifying where formula \eqref{eq:char_func_first_part} should be multiplied by $-1$, or one can determine the value of $k_t(u)$, by spotting the discontinuities of the argument of $\det(\bo{\Phi}_t)$.

In this section, we propose two methods, based on these two approaches, for computing the Fourier–Laplace transform over a grid of $u$ defined by a quadrature.

\begin{enumerate}
    \item \textbf{The Hybrid Trace/Determinant algorithm}  detailed in Section \ref{subsec:approach_based_on_char_func_proxy} follows the first approach and aims to recover the correct sign for the Fourier–Laplace transform by using a rough approximation of it, and compare its sign to the one obtained with \eqref{eq:char_func_first_part}. This approximation is computed using the trace formula \eqref{eq:char_func_trace_formula}, with a weak approximation of the operators.
    \item \textbf{The Lipschitz-based algorithm}   detailed in Section \ref{subsec:approach_based_on_lip_cst} and based on the second approach, leverages the Lipschitz constant of the determinant's argument to identify where the crossing points of the determinant are. 
\end{enumerate}
 
The first method is very fast if a rough approximation is sufficient to determine the correct sign. If not, a more precise approximation is required, which may adversely affect computation speed. The second method is also very fast when the argument of the determinant varies slowly. It has been used, without considering the Lipschitz constant, in \cite{abi2022characteristic}. 

Furthermore, we propose in Section \ref{subsec:prefactor-free_approx} a prefactor-free approximation formula for $\xi_t(u,w)$.
Indeed, although $\bo{\Phi}_t(u,w)$ is in general not Hermitian, one can nevertheless give a consistent meaning to
\begin{equation}
``\det\left(\sqrt{\bo{\Phi}_t(u,w)}\right)",
\end{equation}
and show that
\begin{equation}
\frac{e^{i \pi k_t(u,w)}}{\sqrt{\det(\bo{\Phi}_t(u,w))}} = \frac{1}{\det\left(\sqrt{\bo{\Phi}_t(u,w)}\right)}.
\end{equation}
However, as detailed in Section~\ref{subsec:rank_cond_illustration}, we were only able to establish this identity in finite dimension, after approximating the operators by matrices as in Section~\ref{subsec:discretized_operators}.

A numerical comparison between the two algorithms and the prefactor-free formula is presented in Section~\ref{sec:numerical_results}. First, in the next section, we recall the method from \citet[Section~4]{abi2022characteristic} to approximate the trace and determinant formulas using closed-form expressions derived from a discretization of the operators.

\subsection{Approximation by closed form expressions} \label{subsec:discretized_operators}

To ease notations, we drop the dependence of operators on $u$ and $w$. We fix a number $n \in \mathbb N^*$ of discretization steps, and $t_i = i\frac{T}{n}, \; i \in \{0, 1,\dots, n\}$ a partition of $[0,T]$. Discretizing the $\star$-product \eqref{eq:star_product} leads to the following discretization for $\bo{\Psi}_{t_i}$:
\begin{equation} \label{eq:Psi_disctretized}
     \Psi_{n,i} = a\left(I_n -b K_n^{\top}\right)^{-1}\left(I_n - 2a\frac{T}{n}\tilde{\Sigma}_{n,i}\right)^{-1} \left(I_n -b K_n\right)^{-1}
\end{equation}
where $I_n$ is the $n \times n$ identity matrix, $K_n$ is the lower triangular with components
\begin{equation}
 (K_n)_{jk} = \mathbbm{1}_{k \leq j-1} \int_{t_{k-1}}^{t_k} K(t_{j-1}, s) \, ds
\end{equation}
and 
\begin{equation}
    \tilde{\Sigma}_{n,i} = \left(I_n -b K_n\right)^{-1} \Sigma_{n,i} \left(I_n -bK_n^{\top}\right)^{-1}
\end{equation}
with $\Sigma_{n,i}$ the $n \times n$ discretized covariance matrix $\bo{\Sigma}_{t_i}$, recall \eqref{eq:Sigma}, given by
\begin{equation}
    (\Sigma_{n,i})_{jk} = \nu^2 \mathbbm{1}_{i < (j \wedge k) - 1} \int_{t_i}^{t_{j-1} \wedge t_{k-1}} K(t_{j-1}, s) K(t_{k-1}, s) \, ds.
\end{equation}
For notational convenience, we write $\Psi_n$, $\tilde{\Sigma}_n$, and $\Sigma_n$ instead of $\Psi_{n,0}$, $\tilde{\Sigma}_{n,0}$, and $\Sigma_{n,0}$ when the index $i=0$ is considered.

Defining the n-dimensional vector $g_n := (g_0(t_0), \dots, g_0(t_{n-1})^{\top})$, the discretization of the inner product $\langle \cdot, \cdot \rangle_{L^2_{\mathbb R}}$ leads to the approximation of the determinant formula \eqref{eq:char_fun_sqrt_det}
\begin{equation} \label{eq:det_formula_discretized}
    \mathbb{E} \left[ \exp \left( u \frac{\log S_T}{\log S_0} + w \int_0^T X_s^2 ds \right)\right] \approx e^{i\pi k_n} \frac{\exp \left(\frac{T}{n} g_n^{\top} \Psi_n g_n\right)}{\sqrt{\det(\Phi_n)}}, 
\end{equation}
with $\Phi_n = I_n -2a \frac{T}{n} \tilde{\Sigma}_n$, and $k_n$ is obtained from Algorithm \ref{algo:compute_phi_t} or \ref{algo:rot_count_k_t}. 

Combining the discretization of $\bo{\Psi}_{t_i}$ and a Riemann sum, we obtain the following approximation for $\phi_0$, recall \eqref{eq:phi_t}:
\begin{equation} \label{eq:phi_disctretized}
    \phi_n = -\frac{T}{n} \sum_{i=0}^{n-1} \Tr\left(\Psi_{n,i} \dot{\Sigma}_{n,i} \right)
\end{equation}
with $\dot{\Sigma}_{n,i}$ being the $n \times n$ matrix with components
\begin{equation}
    (\dot{\Sigma}_{n,i})_{jk} = -\nu^2 \mathbbm{1}_{i \leq (j \wedge k) - 1} K(t_{j-1}, t_i) \int_{t_{k-1}}^{t_k} K(s, t_i) \, ds.
\end{equation}
This leads to the following approximation for the trace formula \eqref{eq:char_func_trace_formula}

\begin{equation} \label{eq:trace_formula_discretized}
    \mathbb{E} \left[ \exp \left( u \frac{\log S_T}{\log S_0} + w \int_0^T X_s^2 ds \right)\right] \approx \exp \left(\phi_n + \frac{T}{n} g_n^{\top} \Psi_n g_n\right).
\end{equation}

\subsection{Hybrid Trace/Determinant algorithm for sign adjustment} \label{subsec:approach_based_on_char_func_proxy}
The idea behind the hybrid Trace/Determinant  method is simple and leverages the trace formula \eqref{eq:char_func_trace_formula} as a proxy to identify the correct sign for \eqref{eq:char_func_first_part}. 

As illustrated in Figure \ref{fig:call_price_computation_time_all_algorithms}, computing the Fourier–Laplace transform with high precision using the trace formula \eqref{eq:char_func_trace_formula} is computationally expensive. However, this figure also shows that for a small number of discretization steps, computation times become very low. While the accuracy of the trace formula with limited discretization steps is poor, we can reasonably expect that even a coarse approximation will share the same sign as the true value. Therefore, determining the correct sign can be achieved by using a low-precision estimation of the Fourier–Laplace transform obtained from the trace formula. 

Specifically, this low-precision approximation is computed using a coarse discretization of the operators as per \eqref{eq:trace_formula_discretized}, with $n = n_{\text{coarse}}$ a small integer. This low-precision approximation acts as a fast proxy to determine the correct sign. Finally, the sign obtained from the determinant formula  \eqref{eq:char_func_first_part} along the grid is compared to the sign of this proxy, and is adjusted if the two signs differ, thereby recovering the correct value of the Fourier–Laplace transform. The additional computation time induced by this proxy will be negligible compared to the cost of computing \eqref{eq:char_func_first_part}. The method is detailed in Algorithm \ref{algo:compute_phi_t}, and a sanity check is provided in Figure \ref{fig:call_price_precision_low_T}, as well as computation times in Figure \ref{fig:call_price_computation_time_all_algorithms}.

\begin{algorithm} \small
\caption{Computation of $\xi_t(u)$ for a node $u$ of a quadrature}
\label{algo:compute_phi_t}
\begin{algorithmic}[1]
\Input Number of discretization steps $n_{\text{coarse}}$ for the coarse approximation of $\xi_t(u)$, quadrature node $u \geq 0$.
\Output Value of $\xi_t(u)$.

\State Compute $\xi_{\text{sign}}(u) = \frac{\exp \left( \langle g_t, \Psi_t(u) g_t \rangle_{L_{\mathbb R}^2} \right)}{\sqrt{\det(\bo{\Phi}_t(u))}}$ using the estimate \eqref{eq:det_formula_discretized}.
\State Compute $\xi_{\text{coarse}}(u)$ using the estimate \eqref{eq:trace_formula_discretized} with $n = n_{\text{coarse}}$.
\If{$\Re(\xi_{\text{sign}}(u)) \Re(\xi_{\text{coarse}}(u)) < 0$ \textbf{and} $\Im(\xi_{\text{sign}}(u)) \Im(\xi_{\text{coarse}}(u)) < 0$}
    \State $\xi_t(u) \gets -\xi_{\text{sign}}(u)$
\Else
    \State $\xi_t(u) \gets \xi_{\text{sign}}(u)$
\EndIf

\Return $\xi_t(u)$
\end{algorithmic}
\end{algorithm}

\subsection{Lipschitz-based sign adjustment} \label{subsec:approach_based_on_lip_cst}

In this section, for the sake of readability, we define $r_t(u) := e^{-2\Re(\phi_t(u))}$ and $\theta_t(u) := -2\Im(\phi_t(u))$ so that, recalling \eqref{eq:det_exp_phi}, we have 
\begin{equation} \label{eq:det_phi_t_polar_phi}
    \det(\bo{\Phi}_t(u)) = r_t(u)e^{i\theta_t(u)},
\end{equation}
where $\theta_t$ is continuous from Lemma \ref{lemma:phi_det_C^0}, and is an argument of $\det(\bo{\Phi}_t)$. 

Given a quadrature whose upper bound is $U > 0$, we aim to compute $k_t(u)$ for each node of the quadrature. It is clear that the value of $k_t(u)$ depends strongly on the Lipschitz constant of $\theta_t$ over $[0, U]$. 

Therefore, to compute its value, we should evaluate $\arg(\det(\bo{\Phi}_t))$ on a sufficiently fine grid, determined by the Lipschitz constant of $\theta_t$, denoted by $0 = u_0 < u_1 < \dots < u_n = U$. From these evaluations, we determine the net number of crossings $k_t(u_i)$ for all $0 \leq i \leq n$. Note that this step is independent of quadrature nodes, meaning that it can be done once for all as soon as the model parameters are fixed. Finally, for any node $u \in [0, U]$ of the quadrature, the value of $k_t(u)$ can be interpolated within this precomputed grid. The details are given in Algorithm \ref{algo:rot_count_k_t}.

\begin{algorithm}[h!] \small
\caption{Computation of $k_t(u)$ for a node $u$ of a quadrature}
\label{algo:rot_count_k_t}
\begin{algorithmic}[1]
\Input A quadrature with upper bound $U > 0$, Lipschitz constant $L_{\theta}$ of $\theta_t$ over $[0, U]$, a node of the quadrature $u \in [0,U]$.
\Output $k_t(u)$.

\State Set $N = \lceil \frac{UL_{\theta}}{\pi}\rceil$ and $u_i = i\frac{\pi}{L_{\theta}}, \; i=0, \dots, N-1$, $u_N=U$.
\State Compute $\arg(\det(\bo{\Phi}_t(u_i))), \; i=0, \dots,N$.
\State Initialize $k_t(u_0) \gets 0$ ($u_0=0$)
\For{$i = 0$ to $N-1$}
    \If{$\arg(\det(\bo{\Phi}_t(u_{i+1}))) - \arg(\det(\bo{\Phi}_t(u_i))) > \pi$}
        \State $k_t(u_{i+1}) \gets k_t(u_i) - 1$
    \ElsIf{$\arg(\det(\bo{\Phi}_t(u_{i+1}))) - \arg(\det(\bo{\Phi}_t(u_i))) < -\pi$}
        \State $k_t(u_{i+1}) \gets k_t(u_i) + 1$
    \Else
        \State $k_t(u_{i+1}) \gets k_t(u_i)$
    \EndIf
    \State Store $k_t(u_{i+1})$
\EndFor

\State For any $u \in [0, U]$, set $i=\lfloor \frac{uL_{\theta}}{\pi}\rfloor$
\If{$\arg(\det(\bo{\Phi}_t(u))) - \arg(\det(\bo{\Phi}_t(u_i))) > \pi$}
    \State $k_t(u) \gets k_t(u_i) - 1$
\ElsIf{$\arg(\det(\bo{\Phi}_t(u))) - \arg(\det(\bo{\Phi}_t(u_i))) < -\pi$}
    \State $k_t(u) \gets k_t(u_i) + 1$
\Else
    \State $k_t(u) \gets k_t(u_i)$
\EndIf

\Return $k_t(u)$
\end{algorithmic}
\end{algorithm}

\begin{proof}[Proof of the validity of Algorithm \ref{algo:rot_count_k_t}]
    We need to prove that the algorithm actually returns $k_t(u)$. First, since $\theta_t$ is $L_{\theta}$-Lipschitz over $[0,U]$ and that each $u_i$ belongs to this interval, it follows that
    \begin{equation}
        |\theta_t(u_{i+1}) - \theta_t(u_i)| \leq L_{\theta}|u_{i+1} - u_i| = \pi, \quad \text{for } i = 0, \dots, N-1.
    \end{equation}
    Recall from \eqref{eq:k_t(u,w)} that
    \begin{equation}
        k_t(u) = \frac{1}{2\pi} \left(\theta_t(u) - \arg(\det(\bo{\Phi}_t(u)))\right).
    \end{equation}
    Therefore, $k_t$ satisfies
    \begin{align}
        |k_t(u_{i+1}) - k_t(u_i)| &\leq \frac{1}{2\pi} \left(|\theta_t(u_{i+1}) - \theta_t(u_i)| + |\arg(\det(\bo{\Phi}_t(u_{i+1}))) - \arg(\det(\bo{\Phi}_t(u_i)))|\right) \\
        & \leq \frac{1}{2\pi} (\pi + 2\pi) = \frac{3}{2},
    \end{align}
    which implies that $|k_t(u_{i+1}) - k_t(u_i)| \in \{-1, 0, 1\}$, since $k_t$ is an integer-valued function.

    Now, suppose that $\arg(\det(\bo{\Phi}_t(u_{i+1}))) - \arg(\det(\bo{\Phi}_t(u_i))) > \pi$. Then, since $|\theta_t(u_{i+1}) - \theta_t(u_i)| \leq \pi$, we have $k_t(u_{i+1}) - k_t(u_i) < 0$, which implies $k_t(u_{i+1}) = k_t(u_i) - 1$. A similar argument applies if $\arg(\det(\bo{\Phi}_t(u_{i+1}))) - \arg(\det(\bo{\Phi}_t(u_i))) < -\pi$, in which case $k_t(u_{i+1}) = k_t(u_i) + 1$, or if $|\arg(\det(\bo{\Phi}_t(u_{i+1}))) - \arg(\det(\bo{\Phi}_t(u_i)))| \leq \pi$, in which case $k_t(u_{i+1}) = k_t(u_i)$. 
    
    Finally, for any $u \in [0, U]$, we have $u \in [u_i, u_{i+1})$, where $i = \lfloor \frac{uL_{\theta}}{\pi} \rfloor$. Replacing $u_{i+1}$ by $u$ in the argument above, which is possible as the only requirement on $u_{i+1}$ above was the fact that $|u_{i+1}-u_i|\leq \frac{\pi}{L_{\theta}}$, we obtain the same relationship between $k_t(u)$ and $k_t(u_i)$ than between $k_t(u_{i+1})$ and $k_t(u_i)$.
\end{proof}

Observe that line 1 to 13 of Algorithm \ref{algo:rot_count_k_t} are independent of $u$ and can therefore be executed once for all, separately from the rest of the algorithm.

Furthermore, in order to apply this algorithm, one needs to find at least an estimate of an upper bound of the Lipschitz constant of $\theta_t$ over $[0, U]$. In the integrated variance case, an analytic formula for the upper bound is available and given in Proposition \ref{prop:upper_bound_lip_cst}. In the log-price case, we were unable to find one. Therefore, algorithms specialized in estimating an upper bound of Lipschitz constants, for instance the \textit{Lipschitz Constant Estimation by Least Squares Regression} (LCLS) algorithm proposed in \cite{huang2023sample}, which is a Lipschitz constant estimation algorithm with optimal complexity, should be used as a first step. Note that the Lipschitz constant of $\theta_t$ depends on the model parameters and the maturity $T$. Thus, during a calibration procedure, an upper bound must be estimated for each set of parameter and maturity. In the integrated variance case, $\theta_t$ depends on the mean reversion rate $\kappa$, the vol-of-vol $\nu$ and the maturity $T$. In the log-price case, it depends additionally on the correlation $\rho$. 

This additional step adds complexity, but it is important to point out that if the sample points $u_1 < \dots < u_n$ from a quadrature already satisfy the condition that their increments are smaller than $\frac{\pi}{L_{\theta}}$, then only lines 1 to 13 of the algorithm are required to compute $k_t(u_i)$ for all $u_i$. Therefore, since computing the Fourier–Laplace transform using the determinant formula \eqref{eq:char_func_det_formula} already requires the values of $\det(\bo{\Phi}_t(u_i))$ for all $u_i$, it follows that in this case, computing all $k_t(u_i)$ using Algorithm \ref{algo:rot_count_k_t} would only involve $n$ argument computations and $n$ comparisons, which are computationally negligible. Additionally, the smaller the Lipschitz constant, the more easily the condition on the increments will be satisfied by the quadrature.

In practice, we found that for the rough Stein-Stein model (see Section \ref{sec:numerical_results} for details on the model), the Lipschitz constant tends to be very low, as illustrated in Figures \ref{fig:arg_of_det_H_T} and \ref{fig:arg_of_det_rho}. This implies that for this model, we can reasonably assume that the condition on the increments will always be satisfied for any quadrature, leading to a very fast computation of $k_t$. Consequently, we were able to compute the Fourier–Laplace transform as fast as with the Hybrid method described in Section \ref{subsec:approach_based_on_char_func_proxy}, as shown in Figure \ref{fig:call_price_computation_time_all_algorithms}.

However, if the Lipschitz constant is not sufficiently small, handling it properly using lines 1 to 13 of Algorithm \ref{algo:rot_count_k_t} is crucial to avoid the situation depicted in Figure \ref{fig:arg_of_det_no_care_lip_cst}, again for the rough Stein-Stein model, where in the right graph, the chosen grid is not fine enough to capture all discontinuities of the determinant's argument. The parameters (maturity, Hurst index, correlation and vol-of-vol) chosen for this example are quite extreme, but make it instructive on why the Lipschitz constant must be taken in account.
\begin{figure}[h!]
    \centering
    \includegraphics[width=0.7\textwidth]{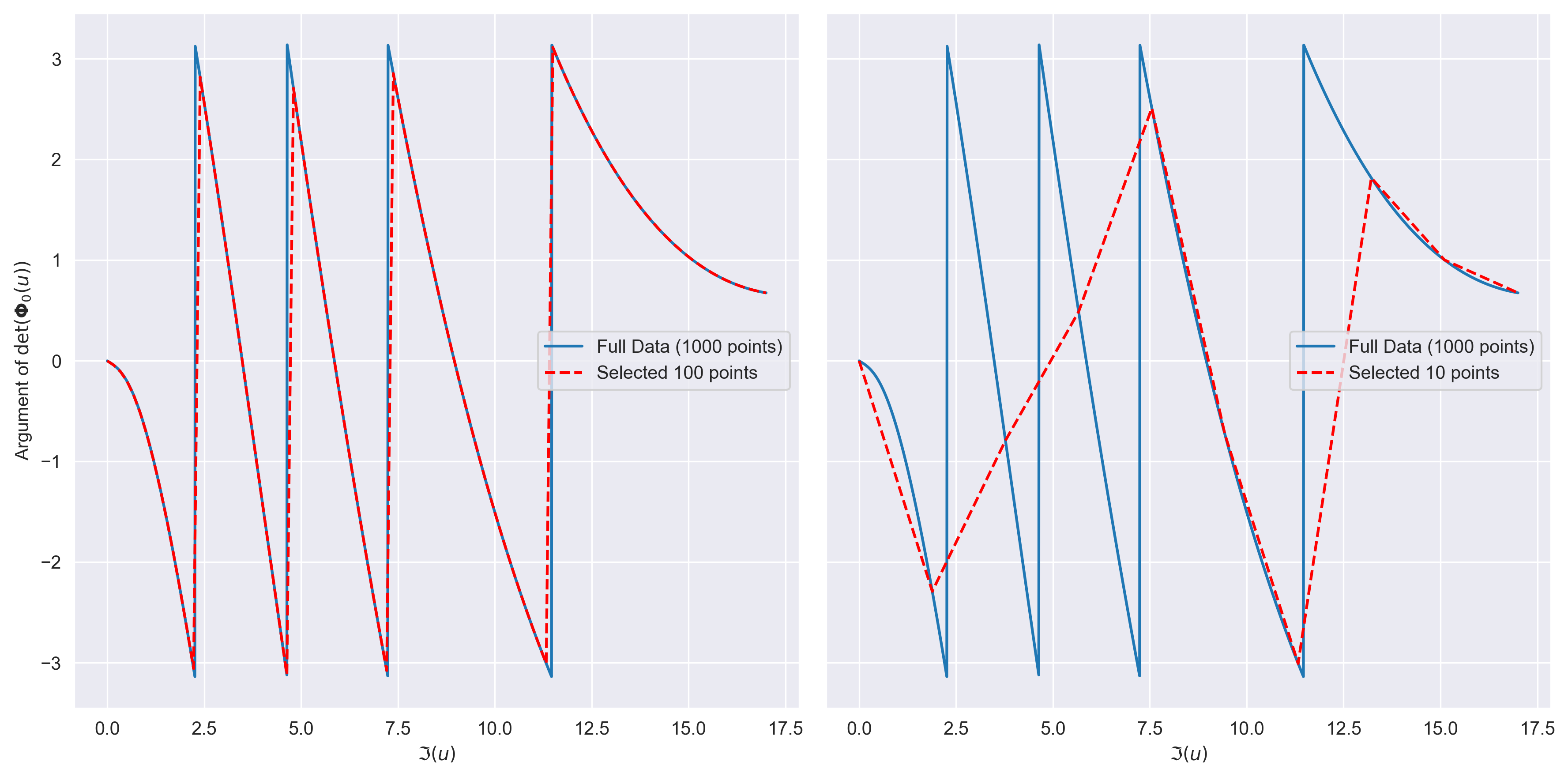} 
    \caption{\small Argument of the Fredholm determinant of $\bo{\Phi}_0$ as a function of $\Im(u)$, in the fractional kernel setting, for a fine grid of $\Im(u)$ (Left, in red), and a large grid (Right, in red). The parameters are $\kappa = 0$, $\nu = 0.7$, $\theta = -0.3$, $\rho = -1$, $X_0 = -0.05$, $H = 0.1$, $T = 3$, $n=200$, $\Re(u) = \frac{1}{2}$, $w=0$.}
    \label{fig:arg_of_det_no_care_lip_cst}
\end{figure}

\begin{proposition} \label{prop:upper_bound_lip_cst}
    Let $0 \leq t \leq T$, and $K$ be a Volterra kernel. Fix $u=0$ and $\Re(w) \leq 0$. Define $L = 2\sum_{n=1}^{+\infty} \arctan\left(\frac{\lambda_n}{1-2\Re(w)\lambda_n}\right)$. Then the map $\Im(w) \in \mathbb R_+ \mapsto \theta_t(w)$ is $L$-Lipschitz. 
\end{proposition}

\begin{proof}
    The proof is given in Appendix \ref{subsec:proof_prop_upper_bound_lip_cst}.
\end{proof} 

\subsection{A prefactor-free approximation formula} \label{subsec:prefactor-free_approx}

In this section, we derive the prefactor-free approximation formula for the Fourier–Laplace transform of the log-price and integrated variance.
To this end, we adopt the notations introduced in Section~\ref{subsec:discretized_operators}, where operators are approximated by matrices. However, to obtain this formula, we will use an other approximation formula for $\phi_0$, which will follow from the following lemma. 
\begin{lemma} \label{lemma:two_integral_formulas_for_phi}
    Let $0 \leq t \leq T$, $(u,w) \in \mc U$ and $K$ be a Volterra kernel of continuous and bounded type in $L^2_{\mathbb R}$. Then
    \begin{equation}
        \phi_t(u,w) = \int_0^1 \mathrm{Tr} \left(\bo{\Psi}_{t,s}(u,w)\bo{\Sigma}_t\right) ds,
    \end{equation}
    where for $0 \le s \le 1$, $\bo{\Psi}_{t,s}(u,w)$ is given by 
    \begin{equation}\label{eq:Psi_t}
        \bo{\Psi}_{t,s}(u,w) := a(u,w) (\text{id} - b(u)\bo{K}^*)^{-1} \left( \text{id} - 2sa(u,w)\tsig_t(u) \right)^{-1} (\text{id} - b(u)\bo{K})^{-1}.
    \end{equation}
\end{lemma}

\begin{proof}
    The proof is given in Appendix \ref{sec:proof_prefactor-free_approx}.
\end{proof}
We now propose a new approximation formula for the determinant formula, which is not the prefactor-free formula, but an intermediate step to obtain it. To ease notations, we drop any dependence on $u$ and $w$. Following the previous lemma, we propose the following approximation for $\phi_0$:
\begin{equation}
    \tilde{\phi}_n := \frac{T}{n} \int_0^1 \Tr\left(\Psi_{n,s}\Sigma_n\right) ds,
\end{equation}
where for $0 \leq s \leq 1$,
\begin{equation} 
 \Psi_{n,s} := a\left(I_n -b K_n^{\top}\right)^{-1}\left(I_n - 2sa\frac{T}{n}\tilde{\Sigma}_n\right)^{-1} \left(I_n -b K_n\right)^{-1}.
\end{equation}
Note that $\Psi_{n,1} = \Psi_n$, recalling the notations of Section~\ref{subsec:discretized_operators}. It induces the following approximation for the determinant formula:
\begin{equation} \label{eq:alternative_FLT_approx_via_trace}
    \mathbb{E} \left[ \exp \left( u \frac{\log S_T}{\log S_0} + w \int_0^T X_s^2 ds \right)\right] \approx e^{i\pi \tilde{k}_n} \frac{\exp \left(\frac{T}{n} g_n^{\top} \Psi_n g_n\right)}{\sqrt{\det(\Phi_n)}}, 
\end{equation}
with 
\begin{equation}
    \tilde{k}_n = \frac{1}{2\pi} \left( \arg(\det(\Phi_n)) + 2\Im(\tilde{\phi}_n) \right).
\end{equation}
As discussed after \eqref{eq:intvariancedetsqrt}, the prefactor-free formula will free us from the $\pm 1$ prefactor by computing $\det(\sqrt{\Phi_n})$ instead of $\sqrt{\det(\Phi_n)}$.
However, since $\Phi_n$ is in general not Hermitian, defining a matrix square root is delicate, and may not even be possible in full generality. 

The following lemma, together with the discussion that follows, aims to give a rigorous meaning to $\det(\sqrt{\Phi_n})$.

\begin{lemma} \label{lemma:tilde_Phi_intro+pos_def}
    Let $(u,w) \in \mc U$ and $K$ be a Volterra kernel of continuous and bounded type in $L^2_{\mathbb R}$. Define 
    \begin{equation} \label{eq:tilde_Phi}
        \tilde{\Phi}_n(u,w) := \left(\id - b(u,w)K_n\right) \Phi_n(u,w) \left(\id - b(u,w)K_n^*\right)
    \end{equation}
    Then,
    $\Re\left(\tilde{\Phi}_n(u,w)\right)$ and $\Im\left(\tilde{\Phi}_n(u,w)\right)$ are hermitian. Moreover, $\Re\left(\tilde{\Phi}_n(u,w)\right)$ is positive definite.
\end{lemma}

\begin{proof}
    The proof is given in Appendix \ref{sec:proof_prefactor-free_approx}.
\end{proof}
Set $(u,w) \in \mc U$. From the previous lemma, and since $K_n$ is strictly triangular, we have
\begin{align}
    \det\left(\Phi_n(u,w)\right) 
    &= \det\left(\tilde{\Phi}_n(u,w)\right) \\
    &= \det\left(\Re\left(\tilde{\Phi}_n(u,w)\right) + i\Im\left(\tilde{\Phi}_n(u,w)\right)\right) \\
    &= \det\Bigg(
        \Re\left(\tilde{\Phi}_n(u,w)\right)^{\frac{1}{2}}
         \\
        &\qquad \quad \cdot \left( I_n + i \Re\left(\tilde{\Phi}_n(u,w)\right)^{-\frac{1}{2}} \Im\left(\tilde{\Phi}_n(u,w)\right) \Re\left(\tilde{\Phi}_n(u,w)\right)^{-\frac{1}{2}} \right) \Re\left(\tilde{\Phi}_n(u,w)\right)^{\frac{1}{2}}\Bigg) \\
    &= \det\left(\Re\left(\tilde{\Phi}_n(u,w)\right)\right) \\
    &\quad \cdot 
       \det\Bigg(
       I_n + i \Re\left(\tilde{\Phi}_n(u,w)\right)^{-\frac{1}{2}} 
       \Im\left(\tilde{\Phi}_n(u,w)\right) 
       \Re\left(\tilde{\Phi}_n(u,w)\right)^{-\frac{1}{2}}
       \Bigg)
       \label{eq:det_Phi_n_decomposed}
\end{align}
Let $(\lambda_k(u,w))_{1 \leq k \leq n}$ denote the real-valued eigenvalues of the symmetric matrix
\begin{equation}
    \Re(\tilde{\Phi}_n(u,w))^{-\frac{1}{2}} \Im(\tilde{\Phi}_n(u,w)) \Re(\tilde{\Phi}_n(u,w))^{-\frac{1}{2}}.
\end{equation} 
We then define
\begin{equation}
    \det\left(\sqrt{\Phi_n(u,w)}\right) := \sqrt{\det\left(\Re\left(\tilde{\Phi}_n(u,w)\right)\right)} \prod_{i=1}^n \sqrt{1 + i \lambda_k(u,w)},
\end{equation}
where from the previous lemma, $\det(\Re(\tilde{\Phi}_n(u,w))) > 0$, and $\sqrt{\cdot}$ inside the product denotes the principal branch of the square root.

We now have all the necessary ingredients to derive the prefactor-free approximation formula.

\begin{proposition} \label{prop:det_sqrt_formula}
    Let $ g_0 \in L^2([0,T], \mathbb{R}) $ and $ K $ be a Volterra kernel of continuous and bounded type in $L^2_{\mathbb R}$. Let $(u, w) \in \mc U$. Then,
    \begin{equation}
        e^{i\pi \tilde{k}_n(u,w)} \frac{\exp \left(\frac{T}{n} g_n^{\top} \Psi_n(u,w) g_n\right)}{\sqrt{\det(\Phi_n(u,w))}} = \frac{\exp \left(\frac{T}{n} g_n^{\top} \Psi_n(u,w) g_n\right)}{\det(\sqrt{\Phi_n(u,w)})},
    \end{equation}
    so that, following \eqref{eq:alternative_FLT_approx_via_trace}, the prefactor-free formula is defined by  
    \begin{equation} \label{eq:prefactor-free_formula}
        \mathbb{E} \left[ \exp \left( u \frac{\log S_T}{\log S_0} + w \int_0^T X_s^2 ds \right)\right] :\approx \frac{\exp \left(\frac{T}{n} g_n^{\top} \Psi_n(u,w) g_n\right)}{\det(\sqrt{\Phi_n(u,w)})}.
    \end{equation}
\end{proposition}

\begin{proof}
    The proof is given in Appendix \ref{sec:proof_prefactor-free_approx}.
\end{proof}
The main advantage of the prefactor-free formula, compared with the previous approximations based on the Hybrid Algorithm~\ref{algo:compute_phi_t} and the Lipschitz Algorithm~\ref{algo:rot_count_k_t}, is that it no longer depends on the behavior of the Fredholm determinant.
The trade-off is that, instead of computing the square root of a complex number (the determinant), one must compute the determinant of a matrix square root, and hence evaluate a matrix square root itself.
When the coarse approximation used in Algorithm~\ref{algo:compute_phi_t} requires a large $n_\text{coarse}$, or when the determinant’s Lipschitz constant is high enough to necessitate a Lipschitz constant estimator algorithm in Algorithm~\ref{algo:rot_count_k_t}, then the additional cost of computing a matrix square root becomes negligible, making the prefactor-free formula computationally more efficient.
In the opposite case, the extra complexity of computing a matrix square root makes the prefactor-free formula less advantageous, as illustrated in Figure~\ref{fig:call_price_computation_time_all_algorithms} for the rough Stein–Stein model.

\section{Numerical illustrations using the rough Stein-Stein model} \label{sec:numerical_results}

In this section, we demonstrate the superior accuracy and computational efficiency of the determinant formula~\eqref{eq:char_fun_sqrt_det}, combined with either the Hybrid Algorithm~\ref{algo:compute_phi_t}, the Lipschitz Algorithm~\ref{algo:rot_count_k_t}, or the prefactor-free formula~\eqref{eq:prefactor-free_formula}, relative to the trace formula~\eqref{eq:char_func_trace_formula}.
We further compare the performance of the Hybrid and Lipschitz algorithms, as well as the prefactor-free formula, in terms of both accuracy and computational time. A Jupyter Notebook containing all the necessary code to reproduce the following numerical illustrations presented is available on Google Colab at \url{https://colab.research.google.com/drive/1O87IJpqGo1E8oGXnZxFqamYnm3Oa-_8d?usp=sharing}.

We illustrate our results on the fractional Stein–Stein model, based on the Riemann–Liouville fractional Brownian motion with the
Volterra convolution kernel $K(t,s) := \mathbbm{1}_{s<t} \frac{(t-s)^{H-\frac{1}{2}}}{\Gamma(H + \frac{1}{2})}$:
\begin{align}
dS_t &= S_t X_t dB_t, \quad S_0 > 0, \\
X_t &= g_0(t) + \frac{\nu}{\Gamma(H + 1/2)} \int_0^t (t - s)^{H - 1/2} \, dW_s,
\end{align}
with $B = \rho W + \sqrt{1 - \rho^2} W_\perp$, for $\rho \in [-1, 1]$, $\nu \in \mathbb{R}$, and a Hurst index $H \in (0, 1)$. For illustration purposes, we will consider that the input curve $g_0$, which can be used in general to fit at-the-money curves observed in the market, has the following parametric form:
\begin{equation}
g_0(t) = X_0 + \frac{1}{\Gamma(H + 1/2)} \int_0^t (t - s)^{H - 1/2} \theta \, ds = X_0 + \theta \frac{t^{H + 1/2}}{\Gamma(H + 1/2)(H + 1/2)}.
\end{equation}

For $0 \leq i \leq n$, and $1 \leq j,k \leq n$, the components of the discretized operators for this model admit the following closed form:
\begin{align}
    (K_n){jk} &= \mathbbm{1}_{k \leq j-1} \frac{1}{\Gamma(1+\alpha)} \left((t_{j-1} - t_{k-1})^{\alpha} - (t_{j-1} - t_k)^{\alpha}\right), \\
    (\Sigma_{n,i})_{jk} &= \nu^2 \mathbbm{1}_{i < (j \wedge k) - 1} \frac{1}{\Gamma(\alpha)\Gamma(1+\alpha)} \frac{(t_{j-1} \wedge t_{k-1} - t_i)^{\alpha}}{(t_{j-1} \vee t_{k-1} - t_i)^{1 - \alpha}} \, {}_2F_1\left(1, 1 - \alpha; 1 + \alpha; \frac{t_{j-1} \wedge t_{k-1} - t_i}{t_{j-1} \vee t_{k-1} - t_i}\right), \\
    (\dot{\Sigma}_{n,i})_{jk} &= -\nu^2 \mathbbm{1}_{i \leq (j \wedge k) - 1} \frac{(t_{j-1} - t_i)^{\alpha - 1}}{\Gamma(\alpha)\Gamma(1 + \alpha)} \left((t_k - t_i)^{\alpha} - (t_{k-1} - t_i)^{\alpha}\right).
\end{align}
where $\alpha = H + \frac{1}{2}$, $\Gamma$ is the Gamma function, and ${}_2F_1$ the Gaussian hypergeometric function.

Finally, the vector $g_n$ is given by
\begin{equation}
    (g_n)_i = g_0(t_i) = X_0 + \theta \frac{t_{i-1}^{\alpha}}{\Gamma(1+\alpha)}.
\end{equation}

In Figures~\ref{fig:call_price_precision_low_T} and \ref{fig:call_price_precision_high_T}, we analyze the convergence of call prices for different maturities: short ($T=0.05$) and long ($T=1$). The analysis covers strike prices ranging from in-the-money to out-of-the-money options. Furthermore, we study the impact of the Hurst index, varying from $H=0.1$, which corresponds to the rough Stein-Stein model, to $H=0.5$, which corresponds to the standard Stein-Stein model.

\begin{figure}[h!]
    \centering
    \includegraphics[width=1\textwidth,height=0.39\textheight]{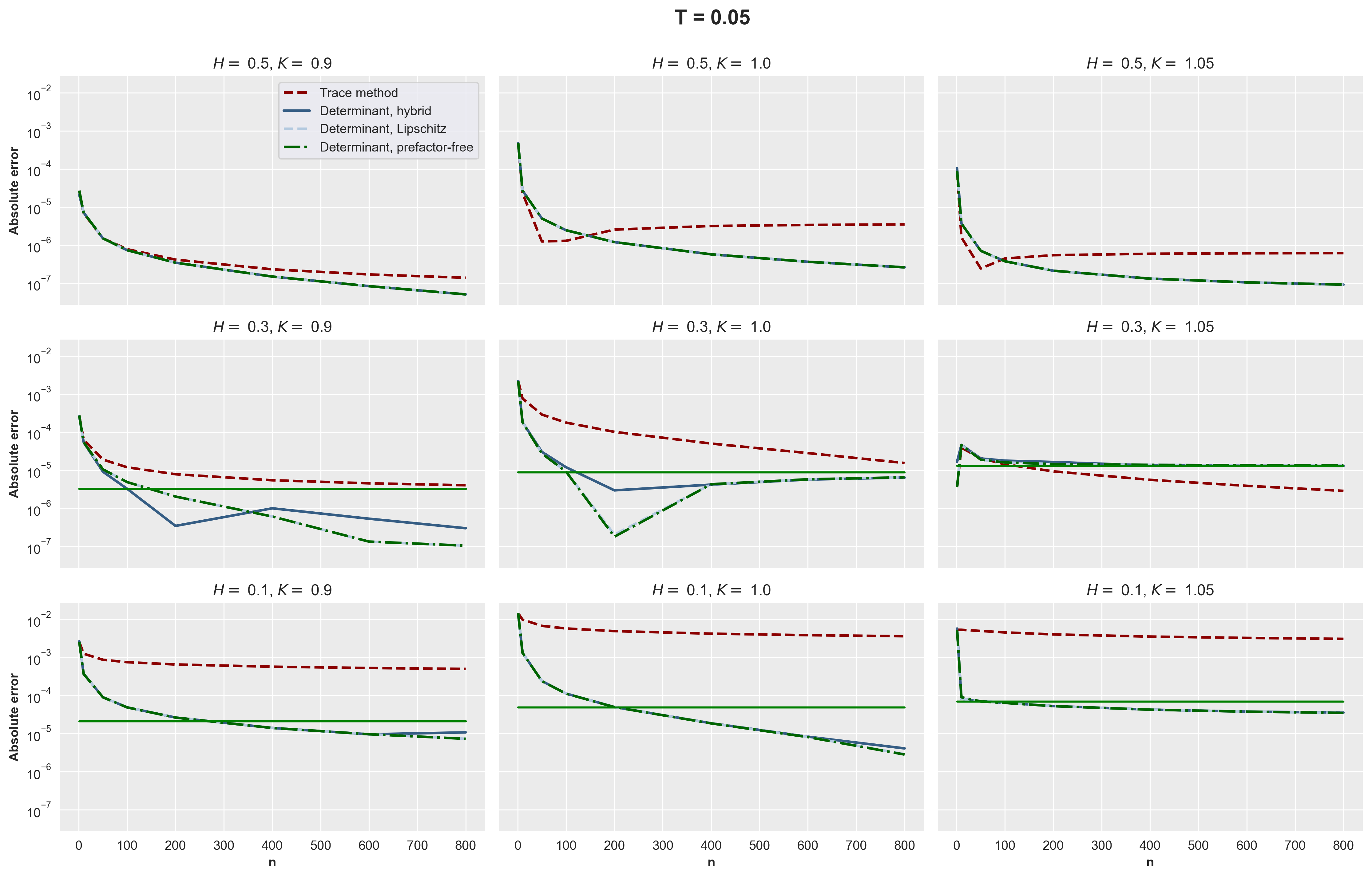} 
    \caption{\scriptsize Convergence of 0.05-year ITM (Left), ATM (Middle), and OTM (Right) call prices in both the conventional Stein-Stein model ($H=0.5$) and the rough Stein-Stein model ($H=0.3$ and $H=0.1$), using the discretization method from Section \ref{subsec:discretized_operators} for (i) the trace formula and (ii) the determinant formula combined with the hybrid and Lipschitz algorithms and the prefactor-free formula. For $H=0.5$, the benchmark is the closed-form solution of the conventional Stein-Stein model, while for $H=0.3$ and $H=0.1$, the benchmark is the 95\% Monte Carlo price. The green line represents the half-length of the confidence interval. The degree of the Gauss-Laguerre quadrature is 30. The parameters are $\kappa=0$, $\nu=0.25$, $\theta=0.1$, $\rho=-0.7$, $X_0 = 0.1$, $S_0 = 1$, and $n_{\text{coarse}}=40$.}
    \label{fig:call_price_precision_low_T}
\end{figure}

\begin{figure}[h!]
    \centering
    \includegraphics[width=1\textwidth,height=0.39\textheight]{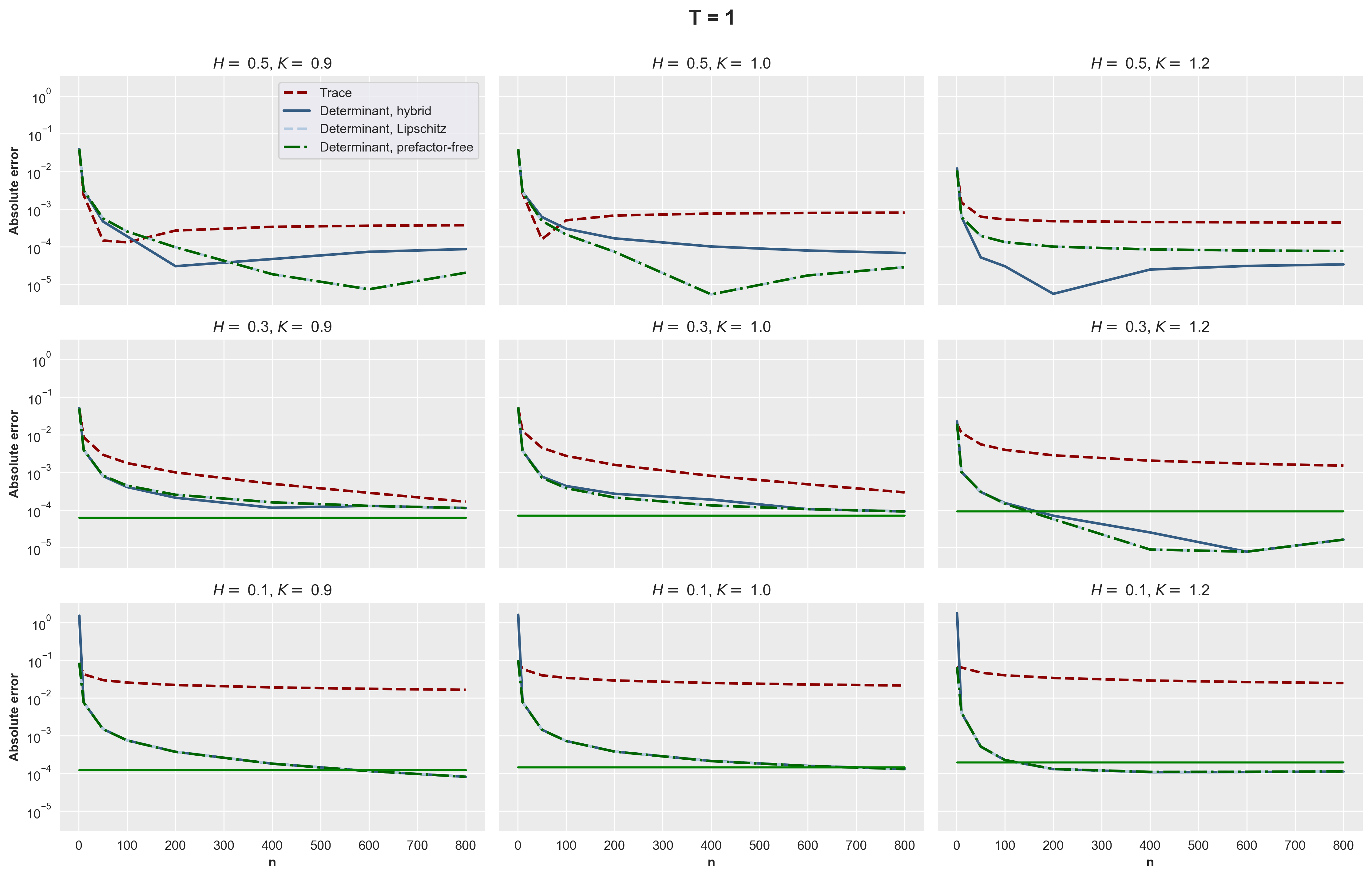} 
    \caption{\scriptsize Convergence of 1-year ITM (Left), ATM (Middle), and OTM (Right) call prices in both the conventional Stein-Stein model ($H=0.5$) and the rough Stein-Stein model ($H=0.3$ and $H=0.1$), using the discretization method from Section \ref{subsec:discretized_operators} for (i) the trace formula and (ii) the determinant formula combined with the hybrid and Lipschitz algorithms and the prefactor-free formula. For $H=0.5$, the benchmark is the closed-form solution of the conventional Stein-Stein model, while for $H=0.3$ and $H=0.1$, the benchmark is the 95\% Monte Carlo price. The green line represents the half-length of the confidence interval. The degree of the Gauss-Laguerre quadrature is 80. The parameters are $\kappa=0$, $\nu=0.25$, $\theta=0.1$, $\rho=-0.7$, $X_0 = 0.1$, $S_0 = 1$, and $n_{\text{coarse}}=40$.}
    \label{fig:call_price_precision_high_T}
\end{figure}

We computed call prices making use of Fourier inversion techniques, in particular the Lewis formula \eqref{eq:Lewis_formula} combined with the Gauss-Laguerre quadrature.

The convergence is studied in term of the number of discretization steps, following the approach described in Section \ref{subsec:discretized_operators}, 
using the trace formula \eqref{eq:char_func_trace_formula} and determinant formula \eqref{eq:char_fun_sqrt_det}, in combination with the hybrid \ref{algo:compute_phi_t} and Lipschitz \ref{algo:rot_count_k_t} algorithms. This allowed us to compare the precision and computation time of the two formulas, as well as evaluate the performance of the two algorithms for determining the correct sign of the Fourier–Laplace transform. For $H=0.5$, the benchmark is the closed-form solution of the conventional Stein-Stein model, see \cite{lord2006rotation}, while for $H=0.3, 0.1$, the benchmark is the 95\% Monte Carlo price, which is computed using $10^6$ trajectories and $10^4$ time steps.

Although the difference in accuracy between the determinant formula combined with the Hybrid or Lipschitz algorithms and the prefactor-free formula is not significant, the difference of precision with the trace formula is irrevocable: the determinant formula largely surpasses the precision of the trace formula for low Hurst indices, as expected in Remark \ref{rmk:why_using_det} and Section \ref{sec:model_&_char_fun}, and illustrated in Figures \ref{fig:call_price_precision_low_T} and \ref{fig:call_price_precision_high_T}. More than that, the trace formula is not even capable of entering the Monte-Carlo confidence interval with less than a thousand time steps, and this is even amplified for a long maturity. 

\begin{figure}[h!]
    \centering    \includegraphics[width=1\textwidth,height=0.3\textheight]{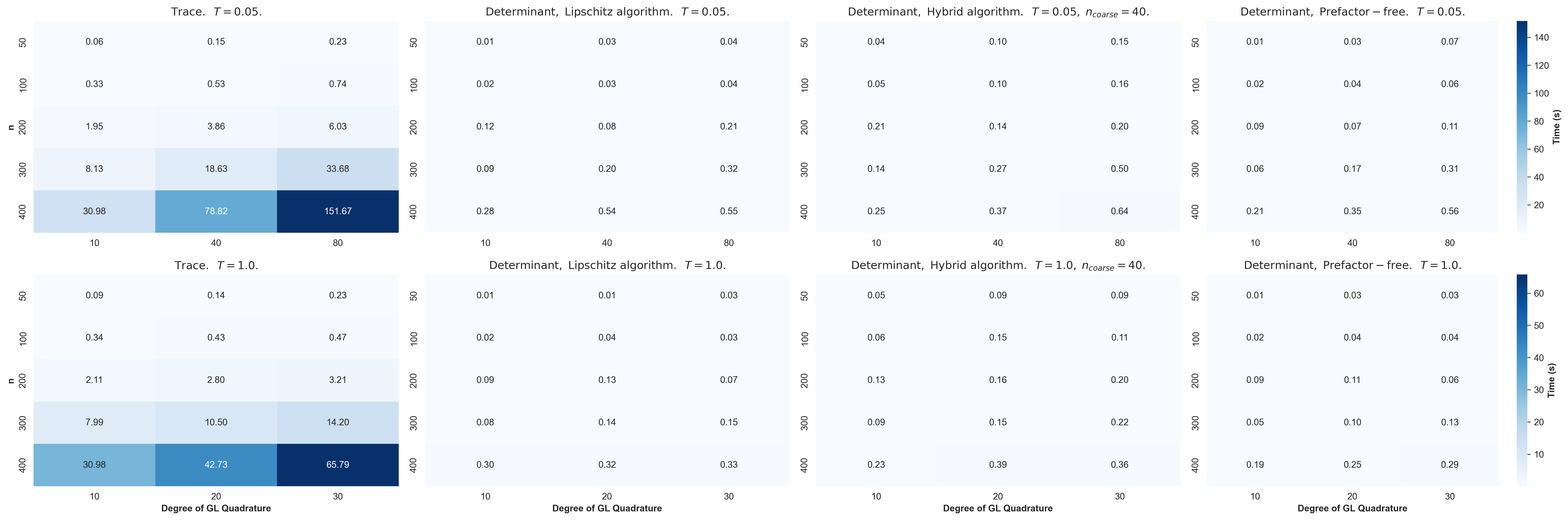} 
    \caption{\small Comparison of computation times for call price evaluation using: the trace formula, the determinant formula with the Lipschitz algorithm, the determinant formula with the Hybrid algorithm, and the determinant prefactor-free formula (from left to right). The top row corresponds to short maturity ($T=0.05$) and the bottom row to long maturity ($T=1.0$). Computation times are displayed as functions of the Gauss–Laguerre quadrature degree and the number of discretization steps. The parameters are $\kappa=0$, $\nu=0.25$, $\theta=0.1$, $\rho=-0.7$, $X_0 = 0.1$, $S_0 = 1$, and $H=0.3$.}
    \label{fig:call_price_computation_time_all_algorithms}
\end{figure}

In Figure \ref{fig:call_price_computation_time_all_algorithms}, we compare the computation times required to evaluate call prices using the trace formula, the determinant formula combined with the  Hybrid, Lipschitz algorithms, or the prefactor-free formula, for different numbers of discretization steps and Gauss–Laguerre quadrature points. The comparison is performed for both short and long maturities with a Hurst index $H = 0.3$. As expected (see Remark \ref{rmk:why_using_det} and Section \ref{sec:model_&_char_fun}), the determinant-based significantly outperforms the trace formula. Furthermore, the computation times with the two algorithms or the prefactor-free formula are of the same order. However, as discussed in Section \ref{Sec:hybrid_lip_algos}, the complexity of these algorithms depends strongly on the kernel $K$ and on the model parameters. Hence, this numerical comparison should be viewed as an illustration rather than a general rule.

\section{Proof of Theorems \ref{thm:crossing_neg_axis_integrated_variance} and \ref{thm:crossing_neg_axis_log_price}} \label{sec:proof_thm_crossing}

We recall that we defined in Section \ref{subsec:approach_based_on_lip_cst}, for a Volterra kernel $K$, the quantities $r_t(u,w) := e^{-2\Re(\phi_t(u,w))}$ and $\theta_t(u,w) := -2\Im(\phi_t(u,w))$ so that, we have 
\begin{equation} 
    \det(\bo{\Phi}_t(u,w)) = r_t(u,w)e^{i\theta_t(u,w)}
\end{equation}

The main idea behind the proofs of Theorems \ref{thm:crossing_neg_axis_integrated_variance} and \ref{thm:crossing_neg_axis_log_price} is to make explicit the polar representation of the determinant of $\bo{\Phi}_t$, recalling that $\bo{\Phi}_t = \text{id} - 2a\tsig_t$ from \eqref{eq:Phi_t}. Specifically, we aim to express the value of $\theta_t$ in terms of the spectrum of $\tsig_t$ and analyze this angle to determine under which conditions we have $e^{i\pi k_t} = -1$.

In section \ref{subsubsubsec:spectrum_study_lambda}, we show that the eigenvalues of $\tsig_t$ are decreasing and continuous in $t$, which helps us, in Section \ref{subsubsec:explicit_polar_form}, to find an explicit expression for $\theta_t$ in terms of the spectrum of $\tsig_t$. In Section \ref{subsubsec:put_every_together}, we combine all these results to prove the theorems. The proof of Corollary \ref{cor:crossing_neg_axis_log_price} will be based on the continuity of $\phi_t$ with respect to $u$ and $\rho$.

\subsection{Spectrum of $\tsig_t$ when $\rho \Im(u)=0$} \label{subsubsubsec:spectrum_study_lambda}

We start by a lemma showing that $\tsig_t(u)$ is a symmetric positive operator under the assumption $\rho \Im(u)=0$.
\begin{lemma} \label{lemma:tilde_sigma_formula}
    Let $0 \leq t \leq T$, $K$ be a Volterra kernel and $u \in \mathbb C$. Suppose that $\rho \Im(u)=0$. Then, there exists a Volterra kernel $A_t(u)$ such that
    \begin{equation}
        \tsig_t(u) = \bo{A}_t(u) \bo{A}_t(u)^*,
    \end{equation}
    making $\bo{\tsig}_t(u)$ a symmetric positive operator in $\mc B(L^2_{\mathbb R})$. Moreover, $A_t(u)$ satisfies the identity
    \begin{equation} \label{eq:A_t(s,z)_A_0(s,z)}
        A_t(u)(s,z) = A_0(u)(s,z) \mathbbm{1}_{t \leq z}, \quad z, s \leq T.
    \end{equation}
\end{lemma}
\begin{proof} 
    First, defining $\bo{K}_t$ as the operator induced by the kernel $K_t(s,z) = K(s,z) \mathbbm{1}_{t \leq z}$, we have from \eqref{eq:Sigma} that $\bo{\Sigma}_t = \bo{K}_t \bo{K}_t^*$. Moreover, as $\rho \Im(u)=0$, then $b(u) \in \mathbb{R}$, which implies that
    \begin{align}
        \tsig_t(u) &= \nu^2 (\text{id} - b(u) \bo{K})^{-1} \bo{K}_t \bo{K}_t^* (\text{id} - b(u) \bo{K}^*)^{-1} \\
        &= \nu^2 (\text{id} - b(u) \bo{K})^{-1} \bo{K}_t \left((\text{id} - b(u) \bo{K})^{-1} \bo{K}_t \right)^* \\
        &= \bo{A}_t(u) \bo{A}_t(u)^*
    \end{align}
    with $\bo{A}_t(u) = \nu (\text{id} - b(u) \bo{K})^{-1} \bo{K}_t$. Furthermore, according to \citet[Lemma A.2]{abi2022characteristic}, the kernel $b(u)K$ admits a resolvent Volterra kernel $R(u)$, so that $\left(\text{id} - b(u) \bo{K}\right)^{-1} = \text{id} + \bo{R}(u)$. Thus, we have
    \begin{equation}
        \frac{1}{\nu}A_t(u) = K_t + R(u) \star K_t,
    \end{equation}
    where the $\star$-product is defined in \eqref{eq:star_product}. Both $K_t$ and $R(u) \star K_t$ are Volterra kernels satisfying \eqref{eq:A_t(s,z)_A_0(s,z)} respectively by definition and star-product properties, following \citet[Example 3.2., (iv)]{abi2022characteristic}. Therefore, $A_t(u)$, as their sum (up to a factor $\nu$), is also a Volterra kernel satisfying \eqref{eq:A_t(s,z)_A_0(s,z)}.
\end{proof}

We now present two technical lemmas concerning the growth and continuity of the eigenvalues of $\tsig_t$, that we introduced in Section \ref{subsec:condition_on_crossings}.

\begin{lemma} \label{lemma:majo_lambda_n}
    Let $0\leq t\leq T$, $u \in \mathbb{C}$, and $n \in \mathbb{N}^*$. Let $K$ be a Volterra kernel. Suppose $\rho \Im(u)=0$. Then 
    \begin{equation}
        \lambda_{n,t}(u) \leq \lambda_{n,0}(u) \leq \lambda_{1,0}(u).
    \end{equation}
\end{lemma}

\begin{proof}
    By definition, we have $\lambda_{n,0}(u) \leq \lambda_{1,0}(u)$. Thus, it only stays to prove that $\lambda_{n,t}(u) \leq \lambda_{n,0}(u)$.
    According to \citet[Lemma II.1.1]{gohberg69}, it suffices to prove that 
    \begin{equation}
        0 \leq \tsig_t(u) \leq \tsig_0(u)
    \end{equation}
    in the sense 
    \begin{equation}
        0 \leq \langle \tsig_t(u) f, f \rangle_{L^2_{\mathbb R}} \leq \langle \tsig_0(u) f, f \rangle_{L^2_{\mathbb R}}, \quad f \in L^2([0,T], \mathbb{R}),
    \end{equation}
    since $\tsig_t(u)$ and $\tsig_0(u)$ are elements of $\mc B(L^2_{\mathbb R})$.
    From now and for the sake of readability, we set $\langle \cdot, \cdot \rangle :=  \langle \cdot, \cdot \rangle_{L^2_{\mathbb R}}$. Since $\rho \Im(u)=0$, we have from Lemma \ref{lemma:tilde_sigma_formula} that 
    \begin{equation}
        \tsig_t(u) = \bo{A}_t(u) \bo{A}_t(u)^*
    \end{equation}
    with $A_t(u)(s,z) = A_0(u)(s,z) \mathbbm{1}_{t \leq z}, \; z, s \leq T$. Therefore, for $f \in L^2([0,T], \mathbb{R})$, we have
    \begin{equation}
        \langle \tsig_t(u) f, f \rangle = \langle \bo{A}_t(u)^* f, \bo{A}_t(u)^* f \rangle \geq 0.
    \end{equation}
    Moreover,
    {\small
    \begin{align} 
        \langle \tsig_0(u) f, f \rangle - \langle \tsig_t(u) f, f \rangle &= \langle \left( \tsig_0(u) - \tsig_t(u) \right) f, f \rangle \\
        &= \int_0^T \left( (\tsig_0(u) f)(s) - (\tsig_t(u) f)(s) \right) f(s) ds \\
        &= \int_0^T \left(\int_0^T \left(\tilde{\Sigma}_0(u)(s,z) - \tilde{\Sigma}_t(u)(s,z)\right)f(z) dz\right) f(s) ds \\
        &= \int_0^T \left(\int_0^T \left(\int_0^T  A_0(u)(s,x)A_0(u)(z,x) dx \right. \right.\\
        &\quad\left. \left. - \int_0^T A_t(u)(s,x)A_t(u)(z,x) dx\right)f(z) dz\right) f(s) ds \\
        &= \int_0^T \left(\int_0^T \left(\int_0^T  A_0(u)(s,x)A_0(u)(z,x) dx \right. \right.\\
        &\quad\left. \left. - \int_t^T A_0(u)(s,x)A_0(u)(z,x) dx\right)f(z) dz\right) f(s) ds \\
        &= \int_0^T \left(\int_0^T \left(\int_0^t A_0(u)(s,x)A_0(u)(z,x) dx \right)f(z) dz\right) f(s) ds \\
        &= \int_0^t \left( \int_0^T A_0(u)(s,x)f(s) ds\right)\left(\int_0^T A_0(u)(z,x)f(z)dz\right) dx \\
        &= \int_0^t (\bo{A}_0(u) f(x))^2 dx \geq 0   
    \end{align}}
    where in the fifth equality, we used \eqref{eq:A_t(s,z)_A_0(s,z)}, and in the penultimate equality, we used Fubini's theorem.
\end{proof}

The following lemma establishes the continuity in time of the eigenvalues.  
\begin{lemma}\label{lemma:lambda_n_C^0}
    Let $u \in \mathbb{C}$, $n \in \mathbb{N}^*$ and $0\leq t\leq T$. Suppose $\rho \Im(u) = 0$. Then 
    \begin{equation}
        \lambda_{n,t+h}(u) \xrightarrow{h \to 0} \lambda_{n,t}(u),
    \end{equation}
    meaning that the eigenvalues of $\tsig_t(u)$ are continuous in $t$.
\end{lemma}

\begin{proof}
    The result follows from \citet[Chapter 1, Theorem 4.2]{gohberg69}, once we have proven that $\tsig_{t}(u)$ is continuous in $t$ with respect to the uniform norm on $\mc B(L^2_{\mathbb R})$. From Lemma \ref{lemma:tilde_sigma_formula}, it suffices to show that $\bo{A}_t(u)$ is continuous in $t$. Let $f \in L^2([0,T], \mathbb{R})$. Then, 
    \begin{align}
        \left\lVert \left(\bo{A}_{t+h}(u) - \bo{A}_t(u)\right) f \right\rVert_{L^2_{\mathbb R}}^2 
        &= \int_0^T \left( (\bo{A}_{t+h}(u)) f(s) - (\bo{A}_t(u)) f(s) \right)^2 ds \notag \\
        &= \int_0^T \left(\int_0^T A_{t+h}(u)(s,x)f(x) dx - \int_0^T A_t(u)(s,x)f(x) dx\right)^2 ds \notag \\
        &= \int_0^T \left(\int_t^{t+h} A_0(u)(s,x)f(x) dx \right)^2 ds \label{eq:int_t_t+eps} \\
        &\leq \lVert f \rVert_{L^2_{\mathbb R}}^2 \int_{[0,T]\times [t,t+h]} (A_0(u)(s,x))^2 dx ds \xrightarrow{h \to 0} 0\label{ineq:C.S}
    \end{align}
    where \eqref{eq:int_t_t+eps} follows from \eqref{eq:A_t(s,z)_A_0(s,z)}, and \eqref{ineq:C.S} from the Cauchy-Schwarz inequality and the monotone convergence theorem.
\end{proof}

\subsection{The polar representation of $\det(\bo{\Phi}_t(u,w))$ when $\rho \Im(u)=0$} \label{subsubsec:explicit_polar_form}

The next two lemmas aim to explicitly express the polar representation of $\det(\bo{\Phi}_t(u,w))$ given by $-2\phi_t(u,w)$, recall \eqref{eq:det_exp_phi}, and express it in terms of the spectrum $(\lambda_{n,t}(u))_{n \in \mathbb{N}^*}$ of $\tsig_t(u)$.\\

The first lemma provides an alternative polar representation of $\det(\bo{\Phi}_t(u,w))$ in terms of $(\lambda_{n,t}(u))_{n \in \mathbb{N}^*}$. The second lemma demonstrates that this polar representation is actually the same as the one given by $-2\phi_t(u,w)$.

\begin{lemma} \label{lemma:det_phi_t_polar} 
    Let $0 \leq t \leq T$, $u, w \in \mc U$ and $K$ be a Volterra kernel. Suppose $\rho \Im(u) = 0$. For $n \in \mathbb{N}^*$, set  
    \begin{align}
        x_{n,t}(u,w) &:= 1 - 2\Re(a(u,w))\lambda_{n,t}(u)\label{eq:a_k}, \\
        y_{n,t}(u,w) &:= - 2\Im(a(u,w))\lambda_{n,t}(u) \label{eq:b_k}.
    \end{align}
    so that, as $\tsig_t(u)$ is of trace class, 
    \begin{equation} \label{eq:det_Phi_t}
        \det(\bo{\Phi}_t(u,w)) = \prod_{n=1}^{\infty} \left(x_{n,t}(u,w)+iy_{n,t}(u,w)\right).
    \end{equation}
    Then, 
    \begin{align}
    \tilde{\theta}_t(u,w)&:=\sum_{n=1}^{\infty}  \arctan\left(\frac{y_{n,t}(u,w)}{x_{n,t}(u,w)}\right), \label{eq:tilde_theta} \\
    \tilde{r}_t(u,w)&:=\prod_{n=1}^{\infty} \sqrt{(x_{n,t}(u,w))^2 + (y_{n,t}(u,w))^2},
    \end{align}
    are well defined and 
    \begin{equation} \label{eq:det_phi_t_polar_arctan}
        \det(\bo{\Phi}_t(u,w)) = \tilde{r}_t(u,w)e^{i\tilde{\theta}_t(u,w)} 
    \end{equation}
\end{lemma}

\begin{proof}
    To simplify notations and enhance readability, we omit any explicit dependence on $t$, $u$, and $w$, as these parameters do not play a role in the proof. For $m \in \mathbb{N}^*$, set $\det(\bo{\Phi})_m := \prod_{n=1}^m \left(x_n + iy_n\right)$. First, from Remark \ref{rmk:det_Phi_well_def}, \eqref{eq:det_Phi_t} is well defined and $\lim_{m \to \infty} \det(\bo{\Phi})_m = \det(\bo{\Phi})$. We also have, since $0 \leq \Re(u) \leq 1$ and $\Re(w) \leq 0$, that
    \begin{equation}
        \Re(a) \leq 0,
    \end{equation}
    which implies, since $\lambda_n \geq 0$ for all $n \geq 1$, that 
    \begin{equation} \label{ineq:a_n_t>1}
        x_n \geq 1, \hspace{0,3cm} n \geq 1. 
    \end{equation}
    Thus, $\arg(x_n + iy_n) \in (-\frac{\pi}{2}, \frac{\pi}{2})$, so that its argument can be written as
    \begin{equation}
        \arg(x_n + iy_n) = \arctan\left(\frac{y_n}{x_n}\right).
    \end{equation}
    Therefore, $x_n + iy_n = \sqrt{x_n^2 + y_n^2}\exp\left(i\arctan\left(\frac{y_n}{x_n}\right)\right)$, and defining
    \begin{align}
        \tilde{\theta}_m &:= \sum_{n=1}^m \arctan\left(\frac{y_n}{x_n}\right), \\
        \tilde{r}_m &:= \prod_{n=1}^m \sqrt{x_n^2 + y_n^2},
    \end{align}
    we have
    \begin{equation}
        \det(\bo{\Phi})_m = \tilde{r}_m e^{i\tilde{\theta}_m}.
    \end{equation}
    Therefore, to prove \eqref{eq:det_phi_t_polar_arctan}, it suffices to prove that $\tilde{\theta}$ and $\tilde{r}$ are well defined . 
    
    First, $\tilde{r}_m = \left|\det(\bo{\Phi})_m\right| \underset{m \to \infty}{\longrightarrow} \left|\det(\bo{\Phi})\right| = \tilde{r}$, which concludes for $\tilde{r}$.
    The same type of argument doesn't work for $\tilde{\theta}_m$. In fact, for example, the sequence $u_n := e^{2i\pi n}$ satisfies $u_n = 1$, so is convergent, but $2 \pi n$ is not.\\
    However, as $|\arctan(x)| \leq |x|$ for $x \in \mathbb{R}$ and $\left|\frac{y_n}{x_n}\right| \leq |y_n|$ (recall that $x_n \geq 1$ from \eqref{ineq:a_n_t>1}), we have that
    \begin{align} \label{ineq:arctan_leq_ima(a)_lambda_n}
        \left| \arctan\left(\frac{y_n}{x_n}\right) \right| &\leq |y_n| =2|\Im(a)|\lambda_n. 
    \end{align} 
    As $\tsig$ is of trace class, we have that
    \begin{equation}
        \sum_{n=1}^{\infty} \lambda_n < +\infty
    \end{equation}
    Thus, $\tilde{\theta}_m$ is convergent and $\tilde{\theta}$ is well defined. 
\end{proof}

\noindent The next lemma shows that the two polar representations for $\det(\bo{\Phi}_t(u,w))$, given in \eqref{eq:det_phi_t_polar_phi} and \eqref{eq:det_phi_t_polar_arctan}, are actually the same.

\begin{lemma} \label{lemma:two_C^0_equal_rep}
    Let $0 \leq t \leq T$ and $u, w \in \mc U$. Suppose that $\rho \Im(u) = 0$. Then
    \begin{equation}\label{eq:polar_rep_are_equal}
        \tilde{\theta}_t(u,w)=\theta_t(u,w).
    \end{equation}
\end{lemma}

\begin{proof}
    For the sake of readability, we do not specify any explicit dependence of the variables on $u$ and $w$. Suppose that $t \mapsto \tilde{\theta}_t$ is continuous. Then, as $t \mapsto \theta_t$ is continuous from Lemma \ref{lemma:phi_det_C^0} and the fact that $\theta_t = -2\Im(\phi_t)$, and given that both $\theta_t$ and $\tilde{\theta}_t$ are an argument of $\det\left(\bo{\Phi}_t\right)$, we have that $\theta_t$ and $\tilde{\theta}_t$ differ by a constant factor $2il_t\pi$ with $l_t \in \mathbb Z$. The final condition $\theta_T=\tilde{\theta}_T=0$ concludes that $l \equiv 0$, so that $\theta \equiv \tilde{\theta}$. 
    
    It stays to prove the continuity in time of $\tilde{\theta}$. First, as each eigenvalue is continuous in time, recalling Lemma \ref{lemma:lambda_n_C^0}, the general term of the series defining $\tilde{\theta}$ is made of continuous functions. We now show that this series converges uniformly in time, which will conclude the proof. We do it by refining inequality \eqref{ineq:arctan_leq_ima(a)_lambda_n}. Utilizing that $\Im(a(u,w)) = \Im(w) + \frac{1}{2}\Im(u^2-u)$ and the inequality $\lambda_{n,t} \leq \lambda_{n,0}$ (as per Lemma \ref{lemma:majo_lambda_n}), we obtain
    \begin{equation} \label{ineq:refined_ineq_arctan}
         \left| \arctan\left(\frac{y_{n,t}}{x_{n,t}}\right) \right| \leq \left(|\Im(w)| + \frac{1}{2}|\Im(u^2-u)|\right)\lambda_{n,0}, \quad t \leq T,
    \end{equation}
    where the right-hand term is independent of $t$, and the general term of a convergent series, since $\tsig_0$ is of trace class.
\end{proof}

The following and last lemma makes the link between the equation $e^{i\pi k_t}=-1$ and the value of $\theta$.

\begin{lemma} \label{lemma:number_crossings}
    Let $0 \leq t \leq T$, $(u,w) \in \mc U$. Suppose the conditions of Theorem \ref{thm:char_fun_sqrt_det_rot_count} are satisfied. Then,
    \begin{equation} \label{eq:number_crossings}
        e^{i\pi k_t(u,w)} = -1 \Longleftrightarrow \theta_t(u,w) \in \bigsqcup_{n \in \mathbb{Z}} \Bigl[\pi + 2n \cdot 2\pi, \hspace{0.1cm} \pi + (2n+1) \cdot 2\pi\Bigr)
    \end{equation}
    Moreover, fixing $\Re(u)$ and $\Re(w)$, we have:
    \begin{equation} \label{eq:no_crossing}
        \begin{aligned} 
        e^{i\pi k_t(u,0)} = 1, \; \forall \Im(u) \in \mathbb R &\Longleftrightarrow -\pi \leq \theta_t(u,0) < \pi, \; \forall \Im(u) \in \mathbb R, \\
        e^{i\pi k_t(0,w)} = 1, \; \forall \Im(w) \in \mathbb R &\Longleftrightarrow -\pi \leq \theta_t(0,w) < \pi, \; \forall \Im(w) \in \mathbb R.
    \end{aligned}
    \end{equation}
\end{lemma}

\begin{proof}
    The proof of \eqref{eq:number_crossings} readily follows from Formulas \eqref{eq:k_t(u,w)} and \eqref{eq:det_exp_phi}, together with the fact that $e^{i\pi k_t} = -1$ if and only if $k_t$ is odd. 
    
    Concerning \eqref{eq:no_crossing}, the implication $\Leftarrow$ is immediate from \eqref{eq:number_crossings}. Concerning the implication $\Rightarrow$, we only treat the case of $\Im(u)$, without loss of generality. Set $s=\Re(u)$, and suppose that $e^{i\pi k_t(u,0)} = 1$,  for all $\Im(u) \in \mathbb R$. Assume for contradiction that there exists $\Im(u_1) \in \mathbb R$ such that $u_1=s+i\Im(u_1)$ verifies $\theta_t(u_1,0) \notin [-\pi, \pi)$. Then, by continuity of $\Im(u) \mapsto \theta_t(s+i\Im(u),0)$ (again from Lemma \ref{lemma:phi_det_C^0} and the fact that $\theta_t(u) = -2\Im(\phi_t(u))$), and since $\theta(s, 0)=0$, there exists $\Im(u_0) \in (0, \Im(u_1)]$, such that $u_0=s+i\Im(u_0)$ verifies $\theta_t(u_0,0) \in \bigsqcup_{n \in \mathbb{Z}} \Bigl[\pi + 2n \cdot 2\pi, \hspace{0.1cm} \pi + (2n+1) \cdot 2\pi\Bigr)$, and thus $e^{i\pi k_t(u_0,0)} = -1$ from \eqref{eq:number_crossings}, which is a contradiction. 
\end{proof}

 \subsection{Putting everything together} \label{subsubsec:put_every_together}
 
We recall that the tangent function is increasing and odd, and it is a bijection from $(-\frac{\pi}{2}, \frac{\pi}{2})$ to $\mathbb{R}$, with its reciprocal, the arctangent function, also being increasing and odd, and a bijection from $\mathbb{R}$ to $(-\frac{\pi}{2}, \frac{\pi}{2})$.

\begin{proof}[Proof of Theorem \ref{thm:crossing_neg_axis_integrated_variance}]
    From this point forward, we assume $u = 0$ and fix $\Re(w) \leq 0$. Additionally, we fix $0 \leq t \leq T$. For simplicity and to enhance readability, we omit explicit references to $t$ and $u$ in the notation, and we will state the dependence on $w$ only if the variable depends on $\Im(w)$.
    We separate the proof into two cases.\\
    
    \noindent {\textbullet} \textbf{Case 1}. $N(\tsig) \in \{1,2\}$. 
    
    If $N(\tsig) \in \{1,2\}$, it follows from \eqref{eq:tilde_theta} and the fact that $\lambda_n = 0$ for $n > N(\tsig)$, that $|\theta(w)| < \frac{\pi}{2}N(\tsig) \leq \pi, \; \Im(w) \in \mathbb R$, which concludes the proof for this case by Lemma \ref{lemma:number_crossings}. \\
    
    \noindent {\textbullet} \textbf{Case 2}. $N(\tsig) \geq 3$. 
    
    For $n \in \mathbb{Z}$, set
    \begin{equation} \label{eq:Pi_n}
        \Pi_n := \left[\pi + 2n \cdot 2\pi, \hspace{0.1cm} \pi + (2n+1) \cdot 2\pi\right).
    \end{equation}
    From Lemma \ref{lemma:phi_det_C^0}, the map $\Im(w) \in \mathbb{R} \mapsto \theta(w) = -2\Im(\phi(w))$ is continuous. Suppose that it is also strictly decreasing with the following asymptotic behavior:
    \begin{equation} \label{eq:asymptotic_theta}
        \lim_{\Im(w) \to -\infty} \theta(w) = \frac{\pi}{2}N(\tsig), \quad \lim_{\Im(w) \to +\infty} \theta(w) = -\frac{\pi}{2}N(\tsig).
    \end{equation} 
    Therefore, $\theta$ is a bijection from $\mathbb R$ to $(-\frac{\pi}{2}N(\tsig),\frac{\pi}{2}N(\tsig))$. Since $N(\tsig) > 3$, we have $\frac{\pi}{2} N(\tsig) > \pi$, which implies, from \eqref{eq:asymptotic_theta} and \eqref{eq:Pi_n}, that there exists an integer $M_0 \in \mathbb{N}^* \cup \{+\infty\}$ such that for all $n$ satisfying $-M_0 \leq n \leq M_0 - 1$, the set $\theta^{-1}(\Pi_n)$ is an interval with strictly positive length, while for all other values of $n$, $\theta^{-1}(\Pi_n)$ is at most a singleton. Therefore, it follows from
    \begin{equation}
        \theta^{-1} \left(\bigsqcup_{n \in \mathbb{Z}} \Pi_n\right) = \bigsqcup_{n \in \mathbb{Z}} \theta^{-1} (\Pi_n)
    \end{equation}
     that $\theta^{-1} \left(\bigsqcup_{n \in \mathbb{Z}} \Pi_n\right)$ is the union of disjoint intervals with strictly positive length.
     Therefore, by Lemma \ref{lemma:number_crossings}, the equality $e^{i\pi n(w)} = -1$ holds on a union of disjoint intervals with strictly positive length. Furthermore, if $N(\tsig) = +\infty$, then, from the asymptotic behavior \eqref{eq:asymptotic_theta}, $M_0=+\infty$, meaning that this union of intervals is infinite and unbounded. 

    It remains to prove that $\Im(w) \mapsto \theta(w)$ is strictly decreasing on $\mathbb{R}$, and verifies $\lim_{\Im(w) \to -\infty} \theta(w) = \frac{\pi}{2}N(\tsig)$, $\lim_{\Im(w) \to +\infty} \theta(w) = -\frac{\pi}{2}N(\tsig)$. We do it by using that $\theta(w)=\tilde{\theta}(w)$ from Lemma \ref{lemma:two_C^0_equal_rep} and the explicit expression \eqref{eq:tilde_theta} for $\tilde{\theta}$:
    \begin{equation} \label{eq:theta(w)}
        \theta(w) = \sum_{n=1}^{+\infty} \arctan\left(\frac{y_n(w)}{x_n}\right) = \sum_{n=1}^{+\infty} \arctan\left(\frac{-2\Im(w)\lambda_n}{1-2\Re(w)\lambda_n}\right),
    \end{equation}
    where only $y_n(w)$ depends on $\Im(w)$, so that we wrote $x_n(w)=x_n$.

    Concerning the monotonicity of $\theta$, note that since $\arctan$ is strictly increasing on $\mathbb{R}$, and that for any $n$, $y_n$ is decreasing in $\Im(w)$, then so is $\theta$. 
    
    We now prove that $\lim_{\Im(w) \to -\infty} \theta(w) = \frac{\pi}{2}N(\tsig)$ and $\lim_{\Im(w) \to +\infty} \theta(w) = -\frac{\pi}{2}N(\tsig)$. First, it follows from the fact that $\Re(w) \leq 0$, that for any nonnegative integer $n$, we have $\lambda_n \geq 0$ and $x_n \geq 1$, so that $\arctan\left(\frac{y_n(w)}{x_n}\right)$ has the same sign than $-\Im(w)$, which does not depend on $n$. Furthermore, taking the limit in $\pm \infty$ in \eqref{eq:theta(w)} shows that:
    \begin{equation} \label{eq:limit_arctan}
        \lim_{\Im(w) \to \pm \infty} \arctan\left(\frac{y_n(w)}{x_n}\right) = \mp \frac{\pi}{2}.
    \end{equation}
    Combining the constant sign and monotonicity of $\Im(w) \mapsto \arctan\left(\frac{y_n(w)}{x_n}\right)$, as well as the limit in \eqref{eq:limit_arctan}, we can conclude using the monotone convergence theorem, that the desired limits hold.

    Finally, we need to prove inequality \eqref{ineq:first_crossing_int_variance}. Let $r>0$ and $N_r \geq 3$ be the number of eigenvalues of $\tsig$ greater than $r$. Then, using that
    \begin{align}
        |\theta(w)| 
        &= \sum_{n=1}^{+\infty} \arctan\left(\frac{2|\Im(w)|\lambda_n}{1-2\Re(w)\lambda_n}\right) \\
        & \geq \sum_{n=1}^{N_r} \arctan\left(\frac{2|\Im(w)|\lambda_n}{1-2\Re(w)\lambda_n}\right) \\
        & \geq N_r \arctan\left(\frac{2|\Im(w)|}{\frac{1}{r}-2\Re(w)}\right),
    \end{align}
    We conclude that if $N_r \arctan\left(\frac{2|\Im(w)|}{\frac{1}{r}-2\Re(w)}\right) > \pi$, then $|\theta(w)| > \pi$. Equivalently, since $N_r \geq 3$, if $\Im(w) > \tan\left(\frac{\pi}{N_r}\right) \left(\frac{1}{2r}-\Re(w)\right)$, then $|\theta(w)| > \pi$. This combined with Lemma \ref{lemma:number_crossings} concludes the proof. 
\end{proof}

We now give the proof of Theorem \ref{thm:crossing_neg_axis_log_price}.

\begin{proof}[Proof of Theorem \ref{thm:crossing_neg_axis_log_price}]
    From this point forward, we assume $w = \rho = 0$ and fix $\Re(u) \leq 0$. Consequently, any dependence on $u$ is reduced to a dependence on $\Im(u)$. Additionally, we fix $0 \leq t \leq T$. For simplicity and to enhance readability, we omit explicit references to $t$ and $w$ in the notation, and will, when there is no ambiguity, identify $u$ to $\Im(u)$. 
    
    We separate the proof into two cases. Note that as we suppose $\rho=0$, then $\tsig$ isn't dependent on $u$, as well as its eigenvalues. Moreover, since $\frac{1}{\nu^2}\tsig$ is independent of $\nu$, we will consider, only in the proof of the first case, the eigenvalues $\lambda_n$ of $\frac{1}{\nu^2}\tsig$ instead of the eigenvalues of $\tsig$. This change of notation will help determining the vol-of-vol threshold. \\

    In this case, we have $x_n(u) = \lambda_n \nu^2 \Im(u)(1-2\Re(u))$ and $y_n(u) = 1 +  \lambda_n \nu^2\left(\Im(u)^2 + \Re(u)(1 - \Re(u))\right)$, so that
    \begin{equation} \label{eq:theta_t_log_price}
        |\theta(u)| = \sum_{n=1}^{\infty} \arctan\left(\frac{\lambda_n \nu^2|\Im(u)||1-2\Re(u)|}{1 +  \lambda_n \nu^2\left( \Im(u)^2 + \Re(u)(1 - \Re(u)) \right)}\right), \quad \Im(u) \in \mathbb{R}.
    \end{equation}

    \noindent \textbf{Case 1}. $N(\tsig) \in \{1,2\}$ or $\Re(u)=\frac{1}{2}$ or $\nu < \nu^*$. 

    The proof for the case $N(\tsig) \in \{1,2\}$ is the exact same than in the proof of Theorem \ref{thm:crossing_neg_axis_integrated_variance}. If $\Re(u)=\frac{1}{2}$, the result comes from the fact that $1-2\Re(u) = 0$ in \eqref{eq:theta_t_log_price} so that $\theta \equiv 0 < \pi$, and we conclude again with Lemma \ref{lemma:number_crossings}.

    Let's now prove the existence of the threshold $\nu^*$. The idea here is to see $\theta$ as a function of $u$ and $\nu$, and to prove that for $\nu$ small enough, we have $|\theta(u,\nu)| < \pi$ for any $\Im(u) \in \mathbb R$. 
    
    Suppose that we have proven that $\theta(u,\nu) \underset{\nu \to 0^+}{\longrightarrow} 0$ uniformly in $\Im(u)$. Then, there exists $\nu_{\pi} > 0$ such that
    \begin{equation}
        |\theta(u,\nu)| < \pi, \quad \Im(u) \in \mathbb{R}, \; \nu \leq \nu_{\pi}.
    \end{equation}
    Set $\nu^* := \sup\{\nu_{\pi} > 0 \;|\; |\theta(u,\nu)| < \pi, \; \forall \Im(u) \in \mathbb{R}, \; \nu \leq \nu_{\pi}\}$. We have that $\nu^* \in (0,+\infty]$, and noticing that for any value of $\Im(u)$, $\nu \mapsto |\theta(u,\nu)|$ is strictly nondecreasing, we conclude that
    \begin{equation}
        \nu < \nu^* \Longrightarrow |\theta(u,\nu)| < \pi, \; \forall \Im(u) \in \mathbb{R},
    \end{equation}
    or equivalently from Lemma \ref{lemma:number_crossings}
    \begin{equation}
        \nu < \nu^* \Longrightarrow e^{i \pi k_t(u)}=1, \; \forall \Im(u) \in \mathbb{R}.
    \end{equation}
    It stays to prove that $\theta(u,\nu) \underset{\nu \to 0^+}{\longrightarrow} 0$ uniformly in $\Im(u)$. This will be done in three steps. First, we will prove that $(\Im(u),\nu) \mapsto \theta(u,\nu)$ is continuous over $\mathbb{R} \times \mathbb{R}_+$. Then, that $|\theta(u,\nu)| \underset{|\Im(u)| \longrightarrow +\infty}{\longrightarrow} 0$ uniformly in $\nu$ over any compact of $\mathbb{R}_+$. Finally, we will combine these two steps to complete the proof. \\
    \emph{First step}. $(\Im(u),\nu) \mapsto \theta(u,\nu)$ is continuous over $\mathbb{R} \times \mathbb{R}_+$
    
    Since $\theta=-2\Im(\phi)$, the continuity of $\phi$ suffices. The continuity in $\Im(u)$ comes from Lemma \ref{lemma:phi_det_C^0}, while the on in $\nu$ is a direct adaptation of the proof of the same lemma.\\
    \emph{Second step}. $|\theta(u,\nu)| \underset{|\Im(u)| \longrightarrow +\infty}{\longrightarrow} 0$ uniformly in $\nu$ over any compact of $\mathbb{R}_+$
    
    Let $V > 0$. As already mentioned, for any $\Im(u) \in \mathbb R$, the map $\nu \mapsto \theta(u,\nu)$ is strictly nondecreasing. Moreover, we have
    \begin{equation}
        \frac{ \lambda_n V^2|\Im(u)||1-2\Re(u)|}{1 +  \lambda_n V^2\left( \Im(u)^2 + \Re(u)(1-\Re(u))\right)} \longrightarrow 0, \quad |\Im(u)| \longrightarrow +\infty.
    \end{equation}
    Since for $a,b,c \in \mathbb R_+^*$, the map $x \in \mathbb R_+ \mapsto \frac{ax}{b+cx^2}$ is strictly decreasing over $\left[\sqrt{\frac{b}{2c}}, +\infty\right)$, then the map $|\Im(u)| \in \mathbb R_+ \mapsto \frac{ \lambda_n V^2|\Im(u)||1-2\Re(u)|}{1 +  \lambda_n V^2\left( \Im(u)^2 + \Re(u)(1-\Re(u))\right)}$ is eventually strictly decreasing. Therefore, by the monotone convergence theorem, we obtain that for any $\nu \leq V$:
    \begin{equation}
        |\theta(u,\nu)| \leq |\theta(u, V)| \underset{|\Im(u)| \longrightarrow +\infty}{\longrightarrow} 0.
    \end{equation} 
    \emph{Third step}. We complete the proof 
    
    Let $\epsilon > 0$. By the second point, there exists $U > 0$ such that for $|\Im(u)| > U$ and $\nu \leq V$,
    \begin{equation}
        |\theta(u,\nu)| \leq \epsilon.
    \end{equation}
    Since $\theta_t$ is uniformly continuous over $[-U, U] \times [0, V]$ and $\theta_t(u, 0) = 0$, we can choose $0 < v < V$ such that
    \begin{equation}
        |\theta(u,\nu)| \leq \epsilon, \quad |\Im(u)| < U, \; \nu \leq v.
    \end{equation}
    Combining the last two inequalities, we conclude that
    \begin{equation}
        |\theta(u,\nu)| \leq \epsilon, \quad |\Im(u)| \in \mathbb{R}, \; \nu \leq v.
    \end{equation}

    \noindent \textbf{Case 2}. $N(\tsig) \geq 3$, $\Re(u) \neq \frac{1}{2}$ and $\nu > \nu^*$. 

    We now prove that if $\nu$ exceeds the threshold $\nu^*$, then $e^{i\pi n(u)} = -1$ for $\Im(u)$ on a finite union of disjoint, non-empty intervals, and that inequality \eqref{ineq:first_crossing_log_price} is satisfied.
    As in the proof of Theorem \ref{thm:crossing_neg_axis_integrated_variance}, we introduce, for $n \in \mathbb{Z}$, the set
    \begin{equation}
        \Pi_n := \left[\pi + 2n \cdot 2\pi, \hspace{0.1cm} \pi + (2n+1) \cdot 2\pi\right).
    \end{equation}
    
    Since $\nu \mapsto \theta(u,\nu)$ is strictly nondecreasing, and by the definition of $\nu^*$, there exists some $|\Im(u)| > 0$ such that $|\theta(u)| > \pi$. Therefore, since $\Im(u) \mapsto \theta(u)$ is continuous and odd, the intermediate value theorem ensures the existence of an integer $M_0 \in \mathbb{N}^* \cup \{+\infty\}$ such that for all $n$ satisfying $-M_0 \leq n \leq M_0 - 1$, the set $\theta^{-1}(\Pi_n)$ contains an interval with strictly positive length, while for all other values of $n$, $\theta^{-1}(\Pi_n)$ is at most a singleton. Therefore, it follows from
    \begin{equation}
        \theta^{-1} \left(\bigsqcup_{n \in \mathbb{Z}} \Pi_n\right) = \bigsqcup_{n \in \mathbb{Z}} \theta^{-1} (\Pi_n)
    \end{equation}
     that $\theta^{-1} \left(\bigsqcup_{n \in \mathbb{Z}} \Pi_n\right)$ contains a union of disjoint intervals with strictly positive length.
     Therefore, by Lemma \ref{lemma:number_crossings}, the equality $e^{i\pi n(u)} = -1$ holds on a union of disjoint intervals with strictly positive length. Moreover, since, as shown previously, $|\theta(u)| \underset{|\Im(u)| \longrightarrow +\infty}{\longrightarrow} 0$, $\Im(u) \mapsto \theta(u)$ is bounded, meaning that $M_0$ must be finite, or equivalently, that the union must be finite.
    
    Next, we prove inequality \eqref{ineq:first_crossing_log_price}. Returning to the standard notation, let $(\lambda_n)_n$ denote the eigenvalues of $\tsig$, rather than those of $\frac{1}{\nu^2} \tsig$. Let $r > 0$ be such that $N_r \geq 3$. We have
    \begin{align}
        |\theta(u)| 
        &= \sum_{n=1}^{\infty} \arctan\left(\frac{\lambda_n |\Im(u)||1-2\Re(u)|}{1 +  \lambda_n \left( \Im(u)^2 + \Re(u)(1 - \Re(u)) \right)}\right) \\
        & \geq \sum_{n=1}^{N_r} \arctan\left(\frac{\lambda_n |\Im(u)||1-2\Re(u)|}{1 +  \lambda_n \left( \Im(u)^2 + \Re(u)(1 - \Re(u)) \right)}\right) \\
        & \geq N_r \arctan\left(\frac{r|\Im(u)||1-2\Re(u)|}{1 +  r\left( \Im(u)^2 + \Re(u)(1 - \Re(u)) \right)}\right),
    \end{align}
    we conclude that $|\theta(u)| > \pi$ if $N_r \arctan\left(\frac{r|\Im(u)||1-2\Re(u)|}{1 +  r\left( \Im(u)^2 + \Re(u)(1 - \Re(u)) \right)}\right) > \pi$, or equivalently, since $N_r \geq 3$, if the second order polynomial inequality $~|\Im(u)|^2 -|\Im(u)|\frac{|1-2\Re(u)|}{\tan(\frac{\pi}{N_r})} +\frac{1}{r} + \Re(u)(1-\Re(u)) < 0$ is satisfied. As $0 \leq \Re(u) \leq 1$, the lowest root of the polynomial is positive, meaning that the inequality is satisfied for $\Im(u)$ between the roots of the polynomial, and in particular for some $\Im(u)$ such that:
    \begin{equation}
        |\Im(u)| > \frac{|1-2\Re(u)|}{2\tan\left(\frac{\pi}{N_r}\right)}-\sqrt{\left(\frac{1-2\Re(u)}{2\tan\left(\frac{\pi}{N_r}\right)}\right)^2 -4\left(\frac{1}{r} + \Re(u)(1-\Re(u))\right)}.
    \end{equation}
\end{proof}

We finish this section by the proof of Corollary \ref{cor:crossing_neg_axis_log_price}.

\begin{proof}[Proof of Corollary \ref{cor:crossing_neg_axis_log_price}]
    We will, when necessary, state the dependence of $\tsig_t$ and $\theta_t$ to $\rho$.
    
    In the same spirit as in the proof of Lemma \ref{lemma:phi_det_C^0}, we can show that $\phi_t$ is continuous as a function of $(u, \rho)$, which implies that $\theta_t = -2\Im(\phi_t)$ is also continuous with respect to $u$ and $\rho$. Similarly, we can demonstrate that $(u,\rho) \mapsto \tsig_t$ is continuous. Consequently, by applying \citet[Theorem 4.2]{gohberg69}, each eigenvalue $\lambda_{n,t}$ of $\tsig_t$ is continuous as a function of $(u, \rho)$. 

    \noindent \textbf{Case 1}. $N(\tsig_t(\rho=0)) \in \{1,2\}$ or $\Re(u)=\frac{1}{2}$ or $\nu \leq \nu^*$.
    
    Then, for $M>0$, it follows from Heine Theorem that $(\Im(u),\rho) \mapsto \theta_t(u,\rho)$ is uniformly continuous over the compact $[-M,M] \times [-1,1]$. Setting $|\theta_t(u_{\text{max}}, 0)| := \max_{|\Im(u)|\leq M}|\theta_t(u,0)|$, we have from Theorem \ref{thm:crossing_neg_axis_log_price} that $|\theta_t(u_{\text{max}}, 0)| <\pi$. Taking $\varepsilon>0$ such that for any $|\rho| \leq \varepsilon$ and $|\Im(w)| \leq M$, $|\theta_t(u, \rho)-\theta_t(u, 0)|<\pi - \theta_t(w_{\text{max}}, 0)$, it follows from the triangle inequality that for such $\rho$ and $\Im(u)$, $|\theta_t(u,\rho)<|\pi|$, which combined with Lemma \ref{lemma:number_crossings}, completes the proof in this case. 

    \noindent \textbf{Case 2}. $N(\tsig_t(\rho=0)) \geq 3$, $\Re(u)\neq \frac{1}{2}$ and $\nu>\nu^*$.
    
    As in Case 1, we can choose $\rho$ sufficiently small such that $\theta(u, \rho)$ is as close as desired to $\theta(u, 0)$ uniformly for $\Im(u)$ within a compact set $[-M, M]$. Therefore, from Lemma \ref{lemma:number_crossings} and Theorem \ref{thm:crossing_neg_axis_log_price}, for sufficiently small $\rho$, $e^{i\pi k_t(u)}$ will also take the value $-1$ on a finite union of disjoint intervals with strictly positive length, each bounded by $M$. 
    
    We now demonstrate the inequality satisfied by $\Im(u^*)$. Let $r>0$ and $N_r$ be the number of eigenvalues of $\tsig_t(\rho=0)$ greater than $r$. From the proof of Theorem \ref{thm:crossing_neg_axis_log_price}, $\theta(u,0)$ is greater than $\pi$ when $\Im(u)$ lies between the two roots of a second order polynomial whose lowest root is positive and equals: 
    \begin{equation}
        \Im(u_r) := \frac{|1-2\Re(u)|}{2\tan\left(\frac{\pi}{N_r}\right)}-\sqrt{\left(\frac{1-2\Re(u)}{2\tan\left(\frac{\pi}{N_r}\right)}\right)^2 -4\left(\frac{1}{r} + \Re(u)(1-\Re(u))\right)}.
    \end{equation}
    For $\delta>0$ such that $\Im(u_r)+\delta$ still lies between the two roots, there exists $\varepsilon>0$ such that for any $|\rho| \leq \varepsilon$, $|\theta(u_r+i\delta, \rho) - \theta(u_r+i\delta,0)| < \theta(u_r+i\delta,0) - \pi$. Thus, $|\theta(u_r+i\delta, \rho)| > \theta(u_r+i\delta, 0) - |\theta(u_r+i\delta, 0)-\theta(u_r+i\delta, \rho)|>\pi$. Thus $\Im(u^*) \leq \Im(u_r) + \delta$, which is the desired inequality.
\end{proof}

\appendix

\section{Spectrum, trace and determinant of compact linear operators} \label{sec:recall_trace_det}

In this section, we recall classical results on operator theory in Hilbert spaces regarding their trace and their determinant. For further details, we refer to \cite{gohberg69}, \cite{simon1977notes, simon2005trace}, \cite{smithies1958integral}, and \citet[Section 2 and 3]{bornemann2010numerical}.

We provide the definition of the rank of a bounded linear operator, along with an important result concerning the spectrum and eigenvalues of compact operators. 

\begin{definition}
    Let $\bo{A}$ be a bounded operator on an Hilbert space $(\mc{H}, \langle \cdot, \cdot \rangle)$. Let $Im(\bo{A})$ be the set of its values. The rank of $\bo{A}$ is the dimension of the subspace $\overline{Im(\bo{A})}$, where the closure is taken with respect to the uniform topology. It is denoted by $N(\bo{A})$.
\end{definition}

\begin{proposition} \label{prop:compact_op_ppt}
    Suppose $\bo{A}$ is a compact operator on an Hilbert space $\mc{H}$, then:
\begin{itemize}
    \item[1.] The spectrum of $\bo{A}$ is at most countable. Each point of the spectrum is isolated, with the possible exception of $0$.
    \item[2.] If the space $\mc{H}$ is of infinite dimension, then $0$ belongs to the spectrum of $\bo{A}$.
    \item[3.] Any nonzero complex number in the spectrum of $\bo{A}$ is an eigenvalue and the corresponding eigenspace is of finite dimension.
    \item[4.] If it is countably infinite, any enumeration $\{\lambda_1, \lambda_2, … \}$ of the spectrum is a sequence with limit 0.
\end{itemize}
\end{proposition}
The trace and the determinant are two important functionals on the space of compact operators. Such quantities are defined for operators of trace class. A compact operator $\bo{A}$ is said to be of trace class if the quantity
\begin{equation} 
    \Tr |\bo{A}| := \sum_{n \geq 1} \langle |\bo{A}| v_n, v_n \rangle
\end{equation}
is finite for a given orthonormal basis $(v_n)_{n \geq 1}$, where $|\bo{A}|:=\sqrt{\bo{A}\bo{A}^*}$. In case the quantity above is finite, then
\begin{equation} \label{eq:Trace_of_trace_class_op}
    \Tr \bo{A} := \sum_{n \geq 1} \langle \bo{A} v_n, v_n \rangle
\end{equation}
is well defined, and it can be shown that the quantity on the right-hand side of \eqref{eq:Trace_of_trace_class_op} is independent of the choice of the orthonormal basis and will be called the trace of the operator $\bo{A}$. Furthermore, Lidskii’s theorem \citet[Theorem 3.7]{simon2005trace} ensures that
\begin{equation}
    \Tr \bo{A} = \sum_{n=1}^{N(\bo{A})} \lambda_n(\bo{A}).
\end{equation}
The trace norm of a trace class operator $\bo{A}$ id defined by 
\begin{equation}
    \doublenorm{\bo{A}}_1 := \Tr(|\bo{A}|).
\end{equation}
The trace norm is a norm on the space of trace class operators. Moreover, if $\bo{A}$ is a trace class operator and $\bo{B}$ is a bounded operator, then $\bo{AB}$ is a trace class operator and we have 
\begin{equation}
    \doublenorm{\bo{AB}}_1 \leq \doublenorm{\bo{A}}_1 \doublenorm{\bo{B}}
\end{equation}

Furthermore, the equivalence
\begin{equation}
    \prod_{n \geq 1} (1 + |\lambda_n|) < \infty \iff \sum_{n \geq 1} |\lambda_n| < \infty,
\end{equation}
allows one to define a determinant functional for a trace class operator $\bo{A}$ by
\begin{equation}
    \det(\id + z \bo{A}) := \prod_{n=1}^{N(\bo{A})} \left( 1 + z \lambda_n(\bo{A}) \right),
\end{equation}
for all $z \in \mathbb{C}$. If in addition $\bo{A}$ is an integral operator induced by a continuous kernel $A \in L^2([0, T]^2, \mathbb{C})$, then one can show that
\begin{equation} \label{eq:fredholm_det_integral_operator}
    \det(\id + z \bo{A}) = \sum_{n \geq 0} \frac{z^n}{n!} \int_0^T \dots \int_0^T 
    \det \big[ \big( A(s_i, s_j) \big)_{1 \leq i, j \leq n} \big] ds_1 \dots ds_n.
\end{equation}

The determinant \eqref{eq:fredholm_det_integral_operator} is named after \cite{Fredholm1903}, who defined it for the first time for integral operators with continuous kernels. The product of a bounded operator and a Trace class operator is a trace class operator.

A compact operator $\bo{A}$ is said to be an Hilbert-Schmidt operator if the quantity $\Tr(\bo{A}\bo{A^*})$ is finite. A trace class operator is an Hilbert-Schmidt operator, and whenever $\bo{K}$ is a linear operator induced by a kernel $K \in L^2([0, T]^2, \mathbb{C})$, $\bo{K}$ is an Hilbert-Schmidt. Furthermore, the product of two Hilbert-Schmidt operators $\bo{K}$ and $\bo{L}$ is of trace class.

\section{Proof of the invertibility of $\id - 2a\tsig_t$} \label{subsec:invertiblity}

\begin{lemma} \label{lemma:invertibility}
    Let $0 \leq t \leq T$, $(u,w) \in \mc U$, and $K$ be a Volterra kernel of continuous and bounded type in $L^2_{\mathbb R}$. Then, $\id-2a(u,w)\tsig_t(u)$ is invertible and $\bo{\Psi}_t(u,w)$ is well-defined.
\end{lemma}

\begin{proof}
    For the sake of readability, we drop the dependence in $(u,w)$ in the notations. Setting $K_t(s,z) := K(s,z) \mathbbm{1}_{z\geq t}$, for $0 \leq s,z \leq T$, we have from \citet[Lemma A.5.]{abi2022characteristic} that 
    \begin{equation}
        \tsig_t = \nu^2\left(\id -b\bo{K}_t\right)^{-1} \bo{K}_t \bo{K}_t^* \left(\id -b\bo{K}_t^*\right)^{-1}.
    \end{equation}
    As $K_t$ is a Volterra kernel, then for any $z \in \mathbb C$, $\id-z\bo{K}_t$ and $\id-z\bo{K}_t^*$ are invertible, so that we have
    \begin{align}
        \tsig_t 
        &= \nu^2\left(\id - b\bo{K}_t\right)^{-1} \bo{K}_t \bo{K}_t^* \left(\id - b\bo{K}_t^*\right)^{-1} \\
        &= \nu^2\left(\id - \Re(b)\bo{K}_t - i\Im(b)\bo{K}_t\right)^{-1} \bo{K}_t \bo{K}_t^*
           \left(\id - \Re(b)\bo{K}_t^* - i\Im(b)\bo{K}_t^*\right)^{-1} \\
        &= \nu^2\left(\id - i\Im(b)\left(\id - \Re(b)\bo{K}_t\right)^{-1}\bo{K}_t\right)^{-1} 
           \left(\id - \Re(b)\bo{K}_t\right)^{-1} \bo{K}_t \\
        &\quad \cdot \bo{K}_t^* \left(\id - \Re(b)\bo{K}_t^*\right)^{-1} 
           \left(\id - i\Im(b)\bo{K}_t^*\left(\id - \Re(b)\bo{K}_t^*\right)^{-1}\right)^{-1} \\
        &= \nu^2\left(\id -i\Im(b)\bo{\tilde{K}}_t\right)^{-1} \bo{\tilde{K}}_t \bo{\tilde{K}}_t^* \left(\id -i\Im(b)\bo{\tilde{K}}_t^*\right)^{-1}.
    \end{align}
    where we set $\bo{\tilde{K}}_t := \left(\id-\Re(b)\bo{K}_t\right)^{-1}\bo{K}_t$. Moreover, again from \citet[Lemma A.2.]{abi2022characteristic}, the resolvent kernel $R_t^b$ of $\Re(b)\bo{K}_t$ is also a Volterra kernel of continuous and bounded type in $L^2_{\mathbb R}$, so that from \citet[Example 3.2., (iv)]{abi2022characteristic}, $\tilde{K}_t = K_t + R_t^b \star K_t$ is also a Volterra kernel of continuous and bounded type in $L^2_{\mathbb R}$. Therefore, we can, without loss of generality, reduce the proof of the lemma to the case $\Re(b)=0$. 

    In this case, we write $b=i\Im(b)$, and it follows from the invertibility of $\id -i\Im(b)\bo{K}_t$ and $\id -i\Im(b)\bo{K}_t^*$ that the invertibility of $\id-2a\tsig_t$ is equivalent to the one of 
    \begin{align} 
        \bo{M}_t 
        &:= \left(\id -i\Im(b)\bo{K}_t\right) (\id-2a\tsig_t)\left(\id -i\Im(b)\bo{K}_t^*\right) \\
        &= \id -i\Im(b)(\bo{K}_t+\bo{K}_t^*) -(\Im(b)^2+2a\nu^2)\bo{K}_t\bo{K}_t^*. \label{eq:M_t}
    \end{align}
    Therefore, we need to prove that $0$ does not belong to the spectrum of $\bo{M}_t$, or equivalently that $1$ does not belong to the spectrum of $\id - \bo{M}_t$. Since $\id - M_t$ is compact, it follows from Point 3 of Proposition \ref{prop:compact_op_ppt} that it suffices to show that $1$ is not an eigenvalue of $\id - \bo{M}_t$. Denote the eigenvalues of $\id - \bo{M}_t$ by $(\lambda_{n,t})_{n \geq 1}$, counting multiplicity, and ordering them as
    \begin{equation}
        |\lambda_{1,t}| \geq |\lambda_{2,t}| \geq \dots \geq 0.
    \end{equation} 
    Let $n \geq 1$ and $f_n$ be an eigenvector associated to $\lambda_{n,t}$ with norm 1. We will prove that $\Re(\lambda_{n,t}) \leq 0$ which will conclude. On the one hand, we have
    \begin{equation} \label{eq:(id-M)f,f_1}
        \langle \left(\id - \bo{M}_t\right) f_n, f_n\rangle_{L^2_{\mathbb C}} = \lambda_{n,t} \lVert f_n\rVert_{L^2_{\mathbb C}}^2 = \lambda_{n,t},
    \end{equation}
    and on the other hand, we have from \eqref{eq:M_t} that
    \begin{equation} \label{eq:(id-M)f,f_2}
        \langle \left(\id - \bo{M}_t\right) f_n, f_n\rangle_{L^2_{\mathbb C}} 
        = i\Im(b)\langle(\bo{K}_t+\bo{K}_t^*)f_n,f_n\rangle_{L^2_{\mathbb C}} +(\Im(b)^2+2a\nu^2)\langle\bo{K}_t\bo{K}_t^* f_n, f_n\rangle_{L^2_{\mathbb C}} 
    \end{equation}
    Equalizing \eqref{eq:(id-M)f,f_1} and \eqref{eq:(id-M)f,f_2} and taking real parts yields
    \begin{equation}
        \Re(\lambda_{n,t}) = (\Im(b)^2 + 2\Re(a)\nu^2)\lVert\bo{K}_t^* f_n\rVert_{L^2_{\mathbb C}}.
    \end{equation}
    Finally, since $0 \leq \Re(u) \leq 1$, $0 \leq \Re(w)$, and $|\rho|\leq 1$, then 
    \begin{equation}
        \Im(b)^2 + 2\Re(a)\nu^2 =\nu^2((\rho^2-1)\Im(u)^2 + 2\Re(w) +  \Re(u)(\Re(u)-1)) \leq 0,
    \end{equation}
    which concludes.
\end{proof}

\section{Proof of Proposition \ref{prop:upper_bound_lip_cst}} \label{subsec:proof_prop_upper_bound_lip_cst}

In order to prove the proposition, which is given in the integrated variance case, we use the notations and results of Section \ref{sec:proof_thm_crossing}.

\begin{proof}[Proof of Proposition \ref{prop:upper_bound_lip_cst}]
    We have, recalling \eqref{eq:theta(w)},
    \begin{equation}
        \theta_t(w) = \sum_{n=1}^{+\infty} \arctan\left(\frac{-2\Im(w)\lambda_n}{1-2\Re(w)\lambda_n}\right)
    \end{equation}
    where $\Re(w) \leq 0$ and $\lambda_n \geq 0$ for all $n \geq 1$. Let $w_1$ and $w_2$ two complex numbers with a common real part $s \leq 0$, and a positive imaginary part. Using the identity $\arctan(x)-\arctan(y) = \arctan\left(\frac{x-y}{1+xy}\right)$ for $x,y$ such that $xy>-1$, applied with $x = \frac{-2\Im(w_1)\lambda_n}{1-2s\lambda_n}$ and $y = \frac{-2\Im(w_2)\lambda_n}{1-2s\lambda_n}$ (they verify $xy\geq 0>-1$ since $\Im(w_1)$ and $\Im(w_2)$ are positive), yields
    \begin{equation}
        |\theta_t(w_1) - \theta_t(w_2)| \leq \sum_{n=1}^{+\infty} \arctan\left(\frac{2\lambda_n (1-2s\lambda_n)|\Im(w_1)-\Im(w_2)|}{(1-2s\lambda_n)^2 + 4\lambda_n\Im(w_1)\Im(w_2)}\right).
    \end{equation}
    Since $4\lambda_n \Im(w_1)\Im(w_2) \geq 0$, we obtain
    \begin{equation}
        |\theta_t(w_1) - \theta_t(w_2)| \leq \left(2\sum_{n=1}^{+\infty} \arctan\left(\frac{\lambda_n}{1-2s\lambda_n}\right)\right) |\Im(w_1)-\Im(w_2)|.
    \end{equation}
\end{proof}

\section{Proofs for Section~\ref{subsec:prefactor-free_approx}} \label{sec:proof_prefactor-free_approx}

\begin{proof}[Proof of Lemma \ref{lemma:two_integral_formulas_for_phi}]
    For the sake of readability, we omit any reference to $t$. Since from Lemma \ref{lemma:tsig_psi_well_def}, $\tsig$ is of trace class and $\id - 2a\tsig$ is invertible, it follows from \citet[IV.1.2.]{gohberg69} that
    \begin{equation}
        \det\left(\bo{\Phi}(u,w)\right) = \exp\left(- 2a(u,w)\int_0^1 \Tr\left(\tsig(u) \left(\id - 2sa(u,w)\tsig(u)\right)^{-1}\right) ds\right).
    \end{equation}
     From \citet[Chapter 9, Section 2, Exercise 20]{conway2019course} and the definition of $\tsig$,
     \begin{align}
         \Tr\left(\tsig(u) \left(\id - 2sa(u,w)\tsig(u)\right)^{-1}\right)
         &= \Tr\left(\left(\id - 2sa(u,w)\tsig(u)\right)^{-1} \tsig(u)\right) \\
         &= \Tr\left(\left(\id - b(u)\bo{K}^*\right)^{-1}\left(\id - 2sa(u,w)\tsig(u)\right)^{-1} \left(\id - b(u)\bo{K}\right)^{-1} \bo{\Sigma}\right) \\
         &= \frac{1}{a(u,w)} \Tr(\bo{\Psi}_s(u,w)\bo{\Sigma}).
     \end{align}
     Therefore, from Lemma \ref{lemma:det_exp_phi}, there exists $l(u,w) \in \mathbb Z$ such that 
     \begin{equation}
         \phi(u,w) - \int_0^1 \Tr\left(\bo{\Psi}_s(u,w)\bo{\Sigma}\right) ds = i\pi l(u,w).
     \end{equation}
     It follows from Lemma \ref{lemma:phi_det_C^0} that 
     \begin{equation}
         (u,w) \in \mc U \mapsto \phi(u,w)
     \end{equation}
     is continuous. The same reasoning than in the proof of the said lemma also proves that 
     \begin{equation}
         (u,w) \in \mc U \mapsto \int_0^1 \Tr\left(\bo{\Psi}_s(u,w)\bo{\Sigma}\right) ds
     \end{equation}
     is continuous. Therefore, $l$ is continuous over $\mc U$, and since it is integer valued, it is constant. Since
     \begin{equation}
         \phi(0,0) = \int_0^1 \Tr\left(\bo{\Psi}_s(0,0)\bo{\Sigma}\right) ds = 0,
     \end{equation}
     we get $l \equiv 0$.
\end{proof}
Before proving Lemma \ref{lemma:tilde_Phi_intro+pos_def}, we need the following lemma.
\begin{lemma} \label{lemma:Sigma_n-K_nK_n^top_positive}
    Let $ g_0 \in L^2([0,T], \mathbb{R}) $ and $ K $ be a Volterra kernel of continuous and bounded type in $L^2_{\mathbb R}$.
    Then 
    \begin{equation}
         \frac{T}{n}\Sigma_n - \nu^2 K_n K_n^\top \succeq 0,
    \end{equation}
    where the notation $A \succeq 0$ means that $A$ is positive semidefinite.
\end{lemma}

\begin{proof}
    For $0 \leq i \leq n$, set
    \begin{equation}
        X_i = X_{t_i}, \quad Y_i = W_{t_{i+1}} - W_{t_i}.
    \end{equation}
    We have, from Remark \ref{rmk:kappa=0}, that
    \begin{equation} \label{eq:X_i-g_0(t_i)}
        X_i - g_0(t_i) = \int_0^{t_i} K(t_i,s)\nu dW_s.
    \end{equation}
    Moreover, we have
    \begin{equation} \label{eq:Y_i}
        Y_i = \int_{t_i}^{t_{i+1}} dW_s.
    \end{equation}
    Set 
    \begin{equation} \label{eq:Z_n}
        Z_n := (X_0-g_0(t_0),\dots, X_{n-1}-g_0(t_{n-1}), Y_0, \dots, Y_{n-1})^{\top}.
    \end{equation}
    It directly follows from \eqref{eq:X_i-g_0(t_i)} and \eqref{eq:Y_i} that the covariance matrix of $Z_n$ is 
    \begin{align}
    \begin{pmatrix}
        \Sigma_n & \nu K_n \\
        \nu (K_n)^\top & \frac{T}{n} I_n
    \end{pmatrix}.
    \end{align}
    It follows from \citet[Theorem 1.12]{zhang2006schur} that
    \begin{equation}
         \frac{T}{n}\Sigma_n - \nu^2 K_n K_n^\top \succeq 0.
    \end{equation}
\end{proof}

\begin{proof}[Proof of Lemma \ref{lemma:tilde_Phi_intro+pos_def}]
    For the sake of readability, we omit any reference to $u$ and $w$.
    The fact that $\Re\left(\tilde{\Phi}_n\right)$ and $\Im\left(\tilde{\Phi}_n\right)$ are hermitian comes from the fact that for $A \in M_n(\mathbb C)$,
    \begin{equation}
        \Re(A) = \frac{1}{2}(A + A^*), \quad \Im(A) = \frac{1}{2i}(A - A^*).
    \end{equation}
    Concerning the positiveness of $\Re\left(\tilde{\Phi}_n\right)$, using that 
    \begin{align}
        \Phi_n 
        &= I_n - 2a\frac{T}{n}\tilde{\Sigma}_n \\
        &= I_n - 2a\frac{T}{n}\left(I_n - b K_n\right)^{-1}\Sigma_n \left(I_n - b K_n^\top\right)^{-1},
    \end{align}
    then 
    \begin{equation}
        \tilde{\Phi}_n = I_n - b \left(K_n + K_n^\top\right) + b^2 K_n K_n^\top -2a\frac{T}{n}\Sigma_n,
    \end{equation}
    and taking the real part yields
    \begin{equation}
        \Re\left(\tilde{\Phi}_n\right) = \left(I_n - \Re(b)K_n\right)\left(I_n - \Re(b)K_n^\top\right) - \Im(b)^2 K_n K_n^\top - 2\Re(a)\frac{T}{n}\Sigma_n.
    \end{equation}
    $\Re\left(\tilde{\Phi}_n\right)$ is decomposed in three parts, the first one being positive definite since $K_n$ is strictly triangular. We prove that the sum of the two last parts is positive. To do so, recall that (we supposed $\kappa=0$, recalling Remark \ref{rmk:kappa=0})
    \begin{equation}
        a = w + \frac{1}{2}(u^2 - u), \quad b = \rho \nu u,
    \end{equation}
    so that
    \begin{align}
        - \Im(b)^2 K_n K_n^\top - 2\Re(a)\frac{T}{n}\Sigma_n 
        &= \Im(u)^2 \left(\frac{T}{n}\Sigma_n - \rho^2 \nu^2 K_n K_n^\top\right) - 2\left(\Re(w) - \frac{1}{2}\Re(u)(1-\Re(u))\right) \Sigma_n \\
        & \geq \Im(u)^2 \left(\frac{T}{n}\Sigma_n - \nu^2 K_n K_n^\top\right) - 2\left(\Re(w) - \frac{1}{2}\Re(u)(1-\Re(u))\right) \Sigma_n,
    \end{align}
    where the inequality comes from the fact that $|\rho| \leq 1$. From Lemma \ref{lemma:Sigma_n-K_nK_n^top_positive} and the fact that $0 \leq \Re(u) \leq 1$ and $\Re(w) \leq 0$, we obtain that above is positive.
\end{proof}
Before proving Proposition \ref{prop:det_sqrt_formula}, we need the following lemma.
\begin{lemma} \label{lemma:sum_arg=im_phi}
    Let $ g_0 \in L^2([0,T], \mathbb{R}) $ and $ K $ be a Volterra kernel of continuous and bounded type in $L^2_{\mathbb R}$. Let $(u, w) \in \mc U$.
    Then 
    \begin{equation}
        \sum_{k=1}^n \arg(1 + i\lambda_k(u,w)) = -2\Im(\tilde{\phi}_n(u,w)).
    \end{equation}
\end{lemma}

\begin{proof}
    We have, as in the proof of Lemma \ref{lemma:two_integral_formulas_for_phi} (in finite dimension here), that 
    \begin{equation}
        \det\left(\Phi_n(u,w)\right) = e^{-2\tilde{\phi}_n(u,w)},
    \end{equation}
    so that $-2\Im(\tilde{\phi}_n(u,w))$ is an argument of $\det(\Phi_n(u,w))$. Also, from \eqref{eq:det_Phi_n_decomposed}, $\sum_{k=1}^n \arg(1 + i\lambda_k(u,w))$ is also an argument of $\det(\Phi_n(u,w))$. Therefore, there exists $l(u,w) \in \mathbb Z$ such that
    \begin{equation} \label{eq:sum_arg_vs_im_phi}
        \sum_{k=1}^n \arg(1 + i\lambda_k(u,w)) + 2\Im(\tilde{\phi}_n(u,w)) = 2i\pi l(u,w).
    \end{equation}
    The same reasoning than in the proof of Lemma \ref{lemma:phi_det_C^0} (in finite dimension again) shows that 
    \begin{equation}
        (u,w) \in \mc U \mapsto \sum_{k=1}^n \arg(1 + i\lambda_k(u,w)) + 2\Im(\tilde{\phi}_n(u,w))
    \end{equation}
    is continuous. Therefore, $l$ is continuous over $\mc U$, and since it is integer valued, it is constant. Since for $0 \leq s \leq 1$,
    \begin{equation}
        \Phi_n(0,0) = I_n, \quad \Psi_{n,s}(0,0) = 0_{M_n(\mathbb C)},
    \end{equation}
    then 
    \begin{equation}
        l(0,0) = 0,
    \end{equation}
    which concludes.
\end{proof}

\begin{proof}[Proof of Proposition \ref{prop:det_sqrt_formula}]
    For the sake of readability, we omit any reference to $u$ and $w$. We show that 
    \begin{equation} \label{eq:link_det_sqrt_&_sqrt_det}
        \sqrt{\det(\Phi_n)} = e^{i\pi \tilde{k}_n} \det\left(\sqrt{\Phi_n}\right).
    \end{equation}
    By \eqref{eq:det_Phi_n_decomposed}, and the fact that for $(a,z) \in \mathbb R^+ \times \mathbb C$, $\sqrt{az} = \sqrt{a} \sqrt{z}$, we have
    \begin{equation}
        \sqrt{\det\left(\Phi_n\right)} 
        = \sqrt{\det\left(\Re\left(\tilde{\Phi}_n\right)\right)} \sqrt{\prod_{k=1}^n \left(1 + i \lambda_k\right)}
    \end{equation}
    Since 
    \begin{equation}
        \sqrt{\prod_{k=1}^n \left(1 + i \lambda_k\right)} = e^{i \pi m_n} \prod_{k=1}^n \sqrt{1 + i \lambda_k},
    \end{equation}
    with $m_n := \frac{1}{2\pi}\left(\arg(\prod_{k=1}^n \left(1 + i \lambda_k\right)) - \sum_{k=1}^n \arg(1 + i\lambda_k)\right)$, we then have 
    \begin{equation}
        \sqrt{\det(\Phi_n)} = e^{i\pi m_n} \det\left(\sqrt{\Phi_n}\right).
    \end{equation}
    Combined with \eqref{eq:link_det_sqrt_&_sqrt_det}, we have that proving that 
    \begin{equation}
        \tilde{k}_n = m_n
    \end{equation}
    will conclude. By definition of $\tilde{k}_n$ and $m_n$, showing that
    \begin{equation}
        \arg\left(\det\left(\Phi_n\right)\right) = \arg(\prod_{k=1}^n \left(1 + i \lambda_k\right)), \quad -2\Im(\tilde{\phi}_n) = \sum_{k=1}^n \arg(1 + i\lambda_k),
    \end{equation}
    is sufficient. 
    The first equality directly follows from \eqref{eq:det_Phi_n_decomposed}, while the second one is ensured by Lemma \ref{lemma:sum_arg=im_phi}.
\end{proof}

\bibliographystyle{plainnat} % or any other style you prefer
\bibliography{Reference} % this assumes your .bib file is named 'Reference.bib'

\end{document}